\documentclass[aoas]{imsart}
\usepackage[dvipsnames]{xcolor}
\RequirePackage{amsthm,amsmath,amsfonts,amssymb}
\RequirePackage[sort]{natbib}
\PassOptionsToPackage{hyphens}{url}\RequirePackage[colorlinks,citecolor=blue,urlcolor=blue]{hyperref}
\RequirePackage{graphicx}

\let\cite\citep
\usepackage{multirow}

\usepackage{caption}
\usepackage{float} 
\usepackage{rotating}
\usepackage{color}
\usepackage{graphicx,amssymb,amsfonts, amsmath}	
\usepackage{enumitem}
\usepackage{chngcntr}
\usepackage[most]{tcolorbox}
\usepackage{tcolorbox}
\usepackage[export]{adjustbox}
\usepackage{subfig}
\usepackage{booktabs}
\usepackage{mathtools}
\usepackage{ wasysym }
\usepackage{svg}
\usepackage{bm}
\usepackage{xspace} 
\usepackage{bbm}
\usepackage{array}
\usepackage[ruled, titlenumbered,noend]{algorithm2e}
\SetArgSty{textnormal} 
\usepackage{siunitx}
\usepackage{subfig}

\newcommand{\eqn}{\mbox{Eq.}}
\newcommand{\eqns}{\mbox{Eqs.}}

\newcommand{\degreeC}{$^{\circ}$C}

\newcommand{\perMolePerSecond}{\mbox{ M}^{-1} \mbox{ s}^{-1}}

\newcommand{\reversible}{detailed-balance\xspace}
\newcommand{\Reversible}{Detailed-balance\xspace}

\newcommand{\statebias}{s_b}
\newcommand{\model}{\mathcal{C}}

\newcommand{\reversiblemodel}{\mathcal{C}^R}
\newcommand{\hatreversiblemodel}{\hat{\mathcal{C}}^R}

\newcommand{\initialpi}{\mathbf{\pi}}

\newcommand{\reversibleCTMC}{$\reversiblemodel=(\statespace, \tmat, \initialpi_0, \statesfinal, \initialpi)$}

\newcommand{\reversibleCTMCviceversa}{$(\statespace, \tmat,  \initialpi_0, \statesfinal, \initialpi) = \reversiblemodel$}

\newcommand{\CTMCviceversa}{$(\statespace, \tmat, \initialpi_0, \statesfinal) = \model$}

\newcommand{\CTMC}{$\model=(\statespace, \tmat,  \initialpi_0, \statesfinal)$}
\newcommand{\hCTMC}{$\hat{\model}=(\hat{\statespace}, \hat{\tmat},  \hat{\initialpi}_0, \hatstatesfinal)$}

\newcommand{\hreversibleCTMC}{$\hatreversiblemodel=(\hat{\statespace},  \hat{\tmat}, \hat{\initialpi}_0, \hatstatesfinal, \hat{\initialpi})$}

\newcommand{\statesFinal}{\mathcal{S}_{\mathrm{target}}}

\newcommand{\statesfinal}{\statespace_{\mathrm{target}}}
\newcommand{\statesinitial}{\statespace_{\mathrm{init}}}

\newcommand{\hatstatesfinal}{\hat{\statespace}_{\mathrm{target}}}

\newcommand{\expectt}{\tau_{\initialpi_{0}}}
\newcommand{\expecttdelta}{\tau_{\initialpi_{0}}^\delta}
\newcommand{\expect}{\mathbb{E}}

\newcommand{\setofstrands}{\mathrm{\Psi}^{*}}

\newcommand{\dGBox}{\Delta G_{\rm box}^{\circ}}
\newcommand{\dGVol}{\Delta G_{\rm volume}^{\circ}}

\newcommand{\dG}{\Delta G}

\newcommand{\kUni}{k_{\mathrm{uni}}}
\newcommand{\kBi}{k_{\mathrm{bi}}}

\newcommand{\argmin}{\text{argmin }}
\definecolor{forestgreen(traditional)}{rgb}{0.0, 0.4, 0.13}

\newcolumntype{L}[1]{>{\raggedright\let\newline\\\arraybackslash\hspace{0pt}}m{#1}}
\newcolumntype{C}[1]{>{\centering\let\newline\\\arraybackslash\hspace{0pt}}m{#1}}
\newcolumntype{R}[1]{>{\raggedleft\let\newline\\\arraybackslash\hspace{0pt}}m{#1}}

\newcommand{\pmat}{\mathbf{P}\xspace}
\newcommand{\tmat}{\mathbf{K}\xspace}

\newcommand{\pathwaytime}{\kappa}

\newcommand{\kssa}{ \hat{k}_\text{SSA}} 
\newcommand{\kpathway}{ \hat{k}_\text{PE}} 
\newcommand{\statespace}{\mathcal{S}}

\newcommand{\state}{s}
\newcommand{\stateInit}{s_{\text{0}}}
\newcommand{\stateFinal}{s_{\text{f}}}

\newcommand{\staterd}{s}

\newcommand*{\logten}{\mathop{\log_{10}}}

\newcommand\sh[1]{\textcolor{black}{#1}}
\startlocaldefs
\theoremstyle{plain}
\newtheorem{theorem}{Theorem}[section]

\newtheorem{proposition}[theorem]{Proposition}

\theoremstyle{definition}

\theoremstyle{remark}

\endlocaldefs

\begin{document}

\begin{frontmatter}
\title{The Pathway Elaboration Method for Mean First Passage Time Estimation in  Large  Continuous-Time Markov Chains with Applications to Nucleic Acid Kinetics}
\runtitle{The Pathway Elaboration Method}
\thankstext{T1}{Joint first authorship for first two authors: Sedigheh Zolaktaf and Frits Dannenberg contributed equally to this work. }

\begin{aug}
\author[A]{\fnms{Sedigheh Zolaktaf} \snm{}\ead[label=e1,mark]{nasimzf@cs.ubc.ca}},
\author[B]{\fnms{Frits Dannenberg} \snm{}\ead[label=e2,mark]{fdannenberg@live.nl}},
\author[A,C]{\fnms{Mark Schmidt} \snm{}\ead[label=e3,mark]{schmidtm@cs.ubc.ca}},
\author[A]{\fnms{Anne Condon} \snm{}\ead[label=e4,mark]{condon@cs.ubc.ca}},
\and
\author[B]{\fnms{Erik Winfree} \snm{}\ead[label=e5,mark]{winfree@caltech.edu}}

\address[A]{University of British Columbia, \printead{e1,e3,e4}}
\address[B]{California Institute of Technology, \printead{e2,e5}}
\address[C]{Alberta Machine Intelligence Institute}
\end{aug}

\begin{abstract}
For predicting the kinetics  of nucleic acid reactions,   continuous-time Markov chains (CTMCs) are widely used. The rate of a reaction can be obtained through the mean first passage time (MFPT)  of its CTMC. However, a typical issue in CTMCs is that the number of states could be large, making MFPT estimation challenging, particularly for \sh{events that happen on a  long time scale (rare events).} We propose the pathway elaboration method, a time-efficient probabilistic truncation-based approach for detailed-balance CTMCs. It can be used for estimating the MFPT for rare events in addition to rapidly evaluating perturbed parameters without expensive recomputations. We demonstrate that pathway elaboration is suitable for predicting nucleic acid  kinetics by conducting computational experiments on 267 measurements that cover a wide range of rates for different types of reactions.   We utilize pathway elaboration to gain insight on the kinetics of two contrasting   reactions, one being a rare event.  We then compare the performance of pathway elaboration with the stochastic simulation algorithm (SSA) for MFPT estimation on 237 of the reactions for which SSA is feasible. We further  build truncated CTMCs with SSA and transition path sampling (TPS) to compare with pathway elaboration.  Finally, we use pathway elaboration to rapidly evaluate perturbed model parameters during optimization with respect to experimentally measured rates for these 237 reactions. The testing error on the remaining 30 reactions, which involved rare events and were not feasible to simulate with SSA, improved comparably with the training error.    Our framework and dataset are available at \url{https://github.com/DNA-and-Natural-Algorithms-Group/PathwayElaboration}.

\end{abstract}

\begin{keyword}
\kwd{continuous-time Markov chain} 
\kwd{mean first passage time}
\kwd{nucleic acid kinetics}
\end{keyword}

\end{frontmatter}


\section{Introduction}
Predicting the kinetics  of reactions involving  interacting nucleic acid strands is desirable for  building autonomous nanoscale devices whose nucleic acid sequences and experimental conditions need to be carefully designed to control their behaviour, such as RNA toehold switches~\cite{angenent2020deep} and oscillators~\cite{srinivas2017enzyme}. By kinetics, we mean non-equilibrium dynamics, such as the
rate of a reaction and the order in which different strands interact  when a system is not in thermodynamic equilibrium.  Accurate and efficient prediction methods would  facilitate the design of complex  molecular devices by reducing, though not  eliminating, the need for debugging deficiencies with wet-lab experiments. 

The kinetics of nucleic acid reactions are often modeled as continuous-time Markov chains (CTMC)  with elementary steps~\cite{schaeffer2015stochastic,flamm2000rna,dykeman2015implementation}.   A CTMC  is a stochastic process on a discrete set of states that have the Markov property so that future possible states are independent of past states given the current state.  The time in a state before transitioning to another state, the holding time, is continuous; to retain the Markov property, holding times follow an exponential distribution with a single rate parameter for each state-to-state transition. 
In elementary step models of nucleic acid kinetics, states  correspond to   secondary structures  and a transition between two states corresponds to the breaking or forming of a base pair. The transition rates are specified with kinetic  models~\cite{metropolis1953equation}  along with  thermodynamic models~\cite{hofacker2003vienna,zadeh2011nupack}.
 We call the states corresponding to the reactants and the products  of a reaction as initial states and target states, respectively.

A fundamental kinetic property of interest in a nucleic acid reaction is the reaction rate constant. We  can estimate the rate using  the    mean  first passage time (MFPT) to reach the set of  target states starting from the set of initial states in the CTMC~~\cite{schaeffer2013stochastic}.   The MFPT is  commonly used to estimate the rate of a process~\cite{schaeffer2013stochastic,reimann1999universal,singhal2004using}.  For a CTMC with  a  reasonable state space size,  a matrix equation can provide an exact  solution to the MFPT~\cite{suhov2008probability}. However, direct application of matrix methods is not feasible for CTMCs that have large state spaces.   In this work we are interested in two computational challenges when the state space  is too large to allow exact matrix methods. The first challenge is   efficiently estimating MFPTs for reactions that happen on long time scale (rare events) and  the second challenge is  efficiently recalculating MFPTs for  mildly perturbed model parameters. Next we describe these challenges and then we describe our contributions.

 For large state spaces  researchers may resort to stochastic simulations~\cite{gillespie2007stochastic,ripley2009stochastic,asmussen2007stochastic,doob1942topics,gillespie1977exact}. 
   For CTMCs in particular, the stochastic simulation algorithm (SSA)~\cite{doob1942topics,gillespie1977exact}  is a widely used  Monte Carlo procedure that can  numerically generate statistically correct trajectories.   By trajectory, we mean a path from one state to another state, plus a holding time at each state along the path.   By sampling enough trajectories from the initial states to the target states, an estimate of the MFPT can be obtained.    However, SSA is inefficient for estimating MFPTs of rare events, that is reactions that happen on a  long time scales, such as reactions that involve high-energy barrier states.  A number of techniques have been developed for efficient sampling of simulation trajectories relevant to the event of interest~\cite{bolhuis2002transition,allen2009forward,rubino2009rare}.  Such sampling techniques, unfortunately, must generally be re-run if model parameters change, which makes it costly to perform parameter scans or to optimize a model.

Alternatively, large CTMCs can be approximated by models with just a subset of   ``most relevant''  states, in order to estimate MFPTs and other properties of interest with matrix methods~\cite{munsky2006finite,kuntz2019exit,singhal2004using}.  In truncated-based models, a set of most relevant states is selected, and transitions are added between selected states that are adjacent in the original CTMC. Methods have also been developed for reusing truncated models under mildly perturbed conditions such as temperature~\cite{singhal2004using}. The key challenge here is to efficiently enumerate a suitable subset of states that are  sufficient for accurate  estimation and also few enough that the matrix methods are tractable. As we discuss further in our related work section, the methods to date fail to do this for applications such as MPFT estimation of nucleic acid reactions that are rare events.


Here we are interested in  a method that successfully addresses both  challenges for  MFPT estimation in large CTMCs:  rare events and efficient recomputation for perturbed model parameters.     We develop a method  which uses both biased and local stochastic simulations to build truncated CTMCs relevant to a (possibly rare) event of interest.  With extensive experiments we demonstrate  that  our method  is  suitable for   predicting   nucleic   acid      kinetics modeled as CTMCs with elementary steps. 
Although in this work we only evaluate our method  in the context of nucleic acid kinetics, we believe it could also be useful in other applications of \reversible CTMCs,  such as chemical reaction networks~\cite{anderson2011continuous} and  protein folding~\cite{mcgibbon2015efficient}.  \\

\noindent \textbf{Our contributions.} In Section~\ref{section-method}, we propose the  \emph{pathway elaboration} method for estimating MFPTs in \reversible CTMCs. Pathway elaboration is a time-efficient probabilistic truncation-based approach which can be used for MFPT estimation of  rare events and also enables the rapid evaluation of  perturbed parameters. In pathway elaboration, we  first construct a pathway by  biasing SSA simulations from the initial states to the targets states. The biased simulations are guaranteed to reach the target states    in  expected time that is linear in the distance from initial to target states.   
Then,   we  expand the pathway by running  SSA simulations for a limited time from every state of the pathway, with the intention of increasing accuracy by increasing representation throughout the pathway.    Finally, we compute all possible transitions  between the sampled states  that were not encountered in the previous two steps. For the resulting truncated CTMC, we   solve a matrix equation to compute the MFPT to the target state (or states).   Since solving matrix equations could be slow for large CTMCs, pathway elaboration includes a $\delta$-pruning step to efficiently prune CTMCs  while keeping MFPT estimates  within  predetermined upper bounds. In this way, solving the system for other parameter settings becomes faster.  
 Figure~\ref{logo} illustrates the pathway elaboration method and its applications.

\begin{figure}[t]
\centering
\subfloat[Pathway construction]{\label{logo1} \includegraphics[width=0.19\textwidth]{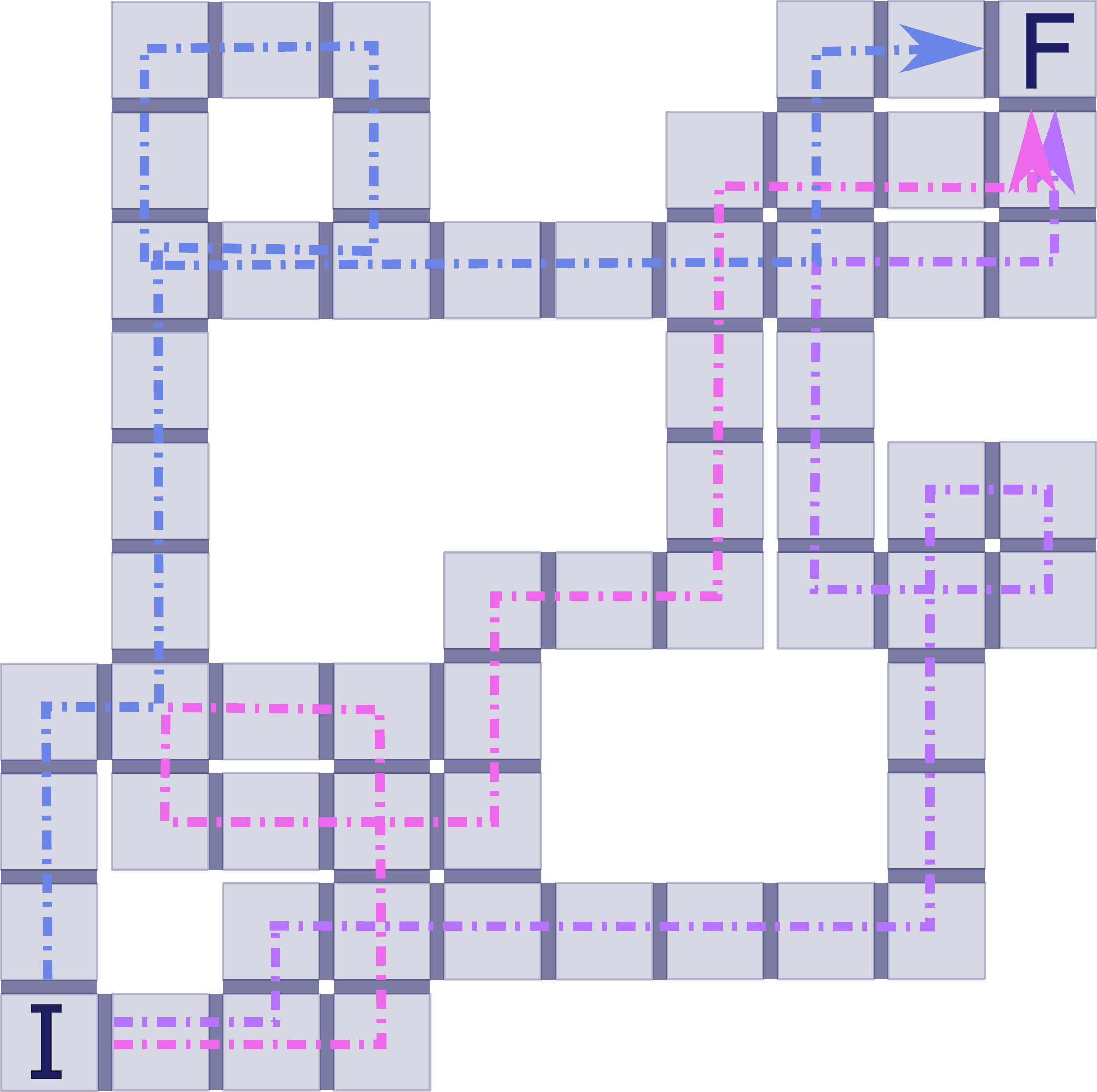}}
\subfloat[State elaboration]{\label{logo2} \includegraphics[width=0.19\textwidth]{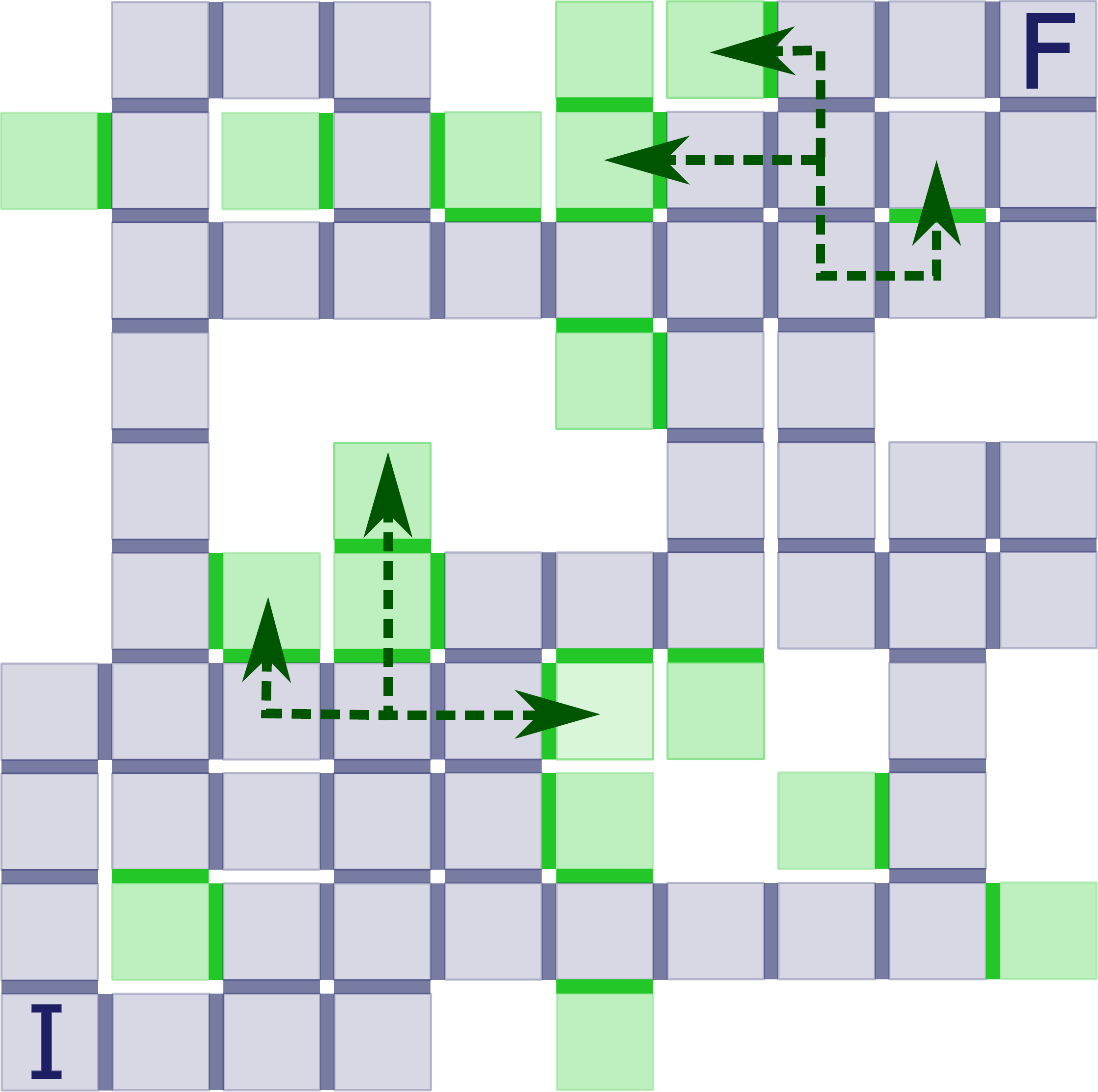}}
\subfloat[Transition construction]{\label{logo3} \includegraphics[width=0.19\textwidth]{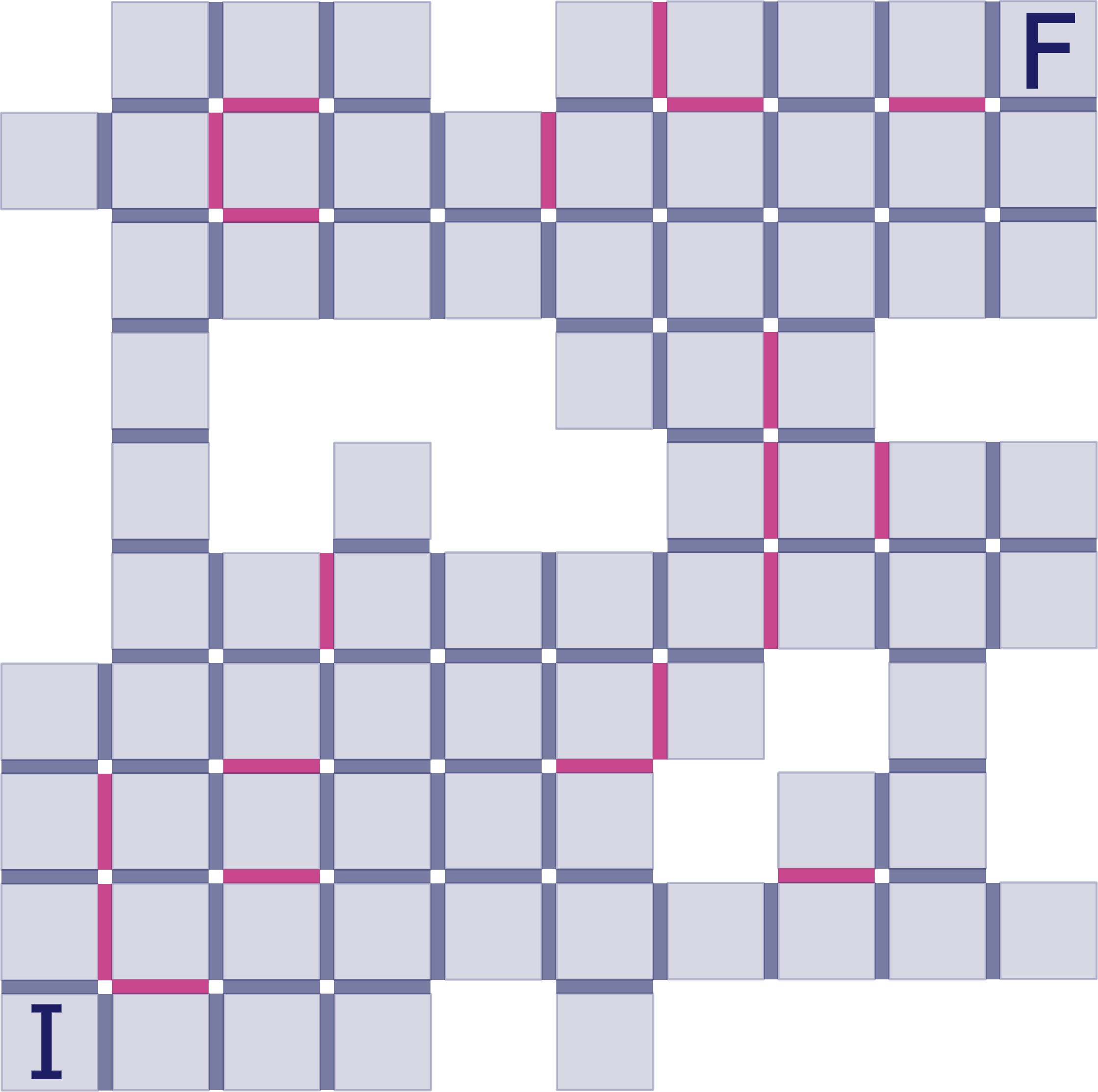}} 
\subfloat[$\delta$-pruning]{\label{logo4}
\includegraphics[width=0.19\textwidth]{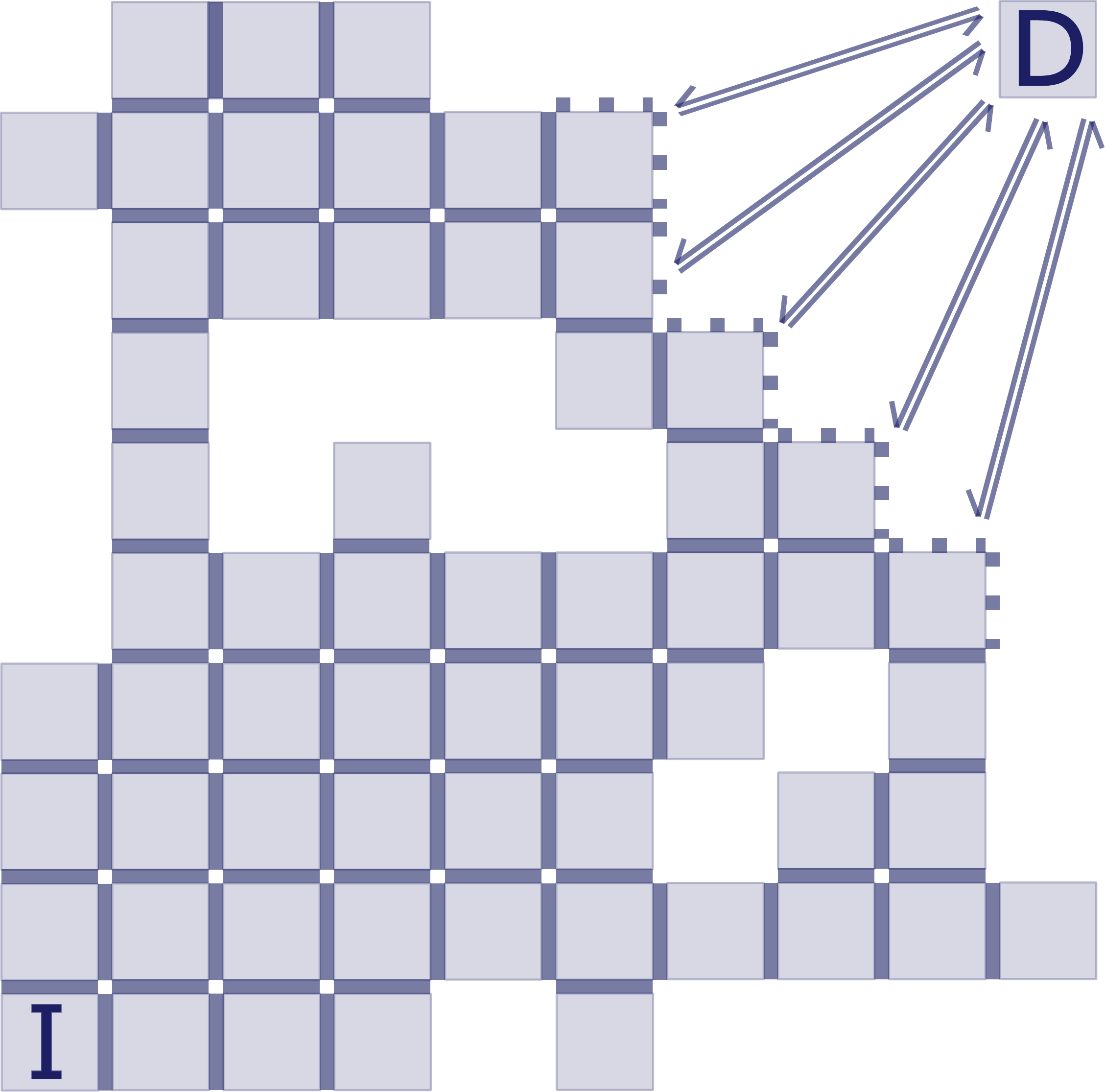}}
\\\subfloat[Updating perturbed parameters]{\label{logo5}   \includegraphics[width=0.19\textwidth]{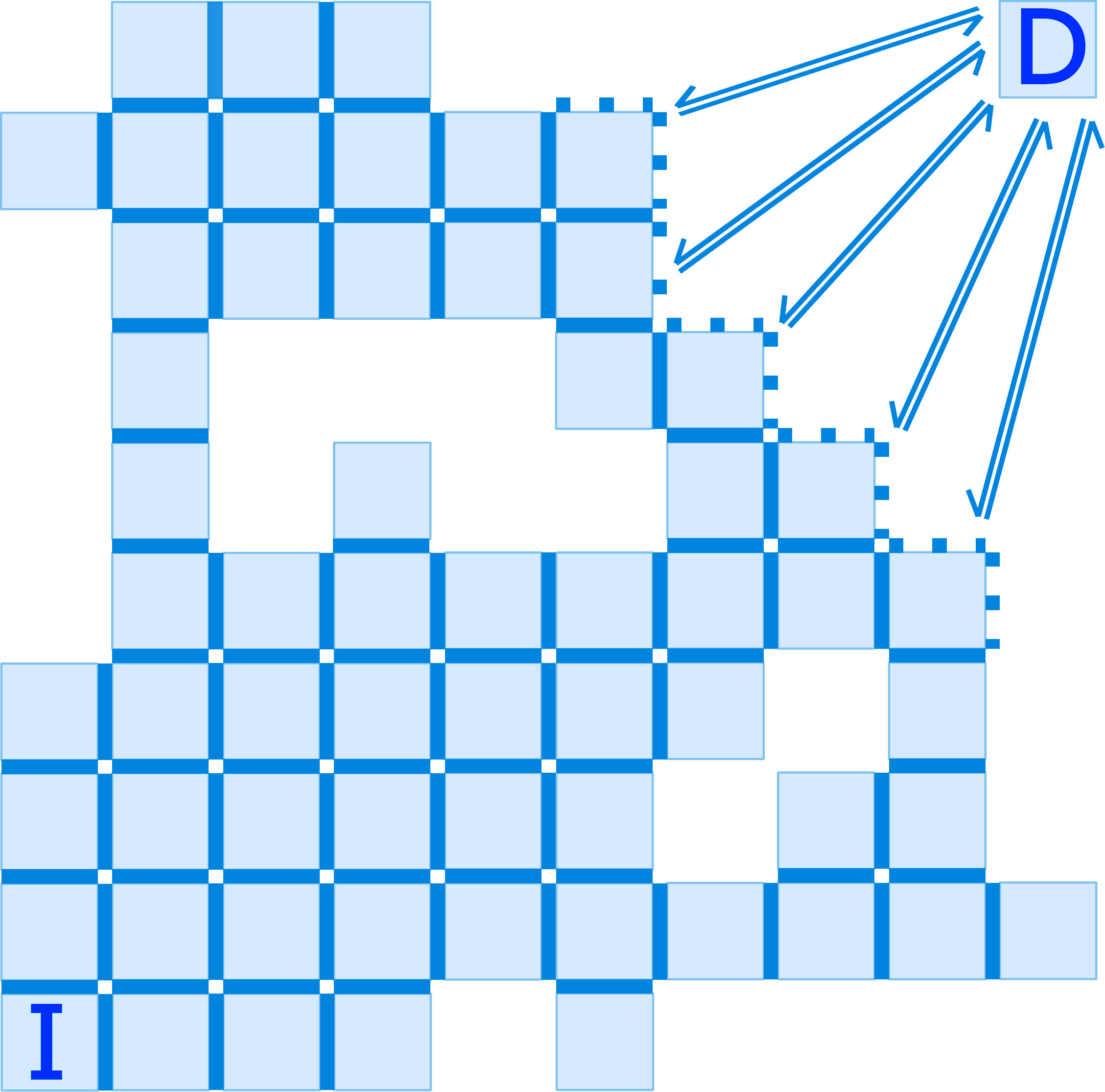}}\quad 
\subfloat[Parameter estimation]{\label{logo6} \includegraphics[width=0.24\textwidth]{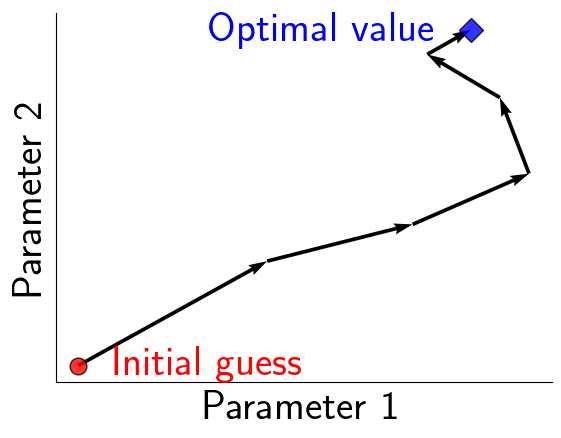}} \quad 
\subfloat[Obtain functionality]{ \label{logo7}\includegraphics[width=0.24\textwidth]{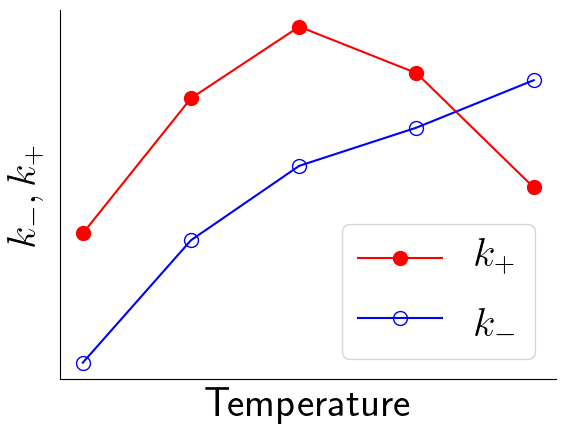}}
 \caption[]{The pathway elaboration method  and its applications. Pathway elaboration makes possible MFPT estimation of rare events and the rapid evaluation of perturbed parameters. Here, in the underling \reversible CTMC, boxes in a square grid represent states of the CTMC, with  transitions between adjacent boxes, initial state I at bottom left and target state F at top right.  (\textbf{a}) From state I, sample paths that are biased towards the target state F. Three sampled paths are shown with  blue, pink and purple dotted lines.  (\textbf{b}) From each sampled state found in the previous step,  run short unbiased simulations to fill in the neighborhood. Simulations from two states are shown with green dashed lines. The green states and transitions are sampled.  (\textbf{c}) Include all missing transitions between the states that were sampled in steps \textbf{a} and \textbf{b}. The red transitions are included. (\textbf{d}) Prune states that are expected to reach the target state quickly  by redirecting their transitions into a new state.  (\textbf{e}) For perturbed model parameters, keep the topology of the truncated CTMC, but update the   transition rates.  (\textbf{f})  We can use  truncated CTMCs for  perturbed  parameters, such as to estimate model parameters or     (\textbf{g}) to predict forward $(k_+)$ and reverse $(k_-)$ reaction rate constants as temperature  changes.  } 
\label{logo}
\end{figure}

To evaluate pathway elaboration, we focus on predicting the kinetics  of nucleic acids.    We implement the method using   the    Multistrand kinetic simulator~\cite{schaeffer2013stochastic,schaeffer2015stochastic}. Multistrand provides a secondary-structure level model of the folding kinetics of  multiple interacting nucleic acid strands, with thermodynamic energies consistent with NUPACK~\cite{zadeh2011nupack} and a stochastic simulation method based on SSA.  The challenges of large state spaces, rare events, and handling perturbed parameters all arise for nucleic acid kinetics. Since the number of  secondary structures may  be  exponentially large  in the length of the strands, applying matrix equations  is infeasible. Also,  SSA often takes a long time to complete for  rare nucleic acid reactions.  Moreover,  the rapid evaluation of mildly perturbed parameters is required, for example to calibrate the underlying kinetic  model or to obtain a desired functionality (see Figures~\ref{logo6} and~\ref{logo7}).     We conduct computational experiments on a  dataset of 267   nucleic acid kinetics~\cite{bonnet1998kinetics,cisse2012rule,hata2017influence,zhang2018predicting,machinek2014programmable} (described in Section~\ref{dataset}).  The dataset consists of various types of reactions, such as  helix association and toehold-mediated three-way strand displacement, for which experimentally measured reaction rate constants vary over 8.6 orders of magnitude. We  partition the 267 reactions into two sets,  237 where SSA is feasible for MPFT estimation, i.e., completes within two weeks,  and the remaining 30 for which SSA is not feasible.


In our  experiments, first in Section~\ref{exp-casestudy},  we conduct a case study and  use pathway elaboration to gain insight on the kinetics of two contrasting  reactions, one being a rare event. Then, in Section~\ref{exp-eval}, to evaluate the estimations of pathway elaboration,  {first}, we compare them with estimations obtained from SSA  for the 237 feasible reactions that were feasible  with SSA. We use SSA since obtaining MFPTs with matrix equations is not possible for many of these reactions and SSA  provides statistically correct trajectories.  We find that for the settings we use, the mean absolute error (MAE) of the $\log_{10}$ reaction rate constant (or equivalently the MAE of  the $\log_{10}$   MFPT) is $0.13$. This is a reasonable accuracy since the  $\log_{10}$  reaction constant predictions of  SSA  vary over $7.7$ orders of magnitude  (see Figure~\ref{kssavskpathwayfigure}). In our experiments, pathway elaboration is on average 5 times faster than SSA on these reactions.   Furthermore, to evaluate the estimations of pathway elaboration, we  build truncated CTMCs using simulations from SSA and transition path sampling (TPS)~\cite{bolhuis2002transition, singhal2004using,eidelson2012transition}.  In our experiments, the  MAE of pathway elaboration with SSA simulations compares well with the MAE of SSA-based truncated CTMCs with SSA simulations.  Moreover, the truncated CTMCs built with TPS have a larger MAE with SSA than pathway elaboration with SSA and the estimations have a larger variance.  Finally, in Section~\ref{section-parameterestimation}, we  use pathway elaboration to rapidly evaluate perturbed model parameters  during optimization of Multistrand kinetic parameters.   We use the same 237 reactions for training  the optimizer and the remaining 30 as our testing set. Using the optimized parameters, pathway elaboration estimates of reaction rate constants on our dataset are greatly improved over the estimates using non-optimized parameters. For the training set, the MAE of the $\log_{10}$ reaction rate constants of pathway elaboration with experimental measurements reduces from  $1.43$ to $0.46$, that is, a $26.9$-fold  error in the reaction rate constant reduces to a  $2.8$-fold error on average. The MAE over the 30 remaining reactions -- which involve rare events and have large state spaces -- reduces from  $1.13$ to  $0.64$, that is, a  $13.4$-fold  error  in the reaction rate constant reduces to  a  $4.3$-fold error on average.  On average for  these 30 reactions, pathway elaboration takes less than two days, whereas SSA  is  not feasible within two weeks. The entire optimization  and evaluation takes less than five days.

\begin{figure}[H]
    \centering
  \subfloat[Hairpin closing]{\label{hairpin-fig}  \includegraphics[width=0.18\textwidth,valign=m]{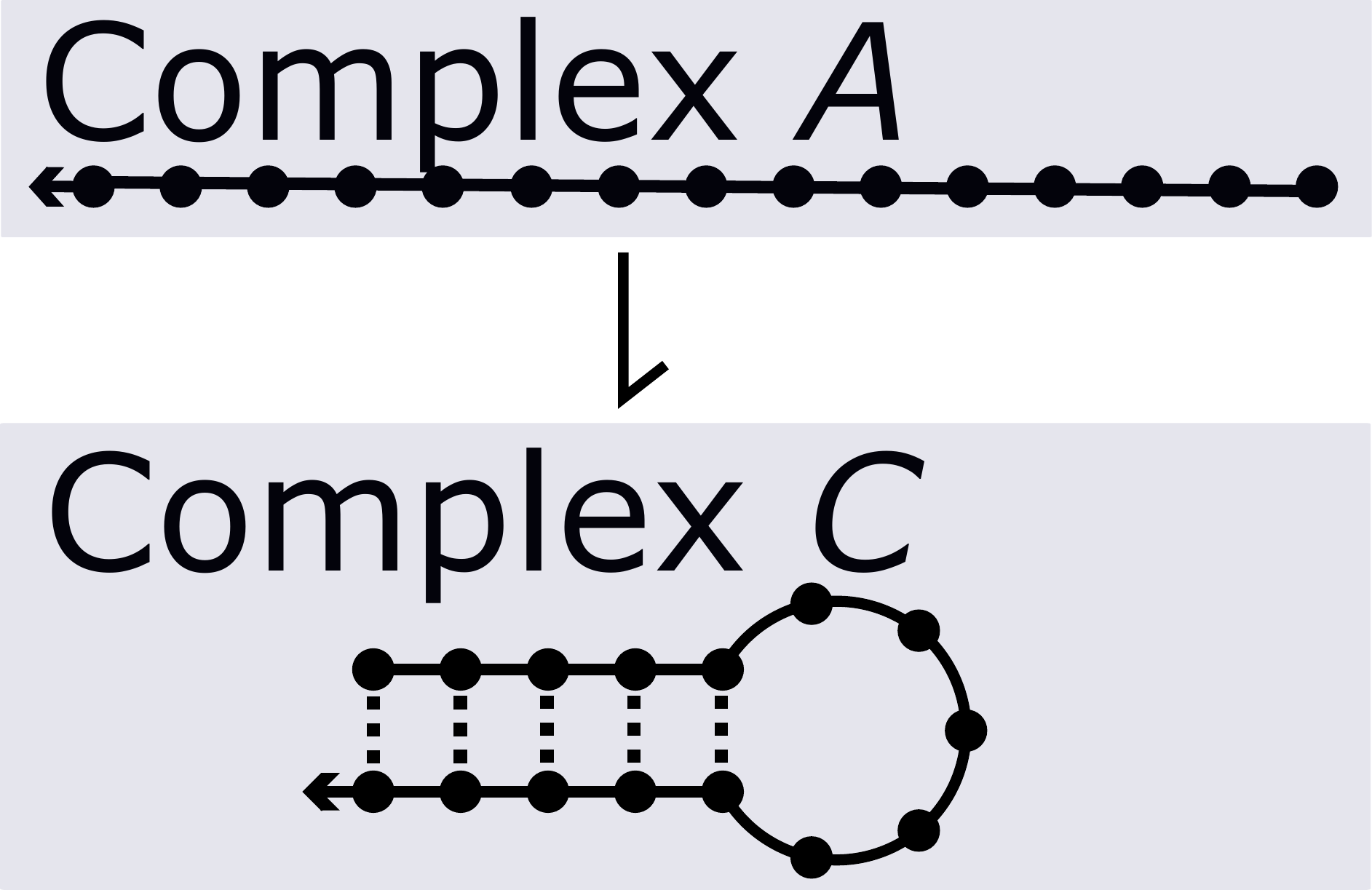}}\quad 
  \subfloat[Helix dissociation]{\label{helix-fig} 
   \includegraphics[width=0.19\textwidth,valign=m]{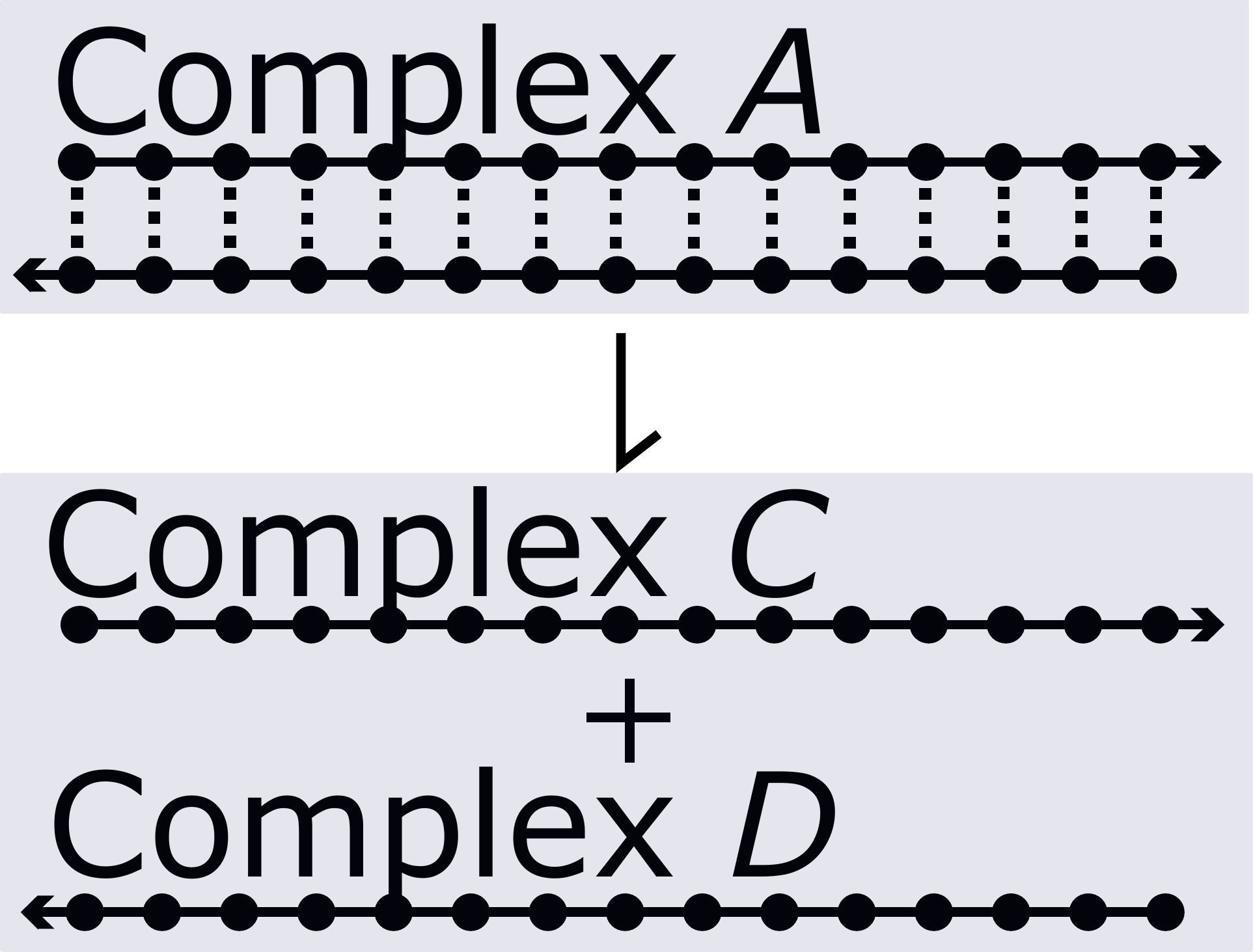}}\quad 
  \subfloat[Toehold-mediated three-way strand displacement]{\label{threewaystranddisplacement-fig}\includegraphics[width=0.52\textwidth,valign=m]{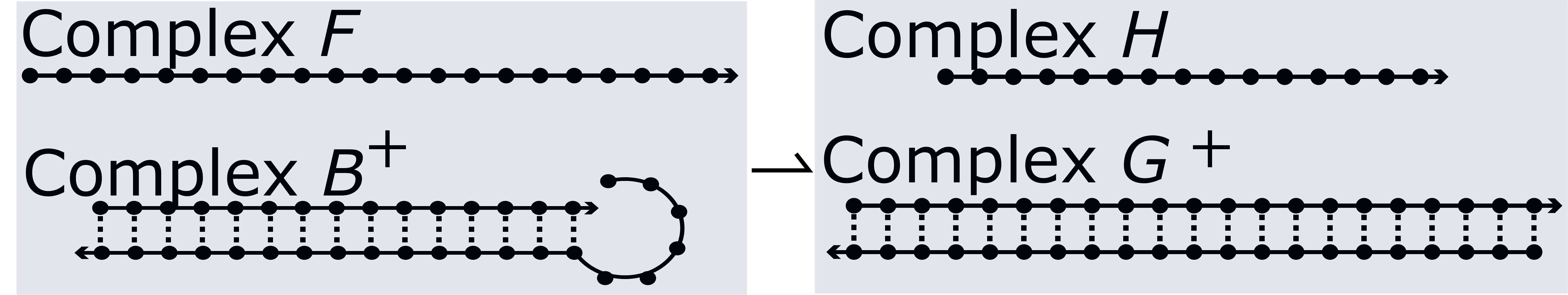}}
 \caption{  Examples of unimolecular and bimolecular interacting nucleic acid strand reactions. (\textbf{a}) Hairpin closing is a unimolecular reaction. It has  one reactant complex ($A$) and one product complex ($C$). The reverse reaction,  hairpin opening, is also  a unimolecular reaction. (\textbf{b}) Helix dissociation is a unimolecular reaction. It has one reactant complex ($A$)  and two product complexes ($C$ and $D$). The reverse reaction, helix association, is  a bimolecular reaction. (\textbf{c}) Toehold-mediated three-way strand displacement is  a bimolecular reaction. It has two reactant complexes ($B$ and $F$)  and two product complexes ($G$ and $H$). }
   \label{reactionexamples} 
\end{figure}

\section{Related Work} \label{relatedwork-section}
 There  exist  numerous Monte Carlo   techniques~\cite{rubino2009rare} 
 for driving simulations towards the target states or to reduce the variance of estimators. For example, importance sampling techniques~\cite{hajiaghayi2014efficient,doucet2009tutorial} use an auxiliary sampler to bias simulations, after which estimates are corrected  with  importance weights.     
  Moreover, many  accelerated  variants of SSA have been developed  for CTMC models of chemically reacting systems ~\cite{gillespie2007stochastic,cao2007adaptive,gillespie2001approximate}, which can be  adapted to simulate arbitrary CTMCs. 
  There also exists  a proliferation of rare event simulation methods  for molecular dynamics~\cite{zuckerman2017weighted,allen2009forward,bolhuis2002transition}. The ideas behind these methods can  be adapted for CTMCs  and can  be used along with SSA for more efficient computations.  
 For example,  in transition path sampling (TPS)~\cite{bolhuis2002transition}  an ensemble of paths are generated using a Monte Carlo procedure. First, a single path is generated that connects the initial and target states. New paths are then generated by picking random  states along the current paths and running time-limited simulations from the states. Sampled states along paths that do not reach the initial or target states are rejected. Even though we could use TPS  along with SSA to simulate rare events for CTMCs~\cite{eidelson2012transition}, it is likely that many of the simulated paths require a long  simulation time. For example, if the energy landscape has more than one  local maximum between the initial and target states, then paths simulated from in between these local maxima could require a long simulation time to reach either the initial or the target states. Moreover,   the simulated paths  could be correlated and depend on the initial path, and therefore the estimations of different runs could have a high variance. The correlation of paths could be reduced by retaining a fraction of the paths but it would also reduce the computational efficiency.
  
  Stochastic simulations are usually not  reusable  for the rapid evaluation of perturbed parameters and have to be adapted.  This is because the holding times of simulated trajectories need to be updated, which requires that information about all transitions from each sampled state is also stored. With the pathway elaboration method we can rapidly evaluate perturbed parameters by updating transition rates of the truncated CTMC.  Stochastic simulation  methods have been to some extent adapted for the rapid evaluation of perturbed parameters.  SSA has  been adapted in the fixed path ensemble inference (FPEI)  approach~\cite{zolaktaf2019efficient} for parameter estimation.  In this approach,   an  ensemble of paths  are generated using SSA and are then compacted and reused for mildly perturbed parameters. To estimate MFPTs, a Monte Carlo approach is used based on expected holding times of states. Despite being useful  for parameter estimation in general, this method  is not suitable for rare events, because  the paths are generated  according to SSA. In Section~\ref{exp-othermethods}, we use SSA to build truncated CTMCs for MFPT estimation and we compare its performance with pathway elaboration.

An alternative to sampling methods is to develop a smaller CTMC, whose MFPT well approximates that of the original large CTMC model. As is the case with sampling methods, techniques that have been developed to approximate the continuous state spaces of molecular dynamics simulations can be adapted for this purpose. In the context of predicting protein folding kinetics, the collection of paths produced by TPS has been used to build a so-called Markovian state model (MSM)~\cite{singhal2004using}. The MSM is the CTMC obtained by including all states and transitions along the sampled paths; since each state appears once, the MSM is more compact then the underlying set of paths.
The MSM approach can easily be adapted to build approximations to large CTMC models, for the purpose of estimating MFPTs and other properties of the CTMC.  The resulting MSM is a truncated CTMC. That is, it contains a subset of the states of the original CTMC, with transitions between states that are adjacent in the original CTMC.  Our pathway elaboration approach is similar to this approach in that sampled paths are used to build a truncated CTMC.  However, in pathway elaboration, we have a state elaboration step that helps model low-energy basins that might have a big influence on MFPTs in a complex landscape with multiple barriers and/or deep basins. Since the CTMCs in our application of predicting nucleic acid kinetics satisfy detailed balance, the initial and final states are still reachable from all states in the CTMC and so the MFPT is still finite when these additional states are included.
  In Section~\ref{exp-othermethods}, we use TPS  to build truncated CTMCs for MFPT estimation and we compare its performance with pathway elaboration.

Another important problem in CTMCs is computing  transient probabilities, that is the probability distribution of the states over time.   Transient probabilities can be computed exactly  with the master equation~\cite{van1992stochastic} for CTMCs that have a  feasible state space size.   An important tool that has been developed  to quantify the error of transient probability estimations for truncated CTMCs is  the finite state projection  (FSP) method~\cite{munsky2006finite}.  The FSP method  tells us that as the size of the state space of the  truncated CTMC grows, the approximation monotonically improves and provides upper and lower bounds on the true transient probabilities. In Section~\ref{section-method}, we show how we can adapt the FSP method to quantify the error of MFPT estimates for truncated CTMCs. As the authors of the FSP method mention, there are many ways to  grow the state space, for example by iteratively adding states that are reachable from the already-included states within  a fixed number of steps.   There have been many attempts to  enumerate  a suitable set of states that provides   good approximations while being small enough that transient probabilities can be computed efficiently~\cite{dinh2016understanding}. In the Krylov-FSP-SSA approach~\cite{sidje2015solving} an SSA approach is used to drive the FSP and  adaptive Krylov methods are  used   to efficiently evaluate the matrix exponential for transient probability estimation.  In brief, the method starts from an initial state space and proceeds iteratively in three steps. First, it drops  states that have become improbable.  Second, it runs SSA  from each state of the  remaining state space to incorporate probable states. Third, it adds states that are reachable within a fixed number of steps.  Despite its great potential, this way of building the state space  may not be suitable for estimating MFPTs of rare events. The pathway elaboration method is similar to this approach in the sense that it uses SSA in the state elaboration step. However, the pathway elaboration method  uses biased simulations to reach target states efficiently.

The Krylov-FSP-SSA method has also been used to build truncated CTMCs for the purpose of optimizing parameter sets that are used for transient probability estimation~\cite{dinh2017application}.  Moreover, in related work~\cite{georgoulas2017unbiased}, an ensemble of truncated CTMCs is used  to obtain an unbiased estimator of transient probabilities, which are  further used for Bayesian inference.

A probabilistic roadmap is another type of graph-based model, related to our work~\cite{kavraki1996probabilistic,tang2005using}. States in a probabilistic roadmap can be selected by random sampling or according to relevant properties, such as having low free energy. Then edges are added to connect nearby (though not necessarily adjacent) states. However, there are some challenges with this method that make it unsuitable for our purposes. First, it is not clear that sampling methods based on state (as opposed to path) properties will include important states on the most likely folding trajectories from initial to target states. Another challenge is determining appropriate transition rates between states that are not adjacent in the CTMC model. Instead, we rely on biased sampling of paths  from initial to target states.

\renewcommand{\footnoterule}{
  \kern -3pt
  \hrule width 0.48\textwidth height 0.4pt
  \kern 2pt
}

\section{Background}
In this section, we first describe the  continuous-time Markov chain (CTMC) model to which our pathway elaboration method (Section~\ref{section-method}) applies and also provide related definitions. Then  we  explain how nucleic acid kinetics can be modeled using CTMCs with the Multistrand kinetic model~\cite{schaeffer2013stochastic,schaeffer2015stochastic} which is background for our experiments (Section~\ref{results}).

   \textbf{Continuous-time Markov chain (CTMC).}\label{ctmcDef}  We indicate a CTMC as a tuple $\model=(\statespace, \tmat,  \initialpi_0, \statesfinal$), where $\statespace$ is a countable set of states,  
$\tmat: \statespace \times \statespace \rightarrow \mathbb{R}_{\geq 0}$ is the rate matrix and $\tmat(s,s) = 0$ for $s\in\statespace$,  $\initialpi_{0} : \statespace \rightarrow  [0,1]$ is the initial state distribution in which 
$	
\sum_{\state \in \statespace}\initialpi_{0}(s)=1
$, and  $\statesfinal$ is the set of target states.   We define the set of initial states  as    $
 \statesinitial = \left\{ \state \in \statespace~ \middle|~ \initialpi_{0}(\state) \neq  0 \right\}.$    For CTMCs considered here, $\statesfinal \cap \statesinitial = \emptyset$.   A transition between states $s,s'\in \statespace$ can occur only if $\tmat(s,s')>0$. The probability of moving from state $s$ to state $s'$ is defined by the transition probability matrix ${\pmat}: \statespace \times \statespace \rightarrow  [0,1]$ where  \begin{equation}\label{pmateq}\pmat(s,s')  = \frac{\tmat(s,s')}{\mathbf{E}(s,s)}.\end{equation} Here  $\mathbf{E}: \statespace \times \statespace \rightarrow \mathbb{R}_{\geq 0}$ is a diagonal  matrix in which $\mathbf{E}(s,s) =\sum_{s'\in \statespace}\tmat(s,s')$ is the exit rate.  The time  spent in state $s$ before a transition is triggered  is  exponentially distributed with exit rate $\mathbf{E}(s,s)$.
The generating matrix $\mathbf{Q}: \statespace \times \statespace \rightarrow \mathbb{R}$ is  $\mathbf{Q}=\tmat-\mathbf{E}$.

\textbf{\Reversible CTMC.}
In a \reversible CTMC \reversibleCTMC,  also known as a reversible CTMC,   a probability distribution $\initialpi : \statespace \rightarrow  [0,1]$  over the states exists that satisfies the detailed balance condition   $
    \initialpi(s)\tmat(s,s')= \initialpi(s')\tmat(s',s)$ 
for all  $s, s' \in \statespace$.   The detailed balance condition is a sufficient condition for ensuring that $\initialpi$ is a stationary distribution ($\initialpi\pmat =  \initialpi$).   For a \reversible finite-state CTMC, $\initialpi$  is the unique stationary distribution of the chain and is also the unique equilibrium distribution~\cite{whitt2006continuous}.

\textbf{Boltzmann distribution.}  In many Markov models of physical systems, eventually  the population of states will  stabilize and reach a Boltzmann distribution~\cite{schaeffer2015stochastic,flamm2000rna,tang2010techniques} at equilibrium. With this distribution, the probability that a system is in a state $s$ is 
\begin{equation}
    \initialpi(s)= \frac{1}{Z}e^{-\frac{E(s)}{k_BT}},
\end{equation}
where   $E(s)$ is the energy of the system at state $s$, $T$ is the temperature, $k_B$ is the Boltzmann constant, and  $Z= \sum_{s \in \statespace} e^{-\frac{E(s)}{k_BT}}$ is the partition function. To ensure that   at equilibrium  states are Boltzmann distributed,  the detailed balance conditions are 
 \begin{equation}
     \frac{  \tmat (\state,\state')}{  \tmat (\state',\state)} = e^{-\frac{E(\state')  -E(\state)}{K_BT}}. 
   \label{detailedbalanceenergy}
 \end{equation}

 \textbf{Reversible transition.} In this work, a reversible transition between states $s$ and $s'$ means  $\tmat(s,s') >0$ if and only if $\tmat(s',s)>0$.

\textbf{Trajectories and paths.}\label{pathsOverCTMC}
A trajectory $(s_0, t_0), ( s_1, t_1, ), ...,  ( s_m, t_m, ) $  with $m$ transitions over a CTMC \CTMC\ is a  sequence of states $s_{i}$ and holding times $t_{i}$ 
for which $\tmat(s_{i},s_{i+1})>0$ and $t_{i} \in \mathbb{R}_{> 0} $ for $i \geq 0$.
  We define a path $s_0, s_1, ..., s_m$ with $m$ transitions over a  CTMC \CTMC\ as a sequence of states $s_{i}$ 
for which $\tmat(s_{i},s_{i+1})>0$.

 \textbf{The stochastic simulation algorithm (SSA).}   SSA~\cite{gillespie1977exact,doob1942topics} simulates statistically correct trajectories  over a CTMC  \CTMC. At state $s_i$, the probability of sampling   $s_{i+1}$ is $\pmat(s_i,s_{i+1})$. At a jump from state $s_i$, it samples the holding time $T_i$ from an exponential distribution with exit rate $\mathbf{E}(s,s) =\sum_{s'\in \statespace}\tmat(s,s')$.

\textbf{Mean first passage time (MFPT).}\label{firstpassagetime}
 In a CTMC   \CTMC, for  a state $s \in \statespace$ and a target state $\stateFinal \in \statesFinal$,  the MFPT $\tau_s$ is the expected time to first  reach  $\stateFinal$ starting from state $s$.  For  state $s$, the  MFPT from $s$ to $\stateFinal$ equals the expected holding time in state $s$ plus the MFPT to $\stateFinal$ from the next visited state~\cite{suhov2008probability}, so
  \begin{align} \label{eqq} 
  \tau_s=  \frac{1}{\mathbf{E}(s,s) }+ \sum_{s' \in \statespace } \frac{\tmat(s,s')}{\mathbf{E}(s,s)}\tau_{s'}.
 \end{align}
 Multiplying the equation by the exit rate  $\mathbf{E}(s,s)=\sum_{s'\in \statespace}\tmat(s,s')$ then yields
\begin{align}\label{fpt00}
\sum_{s' \in \statespace} \tmat(s,s') (\tau_{s'} - \tau_{s}) = -1   .
 \end{align}
  Now writing $\mathbf{t} : \statespace  \setminus \stateFinal \rightarrow \mathbb{R}_{\geq 0}$ to be the vector of MFPTs for each state, such that $\mathbf{t}[s] =\tau_{s}$, we find a matrix equation as \begin{equation}
  \tilde{\mathbf{Q}} \mathbf{t} = -\mathbf {1},
  \label{mfptequation}
  \end{equation} where $\tilde{\mathbf{Q}}$ is obtained from ${\mathbf{Q}}$ by eliminating the row and column corresponding to the target state, and $\mathbf{1}$ is a vector of ones.  If there exists a path from every state  to the final state $\stateFinal$, then $\tilde{\mathbf{Q}}$ is  a weakly chained diagonally dominant matrix   and   is non-singular~\cite{azimzadeh2016weakly}.    The MFPT from the initial states  to the target state $\stateFinal$ is found  as
\begin{align}\label{mfpt}
\expectt = \sum_{\state \in \statespace} \initialpi_{0}(\state) \tau_{s}.
\end{align} 

If instead of a single target state $\stateFinal$ we have a set of target states $\statesfinal$, then to compute the MFPT to $\statesfinal$ we  convert all target states into one state  $\stateFinal$ so that  $\statespace^*= \statespace \setminus \statesfinal \cup  \left\{\stateFinal\right\}$. For $s,s'\in  \statespace^* \setminus \left\{\stateFinal\right\}$, we update the rate matrix  $ \tmat^* : \statespace^*  \rightarrow \mathbb{R}_{\geq 0} $ by $\tmat^*(s,\stateFinal) = \sum_{s''\in\statesfinal} \tmat(s,s'') $, $\tmat^*(s,s') =\tmat(s,s')$,  and   $\tmat^*(\stateFinal,s)$  is not used in the computation of the MFPT (see \eqn~\ref{mfptequation}).

\textbf{Truncated  CTMC.} 
Let $\hat{\statespace} \subseteq \statespace$ be a subset of the states over the  CTMC \CTMC\  or \reversible CTMC \reversibleCTMC\ and let $\hatstatesfinal$ $\subseteq \hat{\statespace}$.  We construct the rate matrix  $\hat{\tmat}: \hat{\statespace} \times \hat{\statespace} \rightarrow \mathbb{R}_{\geq 0}$ as
\begin{equation}\hat{\tmat}(\state, \state') = \tmat(\state,\state').
\label{reconstructratematrix} \end{equation} We  construct the initial probability distribution  $\hat{\initialpi}_0 : \hat{\statespace} \rightarrow [0,1]$ as \begin{equation}
    \hat{\initialpi}_0(s) = \frac{\initialpi_{0}(s)}{\sum_{s\in \hat{\statespace} } \initialpi_{0} (s)}.
    \label{reconstructinitialdistribution}
\end{equation}
 We define the truncated  CTMC as \hCTMC\  and  \hreversibleCTMC\ for  $\model$ and  $\reversiblemodel$, respectively.   For a \reversible $\hatreversiblemodel$,  $\hat{\initialpi} : \hat{\statespace} \rightarrow [0,1]$ defined as \begin{equation}
    \hat{\initialpi}(s) = \frac{\initialpi(s)}{\sum_{s\in \hat{\statespace} } \initialpi (s)}, 
    \label{reconstructequilibrium}
\end{equation}
satisfies the detailed balance conditions in $\hatreversiblemodel$ and is   the unique equilibrium distribution of   $\hat{\statespace}$  in $\hatreversiblemodel$~\cite{whitt2006continuous}.

\subsection{The Multistrand  Kinetic Model of Interacting Nucleic Acid Strands}\label{preliminaries-multistrand}
Here we provide background for our experiments in Section~\ref{results}.   We describe  how the   Multistrand kinetic simulator \cite{schaeffer2013stochastic,schaeffer2015stochastic}   models the  kinetics  of multiple interacting nucleic acid  strands  as CTMCs and how it estimates  reaction rate constants from MFPT estimates for these reactions.

\textbf{Interacting nucleic acid strands (reactions)}. 
Following Multistrand~\cite{schaeffer2013stochastic,schaeffer2015stochastic}, we are interested in modeling the interactions of nucleic acid strands in a stochastic regime.  In this regime, we have a discrete number of nucleic acid strands (a set called $\setofstrands$) in a fixed volume $V$ (the ``box'') and under fixed conditions, such as the temperature $T$ and the concentration of Na$^+$ and Mg$^{2+}$ cations. This regime can be found in systems that have a small volume with a fixed count of each molecule, and  can also be applied to larger volumes when the system is well mixed. Moreover, it can  be used to derive reaction rate constants of reactions in a chemical reaction network that follows mass-action kinetics~\cite{schaeffer2013stochastic,schaeffer2015stochastic}.

Following Multistrand~\cite{schaeffer2013stochastic,schaeffer2015stochastic}, a complex is a subset of strands of $\setofstrands$ that are connected through base pairing (see Figure~\ref{reactionexamples}).   A complex microstate is the  complex base pairs, that is secondary structure. 
A system microstate is  a set of complex microstates, such that each strand $\psi \in \setofstrands$ is part of exactly one complex.
A unimolecular reaction with reaction rate constant $k_1$  (units s$^{-1}$) has the form
\begin{align}\label{reactionsuni}
 A 	\xrightarrow{k_{1}}  C + D,	
\end{align}
and a bimolecular reaction with reaction rate constant $k_2$ (units  M$^{-1}$s$^{-1}$) has the form
\begin{align}\label{reactionsbi}
 B+ F 	\xrightarrow{k_{2}}  G + H.
\end{align}
Each reactant and product is a complex; $A$, $B$, $C$ and $G$ are nonempty
but $D$ and $H$ may be empty complexes.  For example, hairpin closing (Figure~\ref{hairpin-fig}) is a
unimolecular reaction involving one strand, where complexes $A$ and $C$ are
comprised of this one strand, while $D$ is empty. Helix dissociation (Figure~\ref{helix-fig}) is an
example of a unimolecular reaction where complex $A$ has two strands while
$C$ and $D$ are each of one of these strands.  An example of a bimolecular
reaction with two reactants and two non-empty products is  toehold-mediated three-way strand
displacement (Figure~\ref{threewaystranddisplacement-fig}). We discuss these type of reactions further in Section~\ref{dataset}.  We are interested in computing the reaction rate constants of such reactions.

\textbf{The Multistrand kinetic model.}\label{multistrandmodel} 
  Multistrand~\cite{schaeffer2013stochastic,schaeffer2015stochastic} is a kinetic simulator  for analyzing the folding kinetics of multiple interacting nucleic acid strands. It can handle both a system of  DNA strands and a system of RNA strands\footnote{Currently, Multistrand does not handle a system of mixed DNA and RNA strands, though it can be extended to handle such systems using  good thermodynamic parameters.}. 
 The Multistrand kinetic model is a  \reversible CTMC   \reversibleCTMC\ for a set of interacting nucleic acid strands $\setofstrands$  in a  fixed volume   $V$ (the ``box'') and under fixed  conditions, such as the temperature $T$ and the concentration of Na$^+$ and Mg$^{2+}$ cations.    The  state space  $\statespace$  of the CTMC is the set of all non-pseudoknotted system microstates\footnote{ A pseudoknotted secondary structure  has at least two base pairs in which one nucleotide of a base pair is intercalated between the two nucleotides of the other base pair. A non-pseudoknotted system microstates does not contain any pseudoknotted secondary structures. Currently, Multistrand  excludes  pseudoknotted secondary structures due to computationally difficult energy model calculations.}   of the set $\setofstrands$ of interacting strands.    The transition rate $\tmat(s,s')$ is non-zero if and only if $s$ and $s'$ differ
by a single base pair\footnote{Multistrand allows Watson-Crick base pairs to form, that is A-T and G-C in  DNA  and A-U and G-C in RNA. Additionally, it provides an option to  allow G-T in DNA and G-U in RNA.}.   Multistrand distinguishes between unimolecular transitions, in which the number of strands in each complex remains constant, and bimolecular transitions where this is not the case.  There are bimolecular join moves, where two complexes merge, and bimolecular break moves, where a complex falls apart and releases two separate complexes.   

The transition rates in the Multistrand kinetic model obey detailed balance as
\begin{equation}
          \frac{  \tmat (\state,\state')}{  \tmat (\state',\state)} = e^{-\frac{\dGBox(\state')  -\dGBox(\state)}{RT}},
\end{equation}
where $\dGBox(\state)$ is the free energy of state $\state$  (units: $\mbox{kcal}\mbox{ mol}^{-1}$) and   depends on the temperature  $T$ (units: $K$) as  $\Delta G = \Delta H - T \Delta S$,  and  $R \approx  1.98 \times 10^{-3} \mbox{ kcal} \mbox{ K}^{-1} \mbox{ mol}^{-1}$ is the gas constant. The   enthalpy $\Delta H$ and entropy $\Delta S$ are fixed and calculated in the model using thermodynamic models that depend on  the concentration of Na$^+$ and Mg$^{2+}$ cations and also  on a volume-dependent entropy term.     The detailed balance condition  determines the ratio of rates for reversible transitions.  A standard kinetic model that is used  in Multistrand to determine the transition rates is the Metropolis kinetic model~\cite{metropolis1953equation},  where all energetically favourable transitions occur at the same fixed rate and energetically unfavourable transitions scale with the difference in free energy. Unimolecular transition rates are given as
\begin{equation}
    \tmat (\state,\state') = 
  \begin{cases}
    \kUni 											&   \text{if}  \dGBox(\state) <  \dGBox(\state'),	\\
    \kUni	e ^{  - \frac{ \dGBox(\state') - \dGBox(\state)}{ RT} }	       	& \text{ otherwise,}
  \end{cases} 
  \label{metropolis} 
  \end{equation} 
and bimolecular transition rates  are given as 
 \begin{equation}
 \tmat (\state,\state') = 
  \begin{cases}
    \kBi 	u										&  \text{join move,	}		\\
    \kBi   e^{ - \frac{ \dGBox(\state') - \dGBox(\state) + \dGVol }{ RT}   }\times \text{M} 		       	&\text{break move,}
  \end{cases}\label{metropolis}
 \end{equation}
  where $u$ is the concentration of the strands (units: M), $\dGVol=-RT \ln u$, $\kUni > 0 $ is the unimolecular rate constant  (units:  $\mbox{s}^{-1}$),  and $ \kBi > 0 $   is the bimolecular rate constant (units: $\mbox{M}^{-1} \mbox{ s}^{-1}$). The kinetic parameters $\theta = \{\kUni, \kBi \}$   are calibrated from experimental measurements~\cite{wetmur1968kinetics,morrison1993sensitive}.

 The  distribution $\initialpi_0$ is an initial distribution over the microstates of
the reactant complexes, and the set $\statesfinal$ is a subset of the
microstates of the product complexes, which we determine  based on the type of the reaction (see Section~\ref{dataset}).
   To set $\initialpi_0$ for unimolecular reactions,  we use  particular complex microstates.  One illustrative example is the
(unimolecular) hairpin closing reaction, where we set $\initialpi_0(h) = 1$
for the system microstate that has no base pairs and $\initialpi_0(s) = 0$ for all other structures, and
$\statesfinal$ is the system microstate where the 
 For a bimolecular reaction, when  the  bimolecular transitions are slow enough between the two complexes,  it is valid to  assume the  complexes each reach equilibrium before bimolecular transitions occur and therefore are Boltzmann distributed~\cite{schaeffer2013stochastic}. 
Let $\mathcal{CM}$ be the set of all possible complex microstates of a complex $B$ in a volume.
A distribution $\initialpi_{b}$ is Boltzmann distributed with respect to complex $B$ if and only if
\begin{align}
     \initialpi_{b}(c') = \frac{e^{-\dG(c') / RT } }{ \sum_{c \in \mathcal{CM}} e^{-\dG(c) / RT }}
\end{align}
for all complex microstates $c' \in \mathcal{CM}$.  In a bimolecular reaction of the form in $\eqn$~\ref{reactionsbi}, for a system microstate $s$  that has  complex microstates $c$ and $c'$ corresponding  to complexes $B$ and $F$, we define the initial distribution as  $ \initialpi_{0}(s) =  \initialpi_{b}(c) \times \initialpi_{b}(c')$. For all other states, we define $  \initialpi_{0}(s) = 0 $.

Following the conventions of Multistrand~\cite{schaeffer2013stochastic}, we estimate the reaction rate constant   for a  reaction  from its MFPT $\expectt$ (\eqn~\ref{mfpt}). For a reaction in the form of \eqn~\ref{reactionsuni}, 
\begin{equation} 
  k_1= \frac{1}{\expectt}.
\label{rate-mfpt-unimolecular}
\end{equation}
In the limit of low concentrations for a reaction in the form of  \eqn~\ref{reactionsbi}, \begin{equation}
    k_{2} = \frac{1}{u} \frac{1}{\expectt}.
    \label{rate-mfpt-bimolecular}\end{equation}

\section{The Pathway Elaboration Method}
\label{section-method}
 We are interested in  efficiently estimating MFPT  of rare events in \reversible CTMCs and also the rapid evaluation of mildly perturbed parameters.    Our approach is to create a reusable in-memory representation of CTMCs, which we call a truncated CTMC,  and to compute the MFPTs through  matrix equations (\eqns~\ref{mfptequation} and~\ref{mfpt}).

 We propose the  \emph{pathway elaboration} method for building a truncated  detailed-balance CTMC $\hatreversiblemodel$ for a  detailed-balance  CTMC $\reversiblemodel$.   We call this approach the pathway elaboration method as we build a truncated CTMC by  elaborating an  ensemble of  prominent  paths in the system.    The method has three main steps  to build a truncated CTMC, and an additional step for the rapid evaluation of perturbed parameters. 
 \begin{enumerate} \item  The ``pathway construction'' step uses biased simulations to find an ensemble of short paths from the initial states to the target states. This step is inspired by importance sampling~\cite{madras2002lectures,rubino2009rare,andrieu2003introduction,hajiaghayi2014efficient} and exploration-exploitation trade-offs~\cite{sutton2018reinforcement}. \item  The  ``state elaboration'' step uses SSA from every state in the pathway to add additional states to the pathway,  with the intention of increasing accuracy. This step is inspired by the string method~\cite{weinan2002string}. \item   The ``transition construction'' step creates a matrix of transitions between every pair of states  obtained from the first and second steps.  \item The  ``$\delta$-pruning'' step prunes the CTMC obtained from the previous steps to facilitate the rapid evaluation of perturbed parameters. 
 \end{enumerate} 
These steps result in a  truncated detailed-balance CTMC \text{\hreversibleCTMC}. Figure~\ref{logo}, parts (\textbf{a}) to (\textbf{d}), illustrates the key steps of the pathway elaboration method, and Algorithm~\ref{alg:pathwayelaboration} provides high-level pseudocode. We next describe these steps in  detail.\\
 \begin{algorithm}[t]
\DontPrintSemicolon
\SetKwFunction{Keys}{Keys} 
\SetKwFunction{genpath}{ConstructPathway($\model$,$N$,$\beta$,$\initialpi'$) }
\SetKwFunction{genpathreversible}{ConstructPathway($\reversiblemodel$,$N$,$\beta$,$\initialpi'$) }
\SetKwFunction{func}{PathwayElaboration($\reversiblemodel$,$N$,$\beta$,$K$,$\pathwaytime,\initialpi'$)}
\SetKwFunction{ssa}{ ElaborateState($s,\reversiblemodel,$K$,\pathwaytime$)}

\SetKwProg{myfunc}{Function}{}{} \myfunc{\func}
{  
\reversibleCTMCviceversa\;
${\statespace}_0  \gets$  \genpathreversible \; 
$\hat{\statespace} \gets \emptyset $ \;
\For{$s\in {\statespace}_0 $}
{

$\statespace' \gets \ssa$  \tcp*[l]{Run  SSA  $K$ times from $s$ with a time limit of $\kappa$ and return the visited states.} 
$\hat{\statespace} \gets  \hat{\statespace} \cup  \statespace' $
} 
$\hat{\tmat} \gets$  Construct rate matrix from  $\hat{\statespace}$  and $\tmat$    \tcp*[l]{$\eqn$~\ref{reconstructratematrix}.} 
\Return \hreversibleCTMC \; \tcp*[l]{For  $\hat{\initialpi}_0 $ and $\hat{\initialpi}$, see $\eqn$~\ref{reconstructinitialdistribution} and $\eqn$~\ref{reconstructequilibrium}, respectively.}
}   
       
\SetKwProg{myfunc}{Function}{}{}
  \myfunc{\genpath}{  
\CTMCviceversa\;
$\statespace_0 \gets \emptyset $ \; 
\For{n = 1 to N}{
$\text{Sample } s \sim  \initialpi_{0}$ \; 
$\statespace_0 \gets   \statespace_0 \cup \{s\}$ \;
$\text{Sample } \statebias \sim  \initialpi'$ \; 
\For{t =1,2, ... }{ 
\lIf{$s = \statebias$}{break} 
$\text{Sample }  z \sim \text{Uniform}(0,1)$\; 
\If(\tcp*[h]{Bias simulations towards $\statebias$ using $\eqn$~\ref{modifiedinitialpathway2}.}){$z <  \beta$} {Sample   $s'|s \sim  \pmat(\cdot|X_{t-1}=s) $} 
 \Else{Sample $s'|s \sim \Breve{\pmat}_{\statebias}(\cdot|X_{t-1}=s) $         } 
    $\statespace_0 \gets \statespace_0 \cup s'$ \; 
    $s \gets s'$\;
} 

}
\Return $\statespace_0$
}
        \caption{The pathway elaboration method.  }
\label{alg:pathwayelaboration}
\end{algorithm}

\noindent \textbf{Pathway construction.} We construct a pathway by biasing $N$ SSA simulations towards the target states.  We bias a simulation by  using   the  shortest-path distance function $d: \statespace \times \statesFinal  \rightarrow \mathbb{R}_{\geq 0}$  from every state $s \in \statespace$ to a fixed target state $\statebias \in \statesfinal$~\cite{kuehlmann1999probabilistic,hajiaghayi2014efficient}.     For every biased path, we can  use a different  $\statebias$. Therefore, in general, we can sample $\statebias$ from a probability distribution $\initialpi'$ over the target states. Given $\statebias$, we use an exploitation-exploration trade-off approach.  At each transition, the process randomly chooses to either decrease the distance  to $\statebias$ or to explore the region based on the actual probability  matrix of the transitions. 

Let $\mathcal{D}_{\statebias}(s)$ be  the set of all  neighbors of $s$ whose  distance with  $\statebias$ is  one less than the distance of $s$ with  $\statebias$, and let $\pmat(s,s’)$ be as in \eqn~\ref{pmateq}. Instead of sampling states according to $\pmat$, we use $\Tilde{\pmat}: \statespace \times \statespace \rightarrow \mathbb{R}_{\geq 0} $  where 

  \begin{equation}
\Tilde{\pmat}(s,s')  =  \begin{cases}
      \pmat(s,s') =  \frac{\tmat(s, s')}{\sum_{s'' \in \statespace} \tmat(s, s'') } &  0\leq  z  \leq   \beta,\\
    \Breve{\pmat}_{\statebias}(s,s') = \frac{\tmat(s, s')\bm{1}\{ s' \in \mathcal{D}_{\statebias}(s)\}}{\sum_{s'' \in \statespace} \tmat(s, s'')  \bm{1}\{ s'' \in \mathcal{D}_{\statebias}(s)\}}   &    \beta < z \leq 1.
    \end{cases} 
    \label{modifiedinitialpathway2}
\end{equation}
Here $z$ is chosen uniformly at random from $[0,1]$, $\beta$ is a threshold, and $\bm{1}\{.\}$ is an indicator function that is equal to 1 if the condition is met and 0 otherwise.  When $\beta=1 $, then $\Tilde{\pmat}(s,s') = \pmat(s,s') $.

\begin{proposition}
Let $d_\text{max}$ be the maximum distance  from a  state in a CTMC to target state $\statebias$. Then when $0 \leq \beta < 1/2$, the expected length of a pathway that is sampled according to \eqn~\ref{modifiedinitialpathway2} is at most $\frac{d_\text{max} }{1-2\beta}$.\label{proposition1}
\end{proposition}

\begin{proof}
 Based on  the  distance of states with  $\statebias$, we can project a biased path that is generated with \eqn~\ref{modifiedinitialpathway2} to a 1-dimensional random walk $R$, where coordinate $x=0$ corresponds to $\statebias$  and  coordinate $x >0 $ corresponds to  all states $s \neq \statebias$ with $d(s,\statebias)=x$.     From the definition of $\Tilde{\pmat}$  and  since  all states have a path to $\statebias$ by a transition to a neighbor state that decreases the distance by one,  at each step,  the random walk either takes one step closer to  $x=0$   with probability at least $1 - \beta$  or one step further from  $x=0$ with probability at most $\beta$.  If we let $E(R,k)$  denote the expected time for random walk $R$ to reach $0$ from $k$,  then we have that 
 when $0 \le \beta < 1/2$,
\begin{equation}
 E(R,k) \le  \frac{k}{1-2\beta},
\end{equation}
which follows from classical results on biased random walks---see Feller XIV.2 \cite{Feller-1968}.
Therefore, if $0 \leq \beta < 1/2$, the proposition  holds, and  the state space  built with $N$ biased paths from the initial state $\stateInit$ to a target state $\statebias$ has expected size
\begin{equation}\begin{split}
\expect[|\hat{\statespace}|]\leq \frac{N\cdot d(\stateInit, \statebias) }{1-2\beta} \leq \frac{N\cdot d_\text{max} }{1-2\beta}.
\end{split}
\label{eqsizebound}
\end{equation}
If for each biased path, the initial state is sampled from $\initialpi_0$ and the target state is sampled from $\initialpi'$, then we sum over the $N$ sampled (initial state, target state) pairs, and  the total expected state space size is still bounded by $\frac{N\cdot d_\text{max} }{1-2\beta}$.
\end{proof}

For efficient computations, we should compute the shortest-path distance efficiently.  For elementary step models of interacting nucleic acid strands,  we can compute $d(s,\statebias)$ by computing the minimum number of base pairs that need to be deleted or formed to convert $s$ to  $\statebias$.  Multistrand provides  a list of base pairings for every complex microstate in a  system microstate (state) and we can calculate the distance between  two states in a running time of O($b$), where $b$ is the number of bases in the strands.\\

\noindent \textbf{State elaboration.} By using \eqn~\ref{modifiedinitialpathway2}, a biased path could have a low probability of reaching a  state that has a high probability of being visited with SSA.  For example, in some helix association reactions~\cite{zhang2018predicting}, intra-strand base pairs are likely to form  before completing  hybridization. However, the corresponding states  do not lie on the shortest paths from the initial states to the target states.  
Let $c$ be  the minimum number of transitions from $\stateInit$ that are required to  reach   $\staterd$ but which  increase the distance to $\statebias$.   Let the random walk $R$   be defined as the previous step.
Let $P_1$  denote the  probability of reaching $\statebias$ before reaching $\staterd$ for this random walk.    Following classical results on biased random walks~\cite{Feller-1968}, for $\beta \neq 1/2$,  
$
 P_1 \geq  \frac{(\frac{\beta}{1-\beta})^{c} -1}{(\frac{\beta}{1-\beta })^{d_{\statebias}(\stateInit)+c} -1 }.
$   In the extreme case if $\beta  = 0$, then $P_1=1$ and the probability of reaching $\staterd$  will be 0. 

 \begin{figure}
\center
 \includegraphics[width=0.2\textwidth]{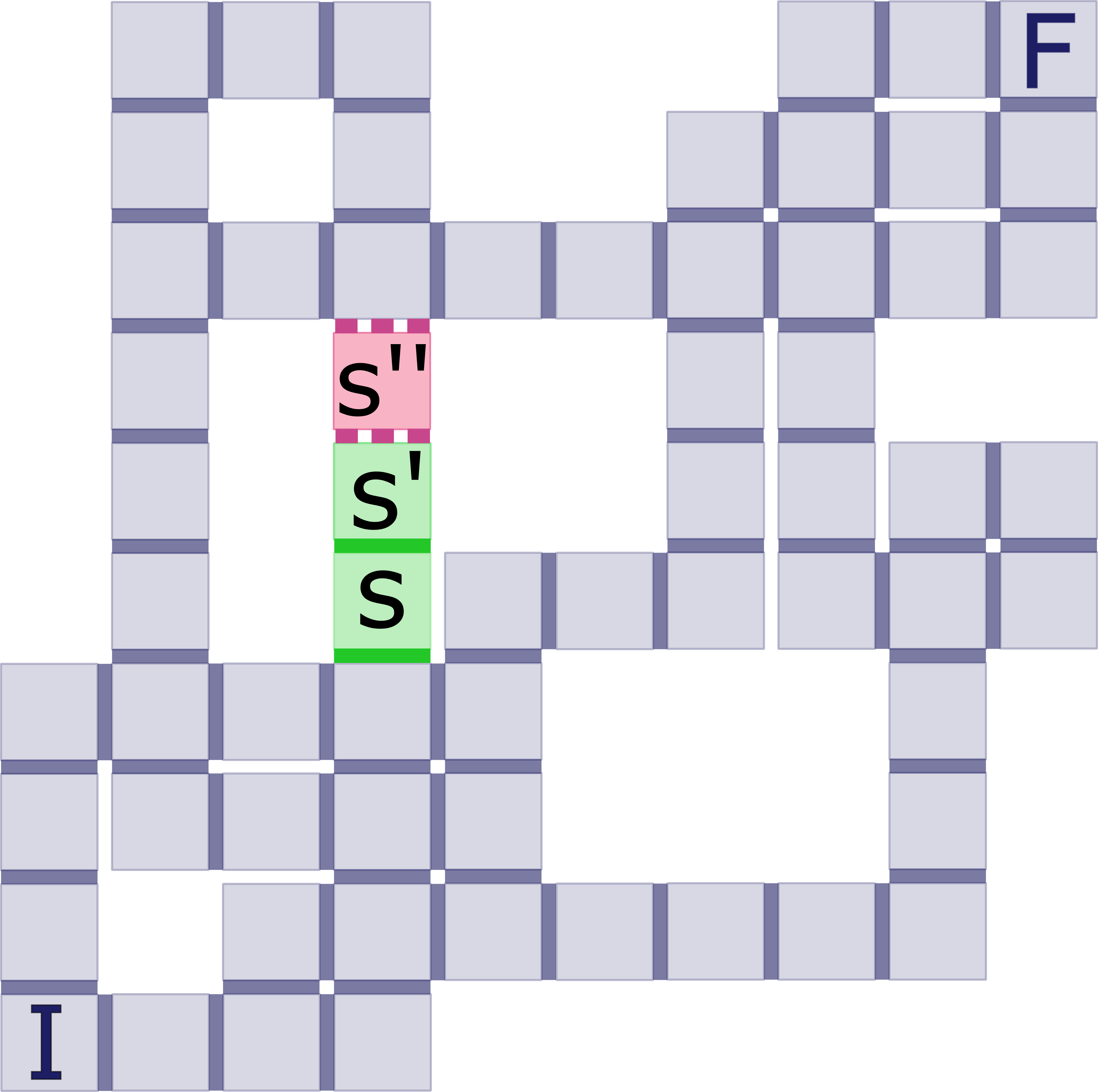}
\caption[]{In the elaboration step, the simulation finds $s$ and  $s'$  but not $s''$. Without detailed balance, a slow transition from $s'$  to $s$ could make the MFPT to $F$ large. However, in the full state space, $s'$  might quickly reach  $F$ via a fast transition to $s''$.}
\label{figdetailedbalance}
\end{figure}

Therefore, for \reversible CTMCs,   we  elaborate the pathway to  possibly include  states  that have a high probability of being visited with SSA but were not included with our  biased sampling. Here, we use SSA to elaborate the pathway;  we run $K$ simulations from each state of the pathway for  a maximum simulation time of $\pathwaytime$, meaning that a simulation stops as soon as the simulation time becomes greater than $\pathwaytime$.  By simulation time we mean the  time of a SSA trajectory, not the wall-clock time.  $K$  and $\pathwaytime$ are tuning parameters that affect the quality of  predictions.  The running time of elaborating the  states in the  pathway with this approach is  O($|\hat{\statespace}|K \pathwaytime k_\text{max}$), where $\hat{\statespace}$ is the state space of the pathway and  $k_\text{max}$ is the fastest rate in the pathway.  Alternatively, we could use a fixed number of transitions instead of a  fixed simulation time.  Another approach is to add all states that are within distance $r$ of every state of the pathway. However, with this approach, the size of the state space  could explode, whereas by using SSA the most probable states will be chosen.

Note that any elaboration  which stops before  hitting the target state might be problematic for  non-\reversible CTMCs.  Trajectories  that stop  while visiting a state for the first time might effectively be introducing a spurious sink into the enumerated state space. Without  reversibility that last transition of the elaboration might be   irreversible.  Sink states that are not a target state  make the MFPT to the target states infinite. For example  in Figure~\ref{figdetailedbalance},  assume   in the elaboration step, the simulation  finds $s$ and   $s'$  but not $s''$ (or any other neighbor of $s'$). Then without the  reversible transition, $s'$ will be a sink state and the MFPT to the target state  $F$ will be infinite. 
Moreover, having reversible transitions that do not obey the detailed balance condition may make MFPT estimations large.  For example, in Figure~\ref{figdetailedbalance} assume that the reversible transitions between $s$ and $s'$ do not obey detailed balance. Also, assume   $\pi(s)$ and $\pi(s')$  are both  high, and  that $\tmat(s,s')$ is large whereas $\tmat(s',s)$ is   small.  Therefore, if the elaboration stops at $s'$ it will make the MFPT large. However, in the full state space, $s'$  might quickly reach  $F$ through a fast transition to $s''$.   Thus,  the state elaboration step may not be suitable for non-\reversible CTMCs. \\


\noindent \textbf{Transition construction.} After the previous two steps,  fast transitions  between the states of the pathway  could still be missing.  To make computations more accurate, we further compute all possible  transitions in  $\hat{\statespace}$ that were not identified in the previous two steps.    In related roadmap planning work~\cite{kavraki1996probabilistic,tang2005using}, states are connected to their nearest neighbors as identified by a distance metric. We can include all missing transitions  by checking whether every two states in $\hat{\statespace}$ are neighbors in $\text{O}(|\hat{\statespace}|^2)$  or by checking for every state in $\hat{\statespace}$ whether its neighbors are also in $\hat{\statespace}$  in O($|\hat{\statespace}|m\}$), where $m$ is the maximum number of neighbors of the states in  the original CTMC.  \\

\noindent \textbf{$\delta$-pruning.} Given a  (truncated) CTMC in which  we can compute the MFPT from every state to the target  state, one question is: 
which states and transitions can be removed from the Markov chain without changing the  MFPT from the initial states significantly?
This question is especially relevant for the rapid evaluation of perturbed parameters, where MFPTs need to be recomputed often. 

 Given a CTMC  \CTMC\ and a pruning bound $\delta$,
  let 
the MFPT from any state  $s$ to $\statesfinal$ be $\tau_s$  and  let the  MFPT from the initial states  to 
$\statesfinal$ be $\expectt$. 
Let $\statespace_{\delta p}  = \left\{  \state \in \statespace~ \middle|~ \tau_s < \delta \expectt \mbox{ and } \pi_{0}(\state) = 0 	\right\}
$ be the set of states that are $\delta$-close to $\statesfinal$ and that are not an initial state.  We construct the $\delta$-pruned CTMC $\model_{\delta} = (\statespace_{\delta},\initialpi_{0}, \tmat_{\delta},\{s_d\})$   over the pruned set of states $\statespace_{\delta} = \statespace  \setminus \statespace_{\delta p} \cup \left\{ \state_{d} \right\} 
$, where $\state_{d}$ is the new target state. For $s, s' \in \statespace_{\delta}  \setminus   \left\{ \state_{d} \right\} $, we update  the rate matrix $ \tmat_{\delta} : \statespace_{\delta}  \rightarrow \mathbb{R}_{\geq 0} $   by  $ \tmat_{\delta}(s, s_d)    = \sum_{\state' \in \statespace_{\delta_p}}\tmat(\state, \state') $ and $ \tmat_{\delta}(s, s') =  \tmat(s,s')$. Note that  $\tmat_{\delta}(s_d,s)$ is not used in the computation of the MFPT  (\eqn~\ref{mfptequation}), so we can simply assume $\tmat_{\delta}(s_d,s)=0$.  Alternatively, to retain \reversible conditions, we can  define the  energy of $s_d$ as  $E(s_d) = -RT \log \sum_{s''\in \statespace_{\delta p}} e^{- \frac{E(s'')}{RT}} $ (see \eqns 7.1 and 7.2 from Schaeffer~\cite{schaeffer2013stochastic}) and  define 
$\tmat_{\delta}(s_d,s)= e^{- \frac{E(s)-E(s_d)}{RT}}\tmat_{\delta}(s,s_d)$.  For the pruned CTMC $\model_{\delta} = (\statespace_{\delta},\initialpi_{0}, \tmat_{\delta},\{s_d\})$, let the MFPT $ \expecttdelta$ be given as usual (\eqn~\ref{mfpt}). Then by construction 
\begin{equation}
 \expecttdelta  \leq \frac{\expectt}{1 + \delta}.
 \label{deltapruningbound}
\end{equation}

We can calculate the MFPT from every state to the target states by  solving  \eqn~\ref{mfptequation} once for CTMC $\model$. Therefore, the running time of $\delta$-pruning depends on the running time of the matrix equation solver that is used. For a CTMC with state space $\statespace$, the running time of a direct  solver  is at most O($|\statespace|^3$). For iterative solvers the running time is generally less than O($|\statespace|^3$).   After the equation is solved, the CTMC can be pruned in O($|\statespace|$) for any $\delta$.    
   Note that for  a given bound $\delta$,  the running time for solving \eqn~\ref{mfptequation}  for the  pruned CTMC $\model_{\delta}$ might still be  high. In that case, a larger value of $\delta$ is required. To set  $\delta$ in practice, it could be useful to consider the number of states that will be pruned for a given $\delta$, that is $|\statespace_{\delta p}|$.  
 \\
 
\noindent \textbf{Updating perturbed parameters.} We are interested in rapidly estimating the MFPT to  target states  given mildly  perturbed parameters.  Our approach is to reuse a truncated CTMC for mild perturbations.  The MFPT estimates will be biased in this way. However,  we could have significant savings in running time by  avoiding the cost of  sampling and building truncated CTMCs  from scratch. We would still have to  solve \eqn~\ref{mfptequation}, but it could be negligible compared to the other costs.    For example in Table~\ref{kssavskpathwaytable}, on average, solving the matrix  equation is  faster than SSA by a factor of 47  and is faster than   building the truncated CTMC by a factor of 10. 

A perturbed thermodynamic model parameter affects the energy of the states. Therefore, to update the transition rates,  we would also  have to recompute the energy of the states.  A  perturbed kinetic model only affects the transition rates. A perturbed experimental condition could affect both the energy of the states and the transition rates.   Therefore, assuming the energy of a state can be updated in a constant time, the  truncated CTMC can be updated in  O($|\hat{\statespace}| + |\hat{\mathcal{E}}|$), where   $\hat{\mathcal{E}}$  is the  set of transitions of the truncated CTMC.      For nucleic acid kinetics  with elementary steps, the energy of a state can be computed from scratch in O($b$) time, 
  or in O($1$) time using  the energy calculations  of  a neighbor state~\cite{schaeffer2013stochastic}.\\

\noindent \textbf{Quantifying the error.} After we build truncated CTMCs, we need to quantify  the error of MFPT estimates when experimental  measurements are not available.  It would help us  set values for  $N$, $\beta$, $K$ and $\pathwaytime$ for  fixed model parameters, and also  evaluate when a truncated CTMC  has a high error for   perturbed model parameters.    For exponential decay processes, one possible approach is to adapt the finite state projection  FSP~\cite{munsky2006finite} method that is developed to quantify the error of truncated CTMCs for transient probabilities. We adapt it as follows. We  combine all target states into one single  absorbing state $\stateFinal$.  We  project all states that are not in the truncated CTMC into an absorbing  state $s_o$ and we redirect all transitions from the truncated CTMC to  states out of the CTMC into $s_o$.  Then we  use the standard matrix exponential equations to compute the full distribution on the state space at a given time. However,  we  only care about the probabilities that  $\stateFinal$ and $s_o$  are occupied.     We search to compute the half-completion time $t_{1/2}$ with bounds by
 \begin{equation}
  \begin{cases}
    t_\text{min} & \text{s.t. } p(\stateFinal \;;\; t_\text{min}) + p(s_o \; ; \;t_\text{min}) =  \frac{1}{2}, \\
    t_\text{max} &  \text{s.t. } p(\stateFinal \; ;  \; t_\text{max}) = \frac{1}{2},		\label{eq-tmax}
  \end{cases} 
  \end{equation} 
where $p(s \; ; \; t)$ is the probability that the process will be at state $s$ at time $t$ starting from the set of initial states.  Since $\stateFinal$ and $s_o$ are the only absorbing states, then $t_\text{min}$ exists and clearly  $t_\text{min} \leq t_{1/2}$. Based on FSP, $p(\stateFinal \; ;  \; t_\text{max})$ is an underestimate of the actual probability at time $t_\text{max}$, if it exists.  A possible way to determine if a solutions exists is to determine  the probability of reaching  state $\stateFinal$ compared to state $s_o$ from the initial states, which can be calculated by solving a system of linear equations (see \eqn~2.13 from \citet{metzner2009transition}).   If the probability is greater or equal to $\frac{1}{2}$ then a solutions exists.     If a solution does not exist for the given statespace, then based on FSP the error is guaranteed to decrease   by adding more states and we can eventually find  a solution to \eqn~\ref{eq-tmax}.   The search for $t_\text{max}$ can be completed with binary search.   Thus, the true $t_{1/2}$ is guaranteed to satisfy $t_\text{min}  \leq t_{1/2} \leq t_\text{max} $.   For exponential decay processes,  the  relation between the half-completion time and the MFPT  is~\cite{cohen1977quantum,simmons1972differential} 
\begin{equation}
t_{1/2} = \frac{\text{ln}2}{\lambda} \text{ and  } \tau = \frac{1}{\lambda}\rightarrow \tau = \frac{t_{1/2}}{ \text{ln} 2}, 
\end{equation}
where $\lambda$ is the rate of the process. 
Thus,  $\frac{t_\text{min}}{\text{ln}2}  \leq \tau  \leq \frac{t_\text{max}}{\text{ln}2} $. 

A drawback of this approach is that we might need a large number of states to find a solution to \eqn~\ref{eq-tmax},  which might make the master equation or the linear system solver infeasible in practice.  Efficiently quantifying the error of MFPT estimates in truncated CTMCs for exponential and non-exponential decay processes is beyond the  scope of this paper. It might be possible to use some other existing work~\cite{kuntz2019exit,backenkohler2019bounding}.


\section{Dataset and Experiments for Interacting Nucleic Acid Strands}\label{results}
We implement  pathway elaboration  for  interacting nucleic acid strands on top of Multistrand~\cite{schaeffer2013stochastic,schaeffer2015stochastic} (see Section~\ref{preliminaries-multistrand} for related background).   Our framework and dataset are available at \url{https://github.com/DNA-and-Natural-Algorithms-Group/PathwayElaboration}.

Here, in Section~\ref{dataset},     we describe our dataset  of  nucleic acid kinetics. In Section~\ref{exp-setup}, we describe  our experimental setup that is common in our experiments. In Section~\ref{exp-casestudy}, we use pathway elaboration in a case study to gain insight on the kinetics of two contrasting  reactions. In Section~\ref{exp-eval}, first we evaluate estimations of pathway elaboration by comparing them with estimations of SSA. Then we   build truncated CTMCs using SSA and TPS on a subset of our dataset and compare their performance with pathway elaboration. After that, we show the effectiveness of the $\delta$-pruning step. Finally, in Section~\ref{section-parameterestimation}, we use pathway elaboration  for the rapid evaluation of perturbed parameters in parameter estimation.

\subsection{Dataset of Interacting Nucleic Acid Strands}
\label{dataset}
We conduct computational experiments on  interacting nucleic acid strands  (see Section~\ref{preliminaries-multistrand} for related background).   The speed at which nucleic acid strands interact is difficult to predict  and depends on reaction topology, strands' sequences, and experimental conditions.   The number of secondary structures  interacting nucleic strands may form is exponentially large in the length of the strands. Typical to these reactions are high energy barriers that prevent the reaction from completing, meaning that long periods of simulation time are required before successful reactions occur. Consider reactions that occur with rates lower than $10000 \perMolePerSecond$ such as 
three-way strand displacement at room temperature (see Table~\ref{sourceTable}). 
 These types of reactions are slow to simulate not because the simulator takes longer to generate trajectories for larger molecules,
 but the slowness is instead a result of the energy landscape: at low temperatures, duplexes simply are more stable, and require longer simulated time until their dissociation is observed. 
 \sh{ Predicting the kinetics of interacting nucleic acid strands is also difficult with classical machine learning methods and neural network models. For example, \citet{zhang2018predicting}  successfully predict hybridization rates with a weighted neighbour voting prediction algorithm and  \citet{angenent2020deep} successfully predict toehold switch function with neural networks.  However, despite the   accurate and fast computational prediction of these methods, to treat different type of reactions or to treat different initial and target states the models have to be  adapted. On the other hand, the  CTMC model of Multistrand can readily be applied to  unimolecular and bimolecular reactions. Moreover, the CTMC model could provide an unlimited number of  unexpected intermediate states.   In contrast, with the neural networks
models of \citet{angenent2020deep},  which use attention maps to  interpret intermediate states, the number is limited.
}
 \begin{table}\centering
\caption{ \label{sourceTable}  Summary of the dataset of 267 nucleic acid kinetics.   The  initial concentration of the reactants is denoted as $u$ and $k$ is the experimental reaction rate constant.   } 
\hspace*{0.3cm}\makebox[\linewidth][c]{
 \begin{tabular}{C{0.8cm}C{3.2cm} C{.9cm}C{0.8cm}C{1.4cm}C{0.9cm}C{1.2cm}C{1.2cm}}
  \hline
Dataset No. & Reaction type \& source$^{\dagger}$ 	 &  \# of reactions	& Mean \# of bases &  $[\text{Na}^{+}]$ (M)	&  T  (\degreeC)  & u (M) & $\logten {k}$  \\ \hline
 \end{tabular}}
\vspace{-0.1cm}\[
 \begin{turn}{90} $\mathcal{D}_\text{train}$\end{turn}
  {
  \left\{
 \begin{tabular}{C{0.8cm}C{3.2cm} C{.9cm}C{0.8cm}C{1.4cm}C{0.9cm}C{1.2cm}C{1.2cm}}
  \hline
     $1$ &  Hairpin   opening~\cite{bonnet1998kinetics}& $63$&$25$ &  \numrange[range-phrase = --]{0.15}{0.5}						&\numrange[range-phrase = --]{10}{49}	 &$ 1\times 10^{-8}$  &  \numrange[range-phrase = --]{1.4}{4.6} \\ \hline 
  $2$ &Hairpin closing~\cite{bonnet1998kinetics} &	 $62$ & $25$&  \numrange[range-phrase = --]{0.15}{0.5}							&\numrange[range-phrase = --]{10}{49} &	$ 1\times 10^{-8}$    & \numrange[range-phrase = --]{3.2}{4.8} \\
   \hline
   $3$ &Helix dissociation~\cite{cisse2012rule}&	 $39$ 		& $18$		& \numrange[range-phrase = --]{0.01}{0.2} &	 \numrange[range-phrase = --]{23}{37} & $ 1\times 10^{-8}$    & \numrange[range-phrase = --]{-1.2}{0.9}\\ \hline
		       	 	$4$ &
  		Helix association~\cite{hata2017influence} & $43$ &$46$	&  $0.195$				&  $25$ & $ 5\times 10^{-8}$ & \numrange[range-phrase = --]{4.0}{6.7}	\\  \hline
  		 $5$ &Helix association~\cite{zhang2018predicting}	& $20$ & $72$&$0.75$ &  	 \numrange[range-phrase = --]{37}{55} &  $1 \times 10^{-5}$ &  \numrange[range-phrase = --]{4.4}{7.4} \\ 
   \hline $6$& Toehold-mediated three-way strand displacement~\cite{machinek2014programmable}  & $10$& $102$
				& $0.05^{\dagger\dagger}$	& 				$23$ &  \numrange[range-phrase = --]{5e-9}{1e-8}&   \numrange[range-phrase = --]{5.3}{6.8}
					  \end{tabular}
\right.
   }
\]
  \vspace{-.4cm}
\[
 \begin{turn}{90} $\mathcal{D}_\text{test}$\end{turn}
  {
  \left\{
\begin{tabular}{C{0.8cm}C{3.2cm} C{.9cm}C{0.8cm}C{1.4cm}C{0.9cm}C{1.2cm}C{1.2cm}}
   \hline        	 	$7$ &
  		Helix association~\cite{hata2017influence} & $4$ &$46$	&  $0.195$				&  $25$ & $ 5\times 10^{-8}$ &  \numrange[range-phrase = --]{4.0}{5.0}  	\\  \hline
						 $8$&  Toehold-mediated three-way strand displacement~\cite{machinek2014programmable} &	 $26$& $ 100$
				& $0.05^{\dagger\dagger}$&23	&  \numrange[range-phrase = --]{5e-9}{1e-8} & \numrange[range-phrase = --]{2.7}{6.3} \\
						\hline

  \end{tabular}
\right.
   }
\]
\vspace{-1.2cm}
  \begin{minipage}{14cm}%
\footnotesize{
$^{\dagger}$ See Figure~\ref{reactionexamples}  for example figures of these reactions. \\
$^{\dagger\dagger}$ The  experiment was  performed without {$\text{Na}^{+}$} in the buffer. 
}
\vspace{1cm}
  \end{minipage}%
\end{table}

 
  We curate a dataset of   267 interacting DNA strands from the  published literature, summarized  in Table~\ref{sourceTable}.   The reactions are annotated with the  temperature, the buffer condition, and the experimentally determined reaction rate constant.   
 The dataset  covers a wide range of slow and fast  unimolecular and bimolecular reactions where the reaction rate constants vary over 8.6 orders of magnitude. 
  For unimolecular reactions,  we consider  hairpin opening~\cite{bonnet1998kinetics}, hairpin closing~\cite{bonnet1998kinetics}, and helix dissociation~\cite{cisse2012rule}. For bimolecular reactions, we consider helix association~\cite{hata2017influence,zhang2018predicting} and  toehold-mediated three-way strand displacement~\cite{machinek2014programmable}.  The reactions from \citet{cisse2012rule} and \citet{machinek2014programmable}  may have mismatches between the bases of the strands. \sh{  The type of reactions in Table~\ref{sourceTable} are widely used in nanotechnology, such as  in molecular  beacon probes~\cite{chen2015dna}}.

  For  bimolecular reactions, we Boltzmann sample initial reacting complexes.   For reactions in which  we define only one target state, in the pathway construction step, we bias the paths towards that state. In this work, for reactions in which we define a set of target states,  we bias  paths towards  only one target state, so that $\initialpi'(\statebias)=1$ for one state and $\initialpi'(s)=0$ for all other states.  \sh{Next we describe these states.} 
    
     \textbf{Hairpin closing and hairpin opening.}  For a hairpin opening reaction, we define the initial state to be the  system microstate  in which a  strand has fully formed a duplex and a loop  (see Figure~\ref{hairpin-fig}).  We define the target state to be the system microstate in which the strand has no base pairs.
 Hairpin closing is the reverse reaction, where a strand with no base pair forms a fully formed duplex and a loop. 
 
 \textbf{Helix dissociation and helix association.}   
For a helix dissociation reaction, we specify the initial state to be the system microstate in which  two strands  have fully formed a  helix (see Figure~\ref{helix-fig}). We define the set of target states to be the set of  system microstates in which  the strands have detached and there are no base pairs within one of the strands. We bias  paths towards the target state in which there are no base pairs formed within any of the strands.
 Helix association is the reverse reaction. We  Boltzmann sample the initial reacting complexes in  which the strands have not formed base pairs with each other. We define the target state to be the system microstate in which the duplex has fully formed. 
 
  \textbf{Toehold-mediated  three-way strand displacement.}
 In this reaction, an invader strand displaces an incumbent strand in a duplex, where a toehold domain facilitates the reaction (see Figure~\ref{threewaystranddisplacement-fig} and Figure~\ref{fig2}).   We Boltzmann sample initial reacting complexes in which the incumbent and substrate form a complex through base pairing and the invader forms another complex.   We define the set of target states  to be the set of microstates where the incumbent is detached from the substrate and there are no base pairs within the incumbent. We bias  paths towards the  target state  in which the substrate and invader have fully formed base pairs  and there are no base pairs within the incumbent.

  In datasets No.~1-6 from Table~\ref{sourceTable}, we consider reactions that are feasible with  SSA with our parameterization of Multistrand, given two weeks computation time,  since we compare SSA results with pathway elaboration results.   We indicate these reactions as $\mathcal{D}_\text{train}$ since we also use them as training set in  Section~\ref{section-parameterestimation}. We indicate
  datasets No. 7-8  as $\mathcal{D}_\text{test}$ since we use them as testing set in  Section~\ref{section-parameterestimation}.

\subsection{Experimental  Setup}~\label{exp-setup}
Experiments are performed on a system with 64 2.13GHz Intel Xeon processors and 128GB RAM in total, running openSUSE Leap 15.1.  An experiment for a reaction is conducted on one processor.  Our framework is implemented in Python, on top of the Multistrand kinetic simulator~\cite{schaeffer2013stochastic,schaeffer2015stochastic}.    To solve the matrix equations in \eqn~\ref{mfptequation},  we use the sparse direct  solver from SciPy~\cite{virtanen2020scipy} when possible\footnote{The implementation we used  allowed the sparse direct solver to use only up to 2GB of RAM}. Otherwise we use the sparse iterative   biconjugate gradient algorithm~\cite{fletcher1976conjugate} from SciPy. 

In all of our experiments, the thermodynamic parameters for predicting the energy of the states are fixed and the energies are calculated with  Multistrand.  Each reaction uses its own experimental condition as provided in the dataset.  In all our experiments, we use the Metropolis kinetic model from Multistrand.  For all experiments except for Section~\ref{section-parameterestimation}, we fix the kinetic parameters  to  the Metropolis Mode parameter set~\cite{zolaktaf2017inferring}, that is $\theta_1= \{ \kUni \approx 2.41 \times 10^6 \text{ s}^{-1}, \kBi \approx 8.01 \times 10^5 \text{ M}^{-1}\text{s}^{-1} \}$.   To obtain MFPTs with SSA,  we use  1000 samples, except for three-way strand displacement reactions   in  which we use 100 samples, since the simulations take a longer time to complete.

\begin{figure}
\center
\subfloat[]{ \label{fig2a}\includegraphics[width=1\textwidth]{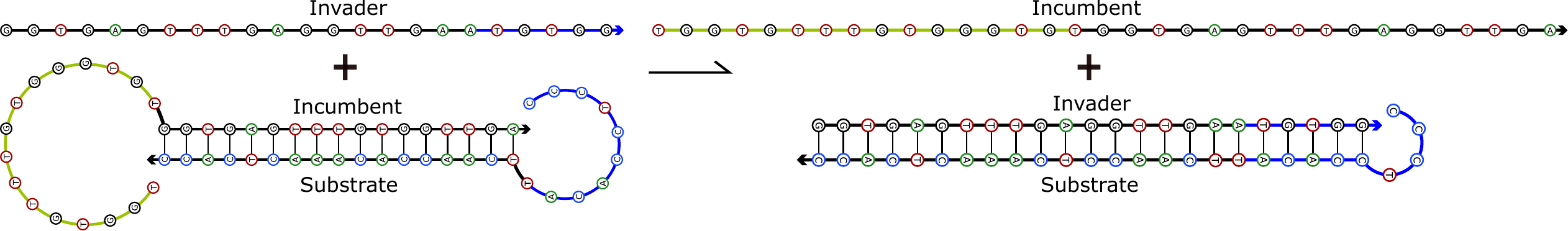}}\\ \subfloat[]{ \label{fig2b}\includegraphics[width=1\textwidth]{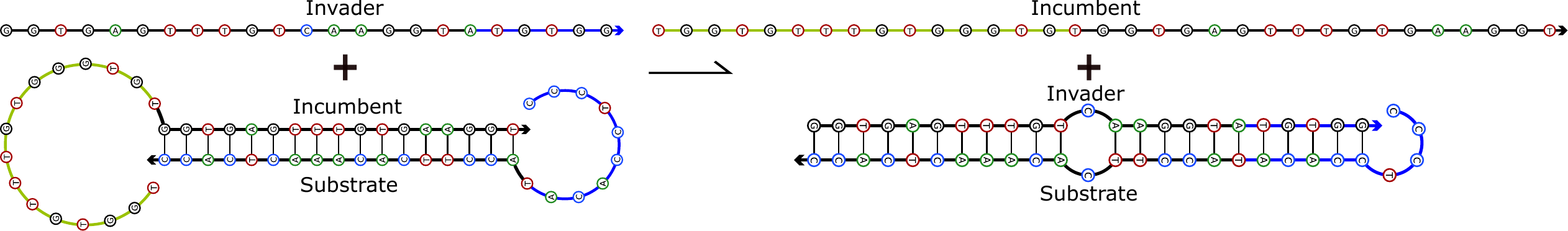}}\\
 \subfloat[]{ \label{fig2c}\includegraphics[width=0.25\textwidth]{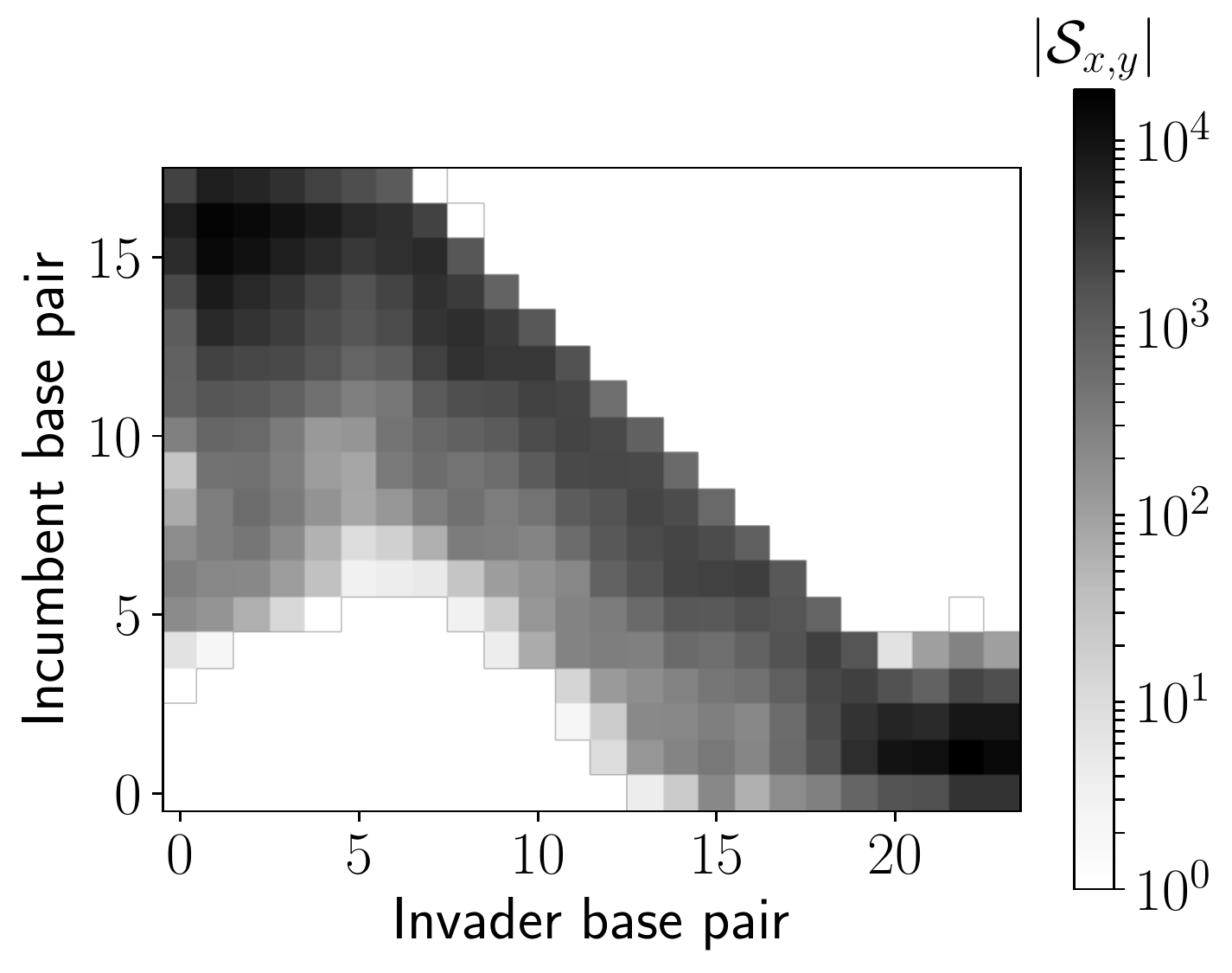}}
  \subfloat[]{ \label{fig2d}\includegraphics[width=0.25\textwidth]{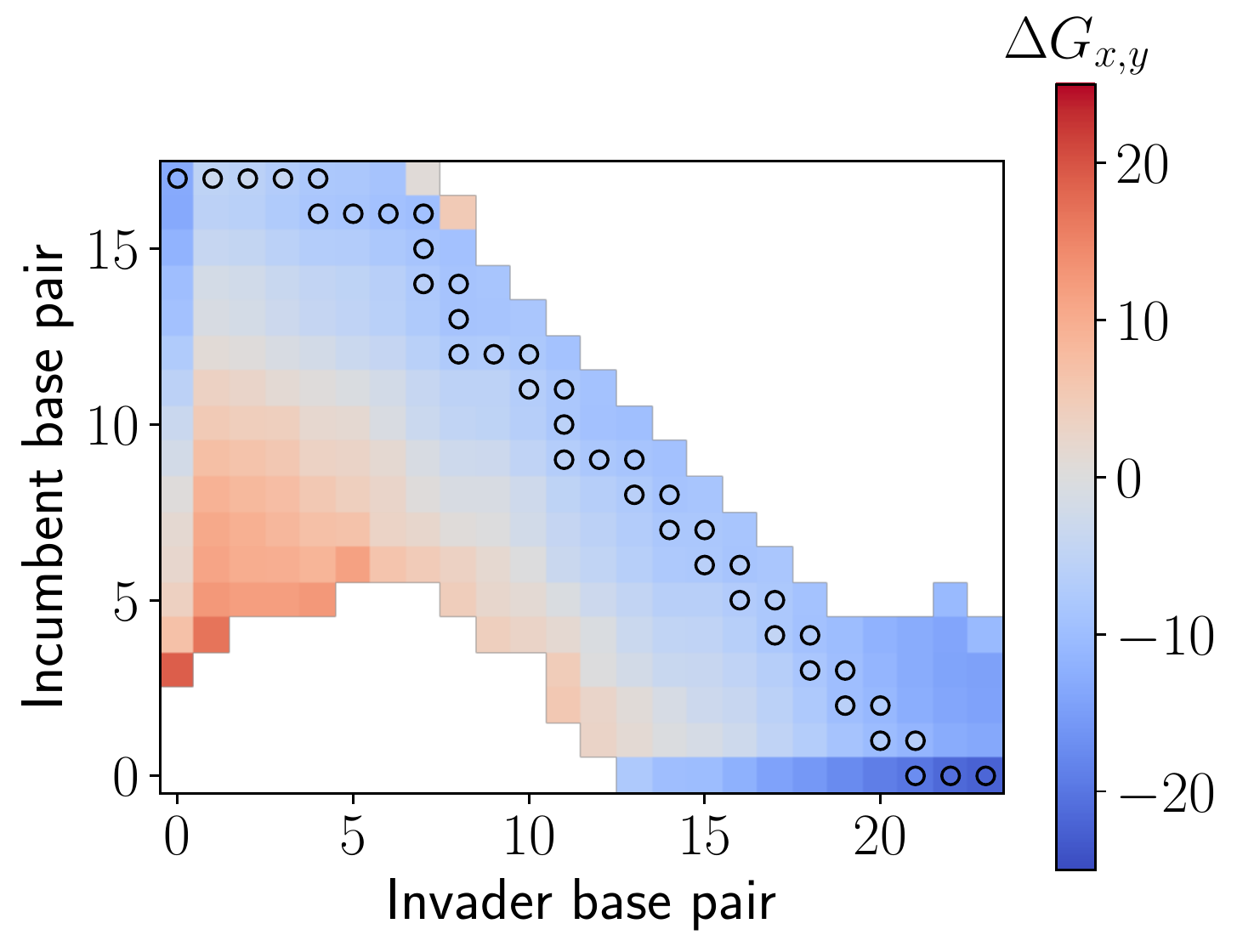}}
  \subfloat[]{\label{fig2e} \includegraphics[width=0.25\textwidth]{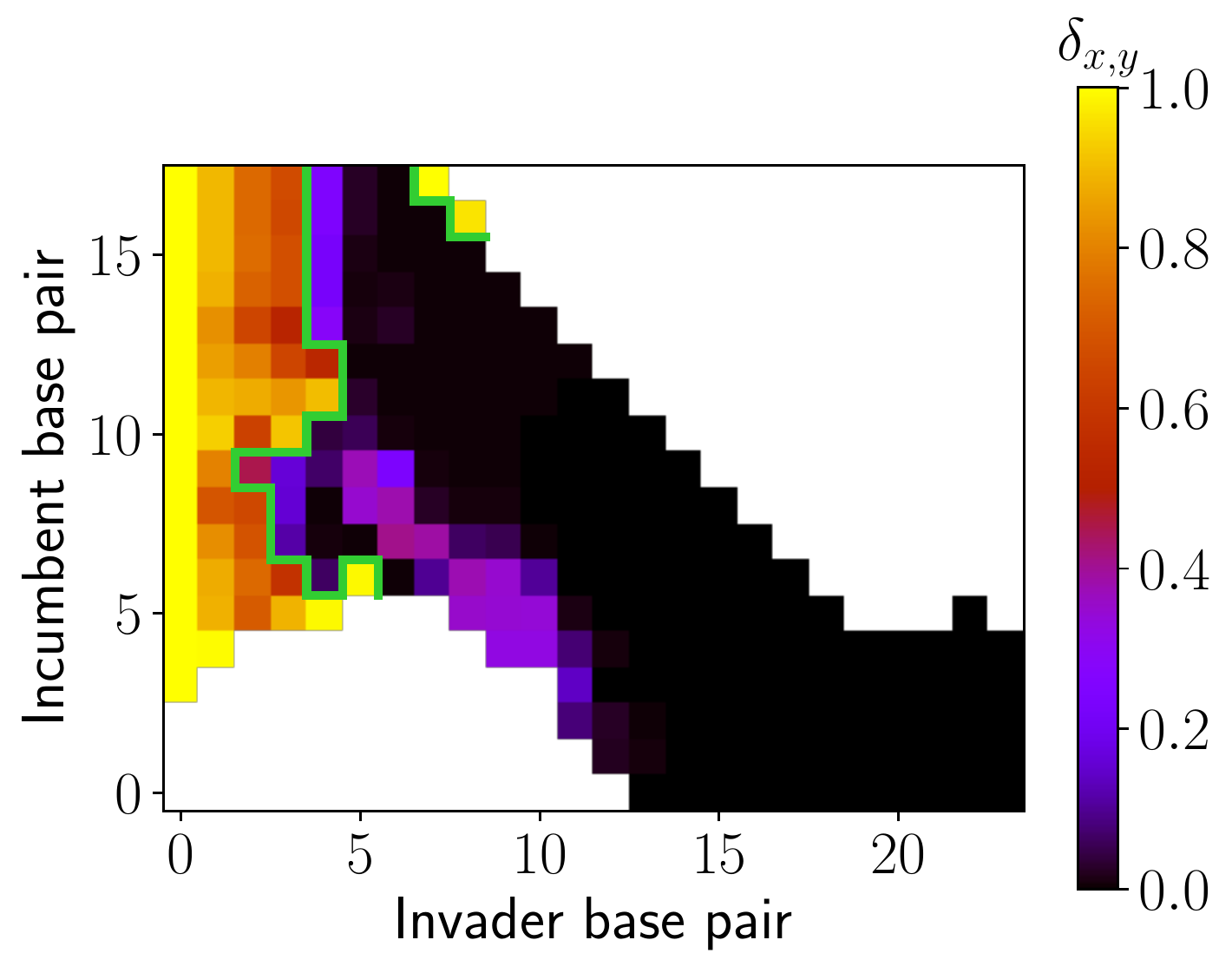}}
  \subfloat[]{ \label{fig2f}\includegraphics[width=0.235\textwidth]{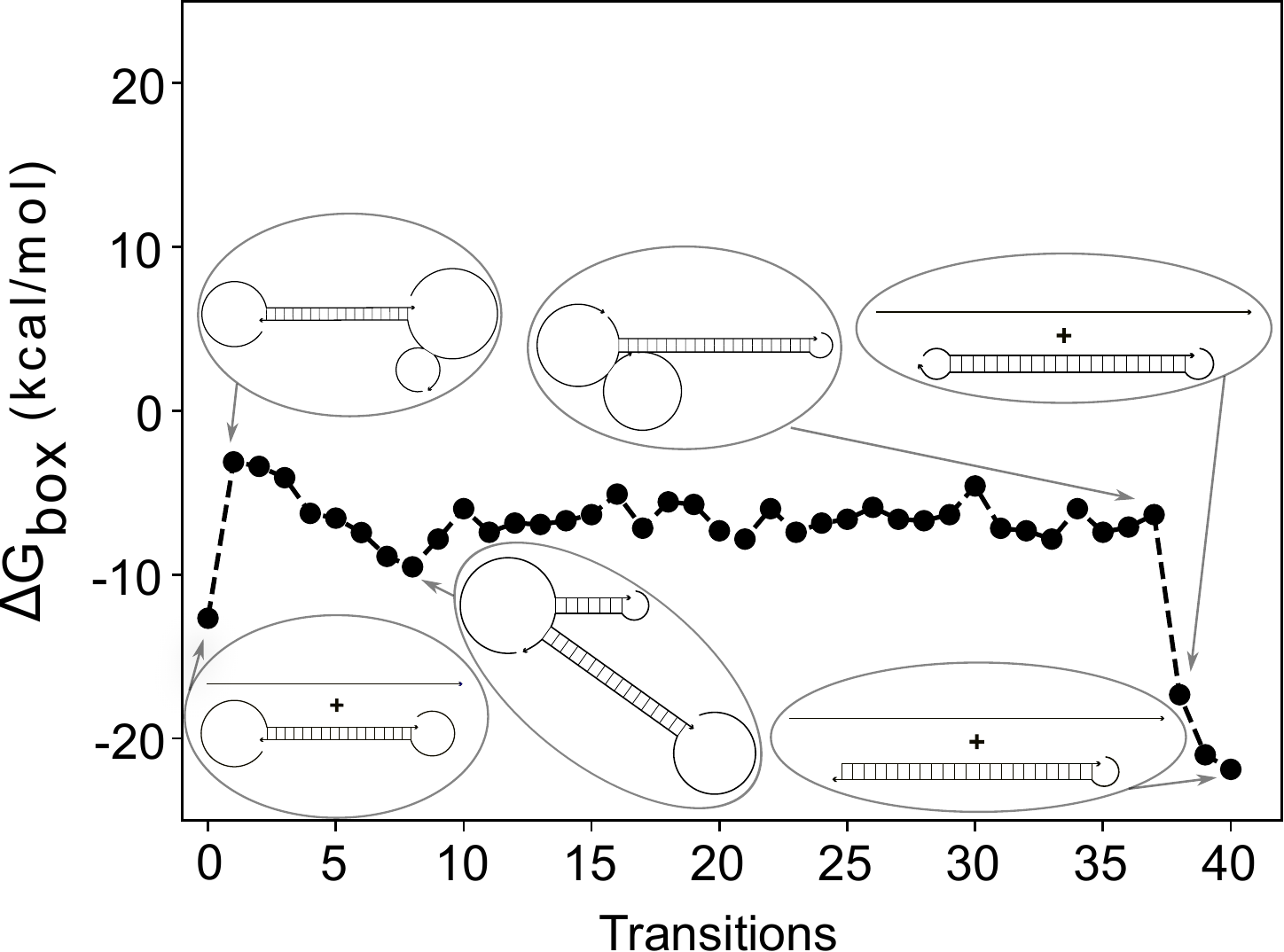}}\\
  \subfloat[]{ \label{fig2g}\includegraphics[width=0.25\textwidth]{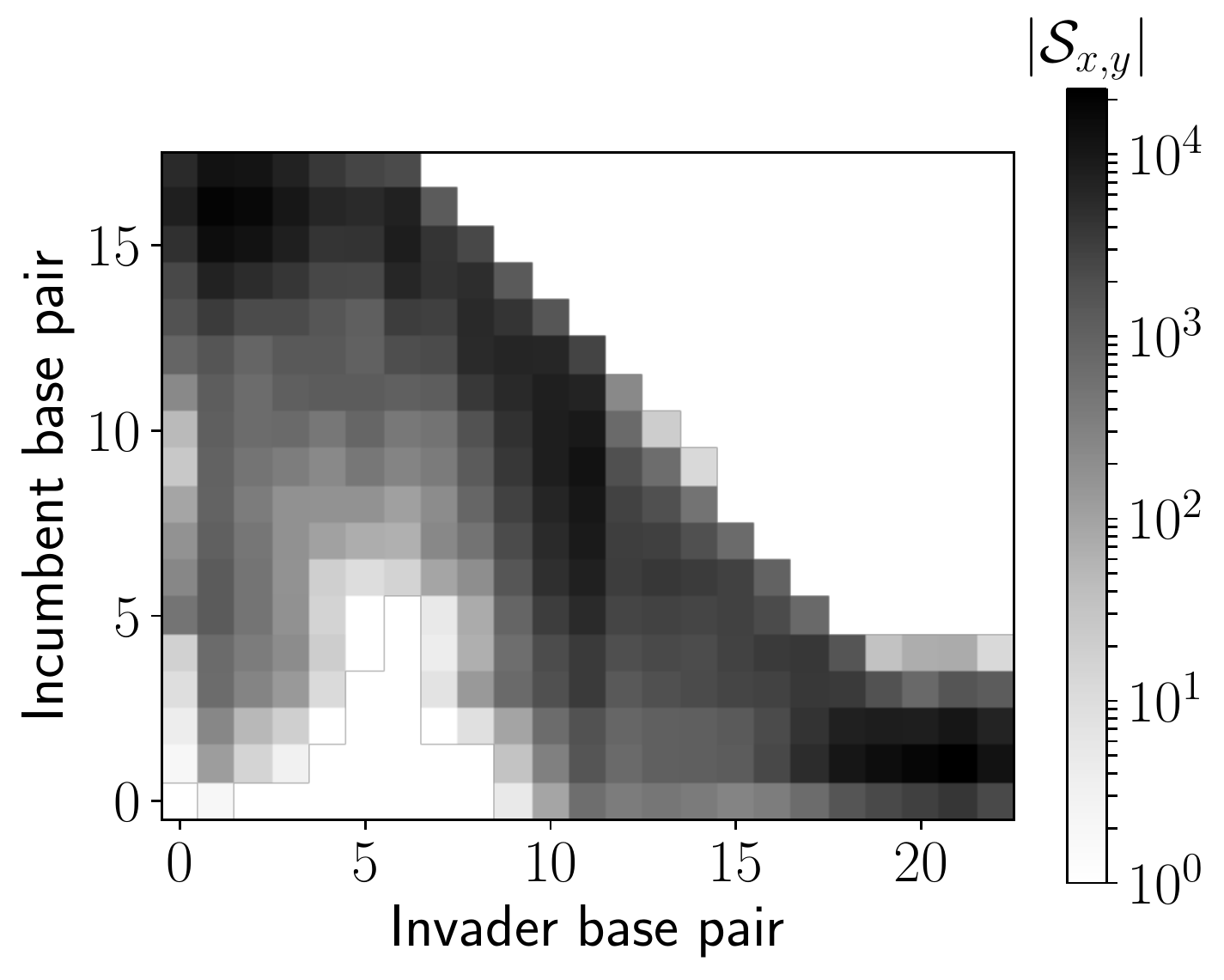}}
  \subfloat[]{\label{fig2h} \includegraphics[width=0.25\textwidth]{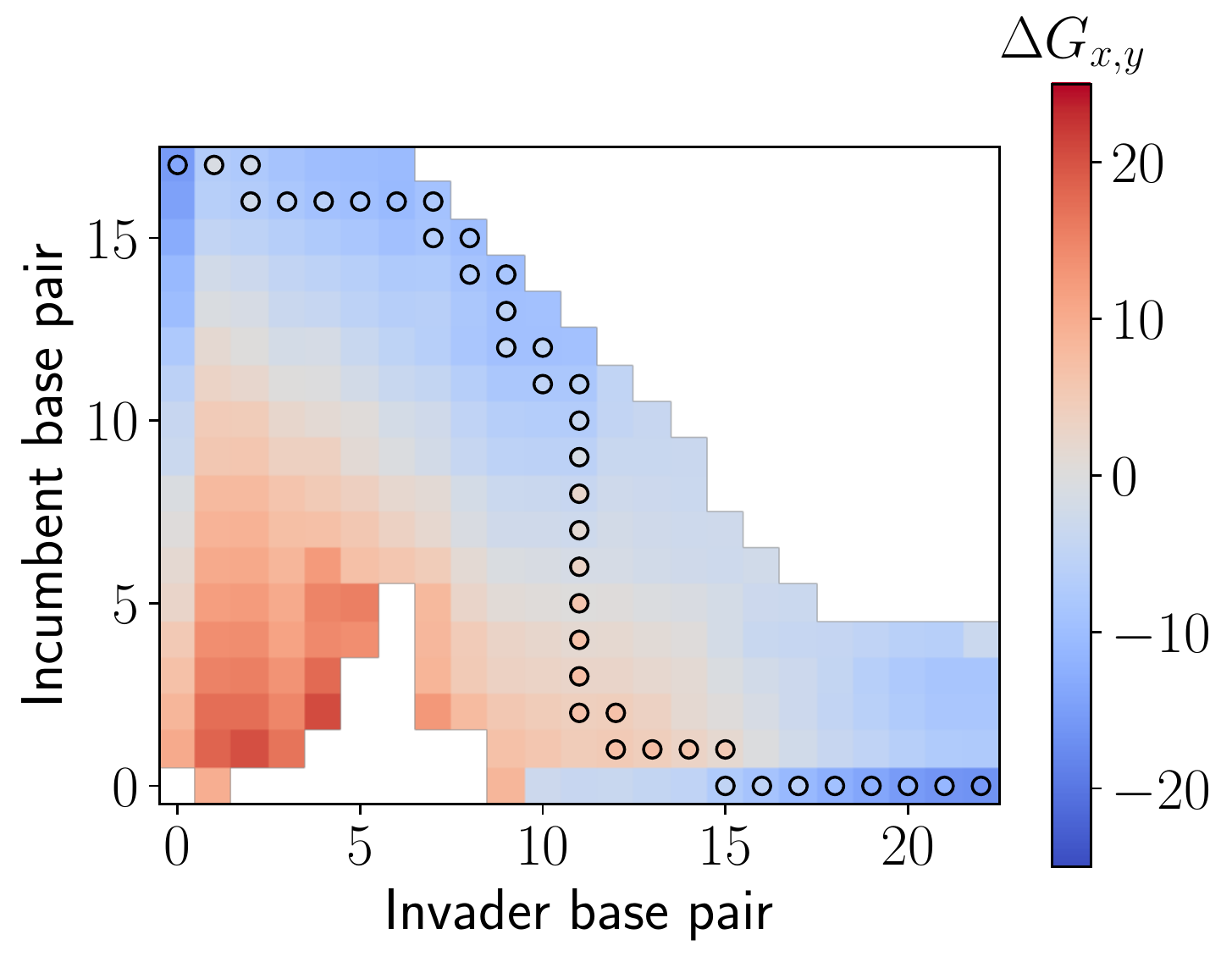}}
\subfloat[]{ \label{fig2i}\includegraphics[width=0.25\textwidth]{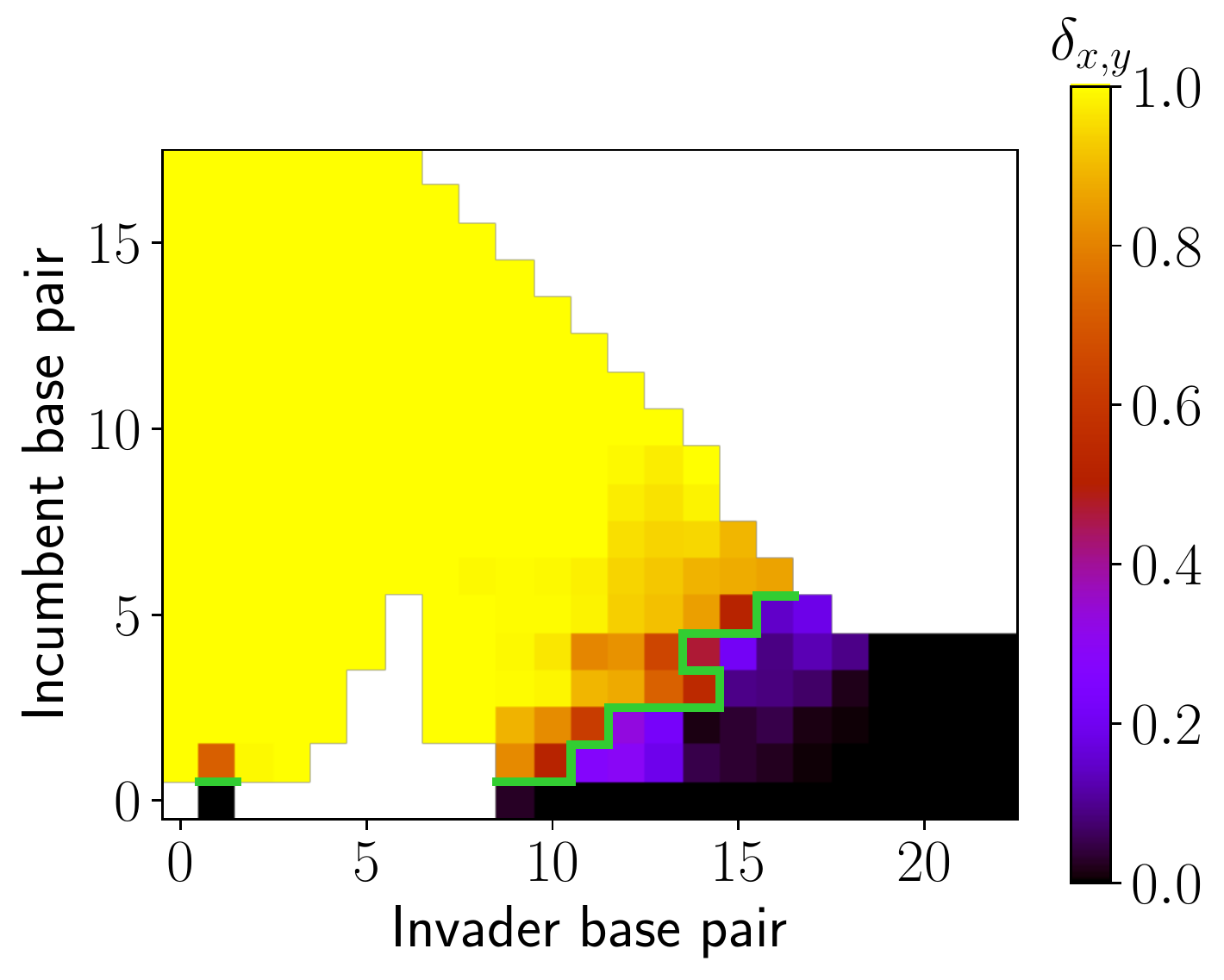}}
\subfloat[]{ \label{fig2j}\includegraphics[width=0.235\textwidth]{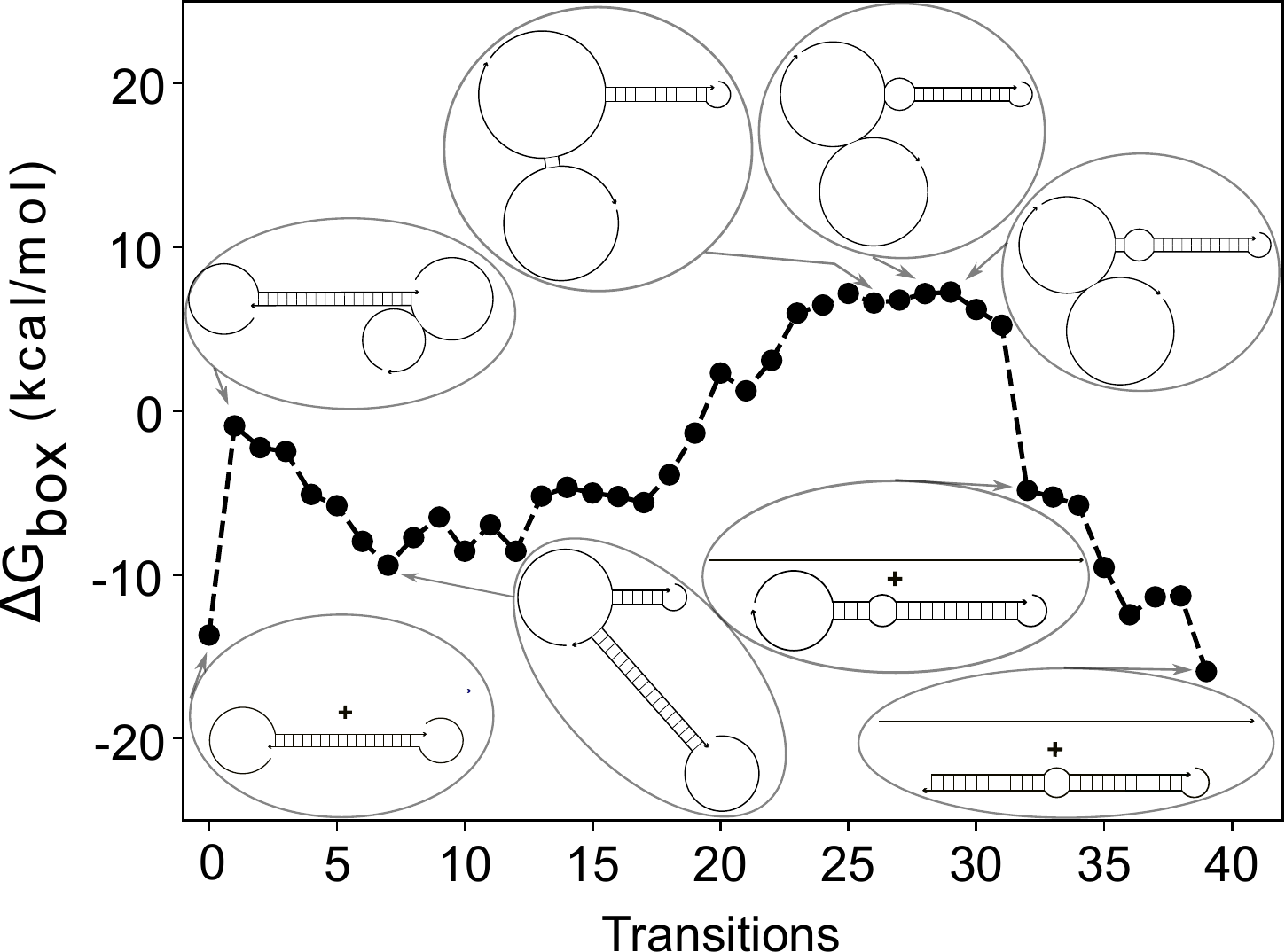}}
\caption{\label{fig2}Results of truncated CTMCs built with pathway elaboration ($N=128$, $\beta=0.6 $, $K=1024$, $\pathwaytime=16$ ns) for two toehold-mediated three-way  strand displacement  reactions from \citet{machinek2014programmable}. (\textbf{a})  A toehold-mediated three-way strand displacement reaction  that has a  6-nt toehold and a 17-nt displacement domain~\cite{machinek2014programmable}.  (\textbf{b})  A toehold-mediated three-way strand displacement reaction that has a 6-nt toehold, a 17-nt displacement domain, and a mismatch exists between the invader and the substrate at position 6 of the displacement domain~\cite{machinek2014programmable}. Figures~\ref{fig2c},~\ref{fig2d},~\ref{fig2e}, and~\ref{fig2f} correspond to Figure~\ref{fig2a}. Figures~\ref{fig2g},\ref{fig2h},\ref{fig2i}, and~\ref{fig2j} correspond to Figure~\ref{fig2b}.  In Figures~\ref{fig2c},~\ref{fig2d},~\ref{fig2e},~\ref{fig2g},~\ref{fig2h}, and~\ref{fig2i},  the x-axis   corresponds to the number of base pairs between the invader and the substrate, and  the y-axis corresponds to the number of base pairs between the incumbent and the substrate.   (\textbf{c, g}) At coordinate $(x,y)$, $|\mathcal{S}_{x,y}|$ is shown, where $\mathcal{S}_{x,y}$ is a system macrostate (a nonempty set of system microstates) equal to the set of states    with coordinate $(x,y)$.   (\textbf{d, h})  At coordinate $(x,y)$,   the free energy  $\Delta G_{x,y}$ is shown, which is defined as $\Delta G_{x,y}=-RT \ln \sum_{s\in \mathcal{S}_{x,y} } e^{\frac{-\Delta G(s)}{RT}}$~\cite{schaeffer2013stochastic}.  The free energy of the  paths in Figures~\ref{fig2f} and~\ref{fig2j} are also shown with the $\circ$ marker in  Figures~\ref{fig2d} and~\ref{fig2h}, respectively.   (\textbf{e, i}) At coordinate $(x,y)$, the value of $\delta_{x,y} =\sum_{s \in \mathcal{S}_{x,y}} \frac{w_s \delta(s)}   {\sum_{s \in \mathcal{S}_{x,y}} w_s}$ is shown, where  $\delta(s) = \tau_s/ \expectt$ and $w_s= e^{\frac{-\Delta G(s)}{RT}}$.  For ease of understanding, the green  ``halfway line'' separates   coordinates where  $\delta_{x,y}$ is greater than 0.5 from coordinates where  $\delta_{x,y}$ is less than 0.5.  (\textbf{f,  j}) The free energy landscape of a random path built with pathway elaboration ($N=1$, $\beta=0$, $K=0$, $\pathwaytime=0$ ns) and the  initial and the final states  and some states near the local extrema are illustrated.	 
 }
\end{figure}

\subsection{Case Study}\label{exp-casestudy}
Here we illustrate the use of pathway elaboration to gain insight on the kinetics of two contrasting reactions from \citet{machinek2014programmable}, one being a rare event.    

Figures~\ref{fig2a}  and~\ref{fig2b} show the two teohold-mediated three-way strand displacement reactions that we consider~\cite{machinek2014programmable}. In the reaction in Figure~\ref{fig2a}, the invader and substrate are complementary strands in the displacement domain.  In the reaction in Figure~\ref{fig2b}, there is a mismatch between the invader and the substrate in the displacement domain.   The rate of toehold-mediated strand displacement is usually determined by the time to complete the first bimolecular transition, in which the invader forms a base pair with  the substrate for the first time.  However, the rate could be controlled by several orders of magnitude by altering positions across the strand, such as using mismatch bases~\cite{machinek2014programmable}.  The reaction  in Figure~\ref{fig2b} is  approximately 3 orders of magnitude slower than the reaction  in Figure~\ref{fig2a}.
For the reaction  in Figure~\ref{fig2a},  $\logten k=6.43$, $ \logten \kpathway=6.62$, $ \logten \kssa=6.75$,   $|\hat{\statespace}|=4.3 \times 10^5$,  the computation time of pathway elaboration  is  $1.4 \times 10^5 $ s, and the computation time of SSA is  $3.9 \times 10^5$ s.
For the reaction  in Figure~\ref{fig2b},   $\logten k=3.17$, $ \logten \kpathway=3.59$,   $|\hat{\statespace}|=7 \times 10^5$, the computation time of pathway elaboration  is  $2.7 \times 10^5 $ s, and SSA is not feasible within  $1 \times 10^6$ s. 

In Figures~\ref{fig2c}-\ref{fig2e} and ~\ref{fig2g}-\ref{fig2i}, we illustrate different properties of  the   truncated CTMCs for the  reactions   in Figures~\ref{fig2a} and~\ref{fig2b}, respectively.  Comparing Figure~\ref{fig2c} with Figure~\ref{fig2g}, we see that  many states are sampled midway  in Figure~\ref{fig2g} due to the mismatch. In Figures~\ref{fig2d} and~\ref{fig2h}, we compare the energy barrier (increase in  free energy) while moving from the beginning of the x-axis towards the end of the x-axis.  In Figure~\ref{fig2d},  we can see a noticeable energy barrier  in the beginning. However, in Figure~\ref{fig2h},  we can see two   noticeable energy barriers, one  in the beginning and one midway.    Figures~\ref{fig2e} and~\ref{fig2i} show states that are $\delta$-close to the target states.  These figures show that with $\delta$-pruning, states that are further  from the initial states and closer to the target states  will be pruned with smaller values of $\delta$, compared to states that are closer to the initial states and further from the target states. Comparing Figure~\ref{fig2e} with Figure~\ref{fig2i}, the states quickly reach the target states after the first several transitions in Figure~\ref{fig2e} (after the energy  barrier). However, in Figure~\ref{fig2i}, the states do not  quickly reach the target states until after the second energy barrier.  Figure~\ref{fig2f} and~\ref{fig2j} show the free energy landscape and some of the secondary structures for a random path from an initial state to a target state for the reactions in Figures~\ref{fig2a} and~\ref{fig2b}, respectively.   For the reaction in Figure~\ref{fig2a}, the barrier is near the first transition. For the reaction in Figure~\ref{fig2b}, there is a noticeable barrier after several base pairs form between the invader and the substrate, presumably near the mismatch.

 \subsection{Mean First Passage Time  and Reaction Rate Constant Estimation}
\label{exp-eval}
 To evaluate  the estimations of pathway elaboration, we compare its estimations with  estimations obtained from SSA  for  the reactions in $ \mathcal{D}_\text{train}$.     Note that for many of these reactions the size of the state space is exponentially large in the length of the strands. Therefore, exact matrix equations is not possible for them. Instead we use SSA since  it can  generate statistically correct trajectories.   We also compare the wall-clock computation time of pathway elaboration with  SSA.
  \begin{figure}[t]
\subfloat[]{ \includegraphics[width=0.49\textwidth]{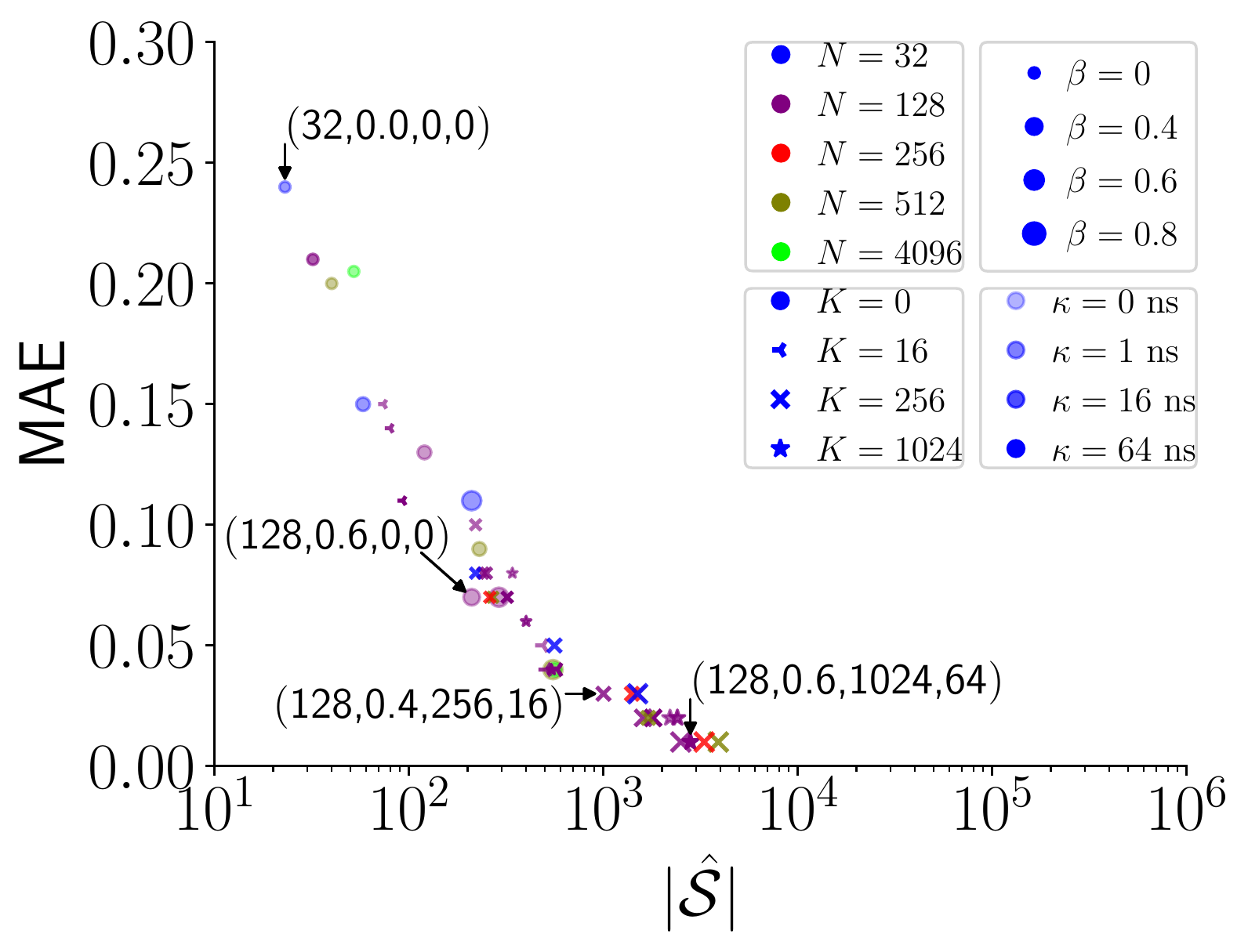}}
 \subfloat[]{ \includegraphics[width=0.49\textwidth]{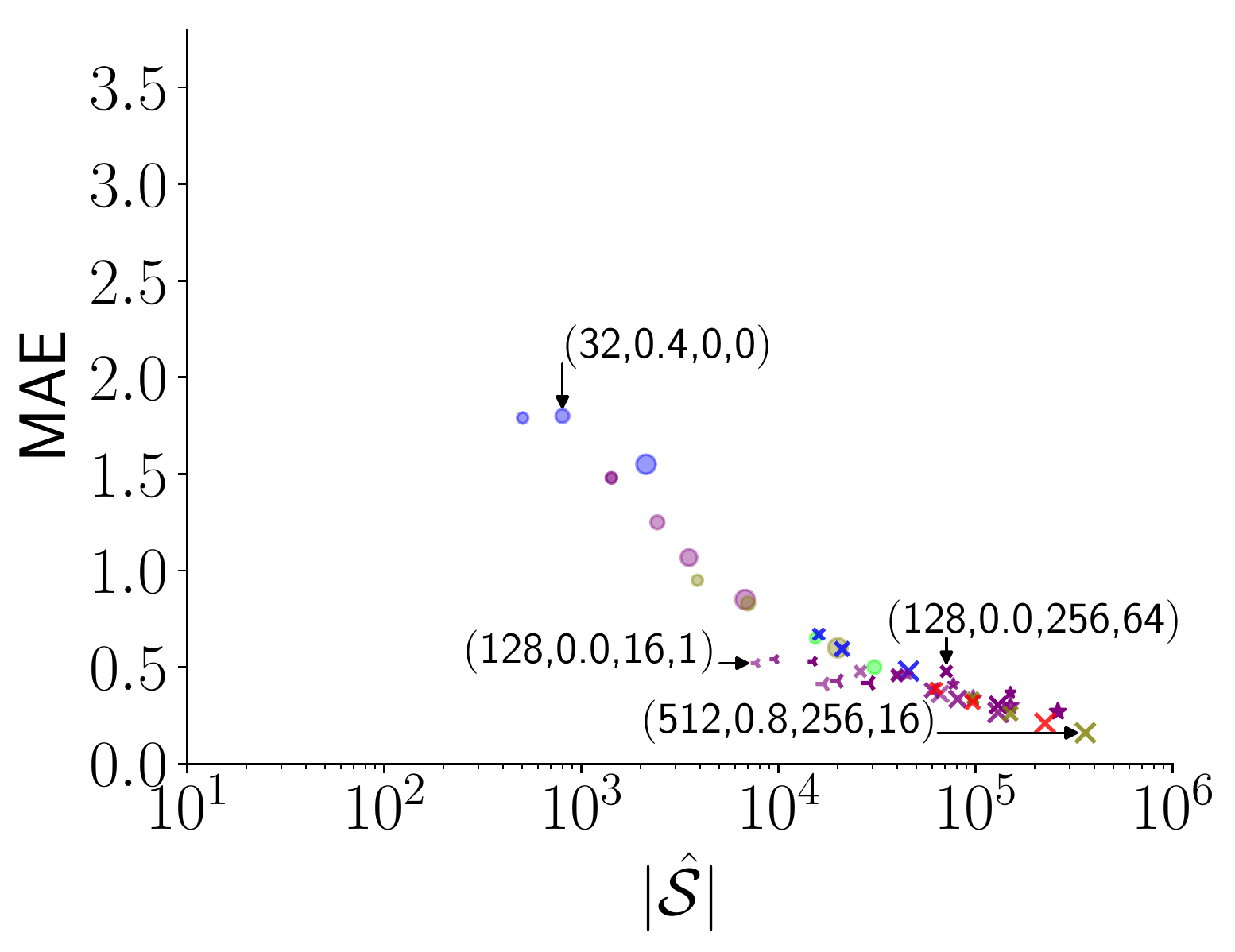}}\\
\subfloat[]{ \includegraphics[width=0.49\textwidth]{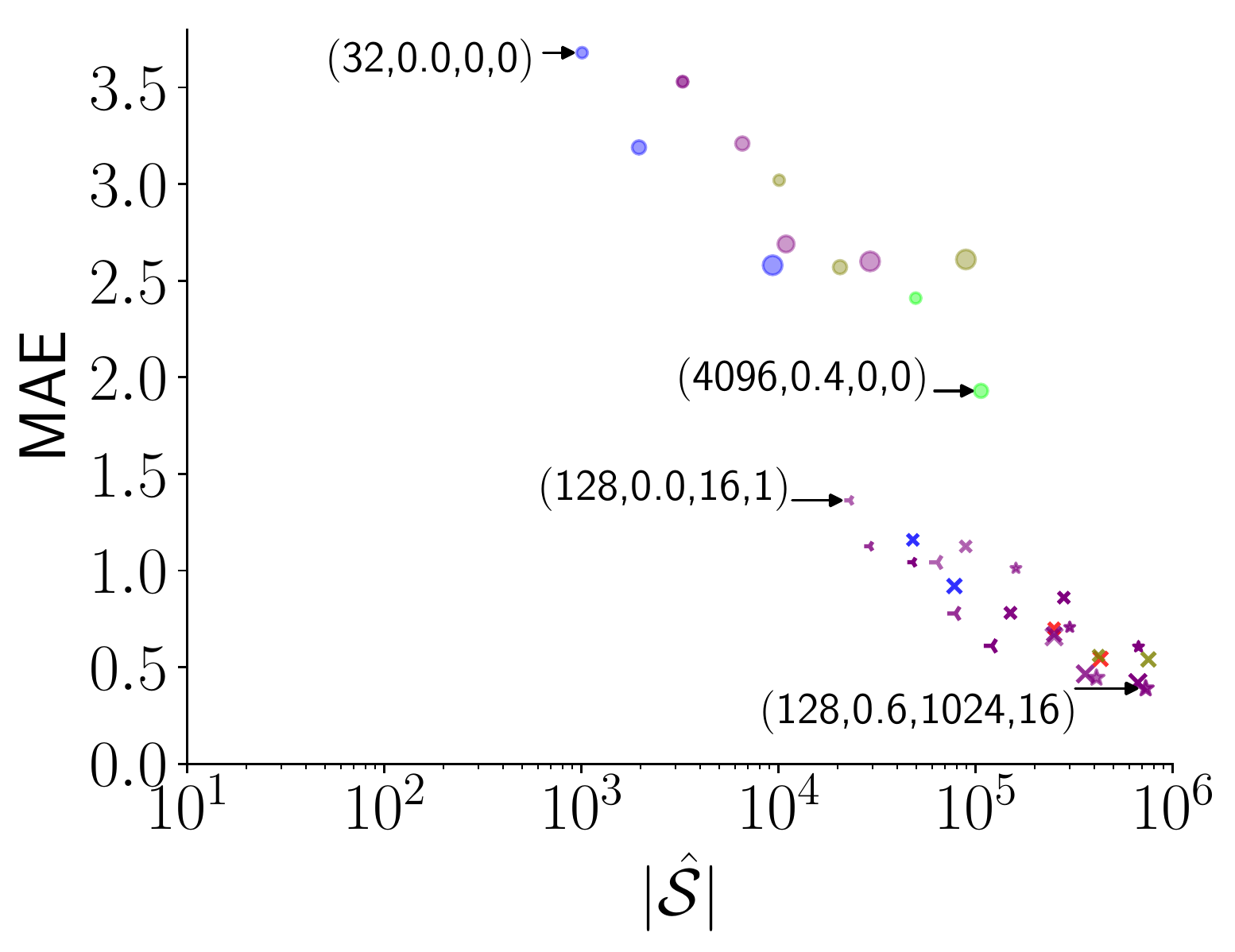}}
  \subfloat[]{ \includegraphics[width=0.49\textwidth]{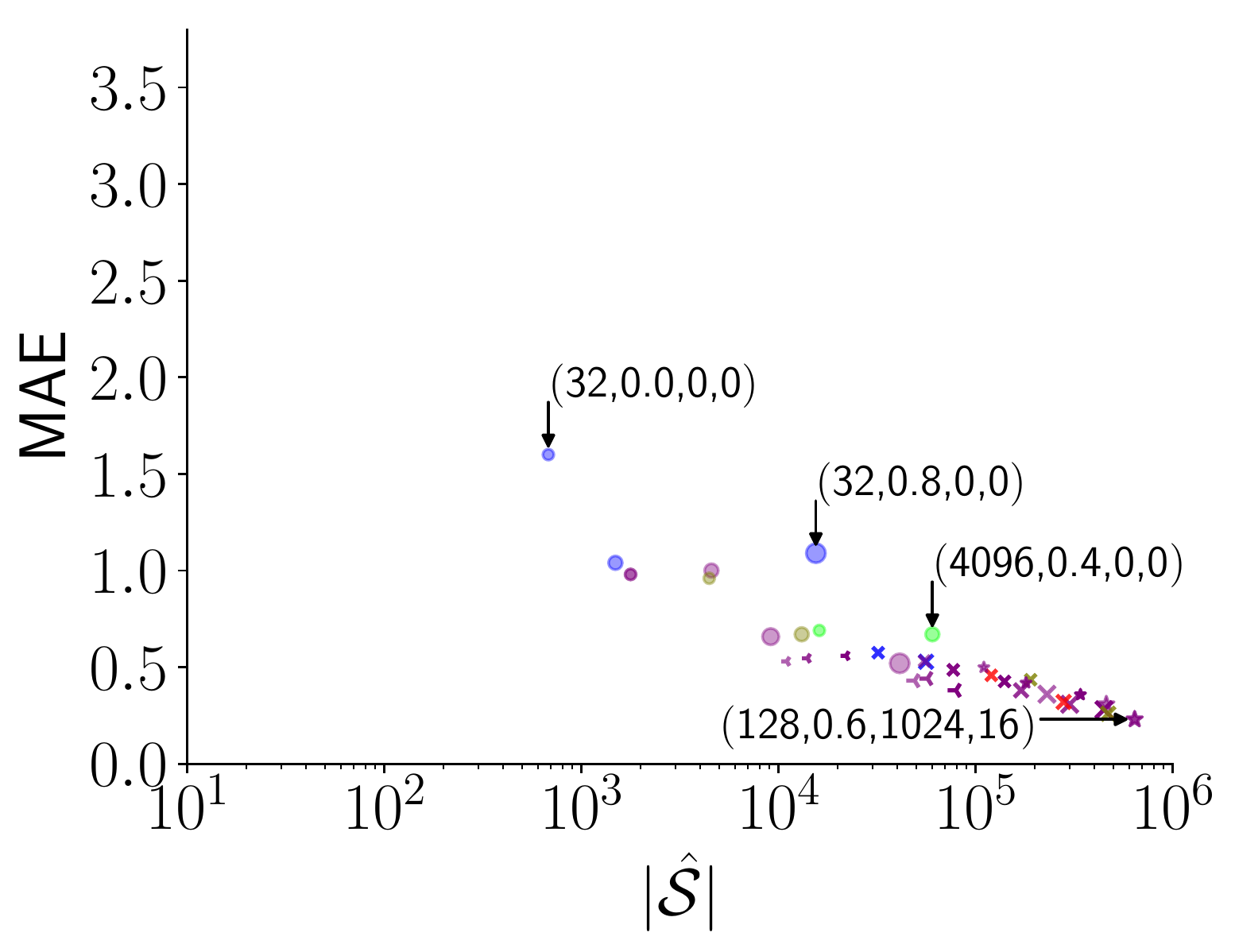}} 
 \caption{The MAE of pathway elaboration with SSA versus $|\hat{\statespace|}$  for  different values of $N$, $\beta$, $K$ and $\pathwaytime$.   (\textbf{a}) datasets No. 1,2, and 3,  (\textbf{b})  dataset No. 4, (\textbf{c})  dataset No. 5,  and (\textbf{d}) dataset No. 6.  The annotated values on the figures correspond to  $N$, $\beta$, $K$, and $\kappa$, respectively.  } 
\label{figure-maevssize}
\end{figure}

 We evaluate the estimations of pathway elaboration based on the mean absolute error (MAE) with SSA, which is defined over a dataset $\mathcal{D}$   as
\begin{equation}
  \text{MAE}=   \frac{1}{|\mathcal{D}|}\sum_{r\in \mathcal{D}}  |\logten  \hat{\tau}_\text{SSA}^r - \logten \hat{\tau}_\text{PE}^r| =   \frac{1}{|\mathcal{D}|}\sum_{r\in \mathcal{D}}  | \logten \hat{k}_\text{SSA}^r - \logten \hat{k}_\text{PE}^r|,
  \label{MAEequation}
 \end{equation}
 where $\hat{\tau}_\text{PE}^r$ and $\hat{\tau}_\text{SSA}^r$ are the estimated MFPTs of SSA and pathway elaboration for reaction $r$, respectively,  and   $\hat{k}_\text{SSA}^r$ and $\hat{k}_\text{PE}^r$ are   the estimated  reaction rate constants of SSA and pathway elaboration for reaction $r$, respectively.  The equality follows from \eqns~ \ref{rate-mfpt-unimolecular} and~\ref{rate-mfpt-bimolecular}.  We use $\logten$ differences since the  reactions rate constants cover many orders of magnitude.  We use  the MAE as our evaluation metric since it is conceptually easy to understand. For example, here,  an MAE of $1$ means on average the predictions are off by a factor of 10.  In the rest of this subsection, we first look at the trade-off between the MAEs and the size of the truncated state space set  $\hat{\statespace}$, with regards to different parameter settings of the pathway elaboration method. Then we look  at the trade-off between the MAE and the computation time.

 \subsubsection{MAE of Pathway Elaboration with SSA  versus $|\hat{\statespace}|$} Figure~\ref{figure-maevssize} shows the MAE  of pathway elaboration with SSA versus  $|\hat{\statespace}|$ of pathway elaboration for different configurations of the   $N$, $\beta$,  $K$, and  $\pathwaytime$ parameters.  Figure  \ref{figure-changingNandBeta} and \ref{figure-changingKandKappa} from the \hyperref[appendix]{Appendix} represent  Figure~\ref{figure-maevssize} by varying only two parameters at a time. The figures show that generally as  $N$ and $\beta$ increase, the MAE decreases.  This is because for a fixed $N$ as $\beta \to 1$ the ensemble of paths will be generated by SSA.     As  $N \to \infty$,  the  truncated state space  becomes larger and is more likely to contain the most probable paths from the initial states to the target states.

Comparing  the MAE of configurations where $K=0$ and $\kappa=0$ with other settings where $K>0$ and $\kappa>0$, shows that the elaboration step helps reduce  the MAE  (in the Appendix, compare Figures~\ref{figure-changingNandBeta1}-\ref{figure-changingNandBeta4} with Figures~\ref{figure-changingNandBeta9}-\ref{figure-changingNandBeta12}).  Particularly, the elaboration step is useful for   dataset No. 4, helix  association from \citet{zhang2018predicting} where intra-strand base pairs can form  before completing  hybridization.  
The plots show that the elaboration step is more useful when $\beta$ is small (in the Appendix, compare Figures~\ref{figure-changingKandKappa1}-\ref{figure-changingKandKappa4} with Figures~\ref{figure-changingKandKappa9}-\ref{figure-changingKandKappa12}). This could be because elaboration helps find rate determining states that were not explored due to the biased sampling. When $\beta \to 1$  the pathway elaboration method will perform as SSA and rate determining states can be found without elaboration.  

    Furthermore, the figures show that as $K$ increases, the MAE decreases. However,   with a large value for $\pathwaytime$  and a small value of  $K$ the performance could be diminished (such as in Figure~\ref{figure-changingKandKappa3} of the Appendix).  In particular, consider that  $K$ and $\kappa$ might involve simulations that go on excursions outside the `main' densely-visited parts of the enumerated state space, and they might even terminate out there.  Such excursions might very well introduce significant local minima into the enumerated state space - even when no significant local minima exist in the original full state space.  For example, consider an excursion that goes off-path down a wide slope, perhaps toward the target state.  If it terminates before reaching a target state, then a hypothetical simulation in the enumerated state space could get stuck, needing to climb back up the slope to the point where the excursion began.  The expected hitting time in the enumerated state space will account for such wasted time, thus leading to an over estimation of the MFPT.   Therefore, $\pathwaytime$ should be tuned  with respect to $K$.  
      
\begin{figure*}

\begin{minipage}{1\textwidth}
	\captionof{table}{\label{kssavskpathwaytable} Pathway elaboration ($N=128$, $\beta=0.6$, $K=256$,  $\pathwaytime=16$ ns) versus SSA. The \textit{mean} statistics are averaged over the `\# of reactions'. Also, the pathway elaboration experiments  are repeated three times and their mean is calculated.  MAE refers to the mean absolute error of pathway elaboration with SSA (\eqn~\ref{MAEequation}). $|\hat{\statespace}|$ is the size of  the truncated state space.  See  Figure~\ref{kssavskpathwayfigure} for an illustration of individual reaction predictions. } 
			\centering 		\scalebox{0.95}{		\begin{tabular}{C{1cm}C{1cm}C{1.cm}C{2.2cm}C{2.2cm}C{2.2cm}C{2.2cm}}\hline Dataset No. & \# of reactions & $\text{MAE} $&  Mean $  |\hat{\statespace}|$ for pathway elaboration  &    Mean matrix computation  time (s) for pathway elaboration & Mean computation  time (s) for pathway elaboration& Mean computation    time (s) for SSA  \\
 \hline $1$ &$63$&$0.04$ & $5.7\times10^{2}$&$4.5\times10^{-3}$&$1.0\times10^{3}$ & $2.7\times10^{1}$\\ \hline 
 $2$&$62$&$0.03$ & $1.8\times10^{3}$&$1.5\times10^{-2}$&$1.0\times10^{3}$ & $1.2\times10^{1}$\\ \hline 
 $3$&$39$&$0.04$ & $5.3\times10^{2}$&$6.8\times10^{-3}$&$1.6\times10^{3}$ & $3.8\times10^{3}$\\ \hline 
 $4$ &$43$&$0.29$ & $8.1\times10^{4}$&$3.0\times10^{1}$&$2.1\times10^{4}$ & $4.9\times10^{5}$\\ \hline 
 $5$ &$20$&$0.51$ & $3.8\times10^{5}$&$2.3\times10^{4}$&$1.6\times10^{5}$ & $3.7\times10^{4}$\\ \hline 
 $6$ &$10$&$0.31$ & $3.0\times10^{5}$&$1.3\times10^{3}$&$1.3\times10^{5}$ & $3.8\times10^{5}$\\ \hline 
 All datasets & $237$&$0.13$&$6.0\times10^{4}$ & $2.0\times10^{3}$&$2.4\times10^{4}$& $1.1\times10^{5}$\\ \hline 
		\end{tabular}} 	
\end{minipage}

\begin{minipage}{1\textwidth}
\subfloat[]{\label{ssavspathwayfigunimol}\includegraphics[width=0.25\textwidth]{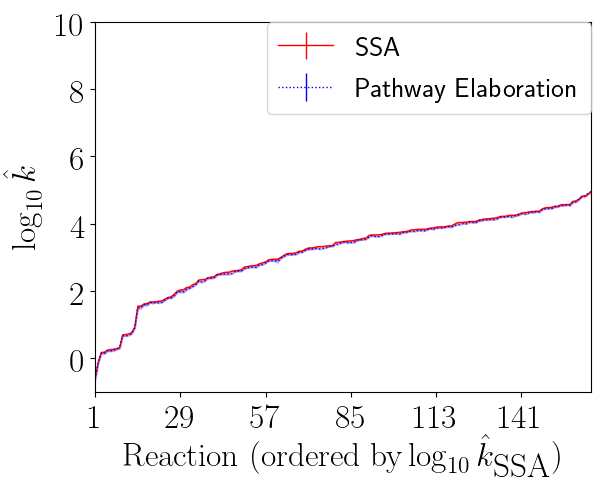}}
\subfloat[]{\label{ssavspathwaybiunimol}\includegraphics[width=0.25\textwidth]{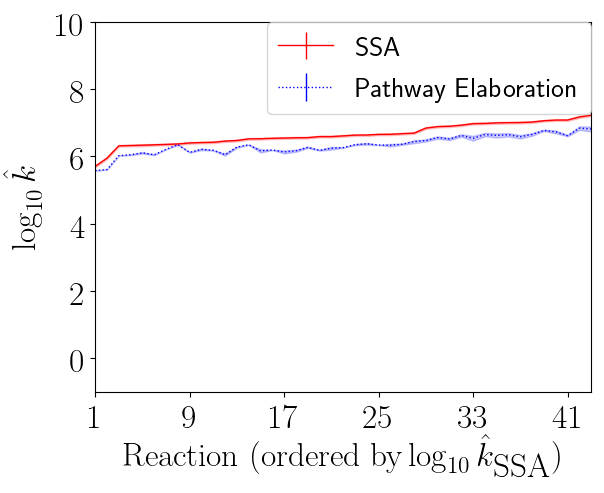}}
\subfloat[]{ \includegraphics[width=0.25\textwidth]{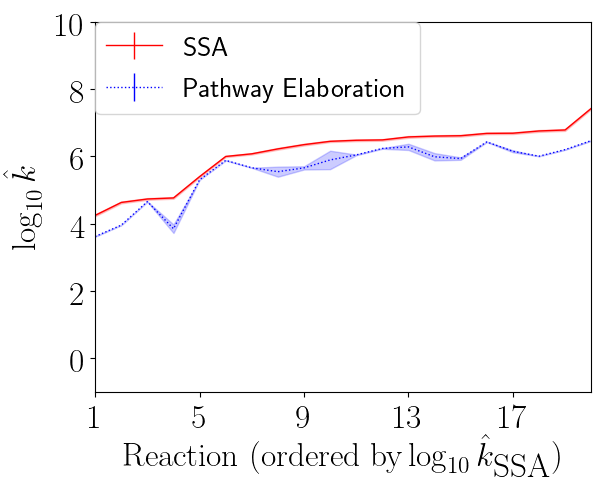}}  
\subfloat[]{ \includegraphics[width=0.25\textwidth]{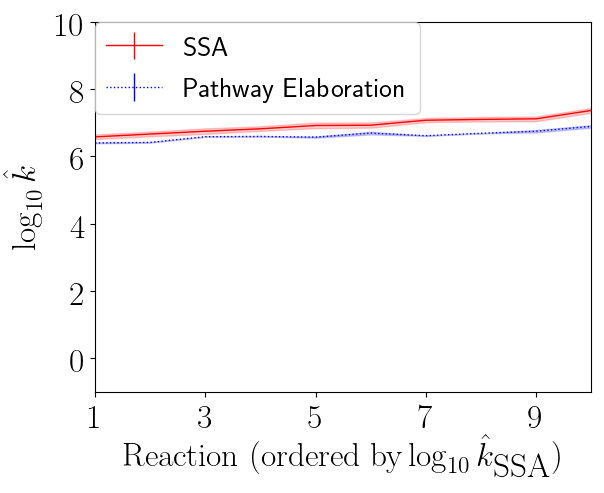}}  
\caption[]{  \label{kssavskpathwayfigure} The  $\logten \kssa$ and $\logten \kpathway$ ($N=128$, $\beta = 0.6$, $K=256$,  $\pathwaytime=16 \text{ ns}$) for  (\textbf{a}) datasets No. 1,2, and 3, and (\textbf{b}) dataset No. 4, (\textbf{c}) dataset No. 5, and  (\textbf{d}) dataset No. 6.  The reactions are ordered along the x-axis by  their predicted $\logten \kssa$.  The pathway elaboration experiments  are repeated three times. For each reaction,   $\logten \kpathway$  is  calculated by the average of the three experiments. The shaded area for pathway elaboration indicates the range (minimum to maximum) of the three experiments.  The shaded area for SSA indicates the 95\% percentile bootstrap of the $\logten \kssa$.  
 }
\end{minipage}\hfill
\end{figure*}

\subsubsection{MAE  of Pathway Elaboration with SSA  versus Computation Time}     Table~\ref{kssavskpathwaytable} illustrates the MAE and the computation time of pathway elaboration  for when $N=128$, $\beta=0.6$,  $K=256$, and $\pathwaytime=16 \text{ 
ns}$  compared with SSA.   We  illustrate this parameter setting  because it provides a good trade-off between accuracy and computational time for the larger reactions.   For the smaller reactions, we could achieve the same MAE with less computational time (by using smaller values for the parameter setting).  Figure~\ref{kssavskpathwayfigure} further shows  the prediction of pathway elaboration  for this parameter setting compared to the prediction of SSA for individual reactions.   In Table~\ref{kssavskpathwaytable}, the MAE for  unimolecular reactions  is smaller than $0.05$, whereas for  bimolecular reactions it is larger than $0.29$. This is because the CTMCs for the  bimolecular reactions in our dataset  are naturally bigger than the CTMCs for the unimolecular reactions in our dataset, and require larger truncated CTMCs. The MAE can be further reduced by changing the parameters (as shown in Figure~\ref{figure-maevssize}).       With our  implementation of pathway elaboration, the computation time of pathway elaboration  for datasets No. 3,  No. 4, and  No. 6   are  2 times, 20 times, 3 times smaller than SSA, respectively.   The computation time of SSA for  datasets No. 1, No. 2, and No. 5  is smaller than the computation time of pathway elaboration. This is because pathway elaboration has some overhead,  and in cases where SSA is already fast it can be slow.  However, as we show in Section~\ref{section-parameterestimation}, even for these reactions, pathway elaboration could still be useful for the rapid evaluation of perturbed parameters.  Also, the  computation time for pathway elaboration  could be significantly improved with more efficient implementations of the method.

\subsubsection{Pathway Elaboration versus other Truncation-Based Approaches}\label{exp-othermethods}
Here we compare pathway elaboration with two other truncation-based approaches. (We do not compare with the probabilistic roadmap method, because we could not find a working implementation of this method, and because of the difficulty of determining appropriate transition rates between non-adjacent states as noted in our related work section.)
The first truncated CTMC model that we include in our comparison uses SSA to sample paths from initial to target states, and builds a CTMC from these states. We call this method SSA-T, where the "T" stands for truncated. Since the sampled paths from SSA are statistically correct, we want to see whether the estimate obtained by pathway elaboration compares well with the unbiased SSA-T estimates. We compare SSA-T with pathway elaboration only on our first two datasets, since SSA-T, being unsuitable for rare events, is too slow to run on our other datasets with the implementation that we used.
 Our second truncated CTMC model uses transition path sampling, and so we call it TPS-T. Our implementation of TPS-T first generates a single path that connects the initial and target states, using SSA. Then a new path is generated by choosing a random state in the most recently generated path, and finding a path from this randomly-chosen state. If the simulated path reaches the initial state before reaching a target state, we continue the simulation until a target state is reached. Moreover, we could define the simulations for TPS-T to be time-limited and include states from these simulations in the truncated CTMC. However, in our experiments stopping simulations early generally result in higher MAE  compared to  continuing simulations until  target states are reached.  In  our experiments, we generate 128 paths in total for both SSA-T and TSP-T. As in Table~\ref{kssavskpathwaytable}, for pathway elaboration, we use $N=128$, $\beta=0.6$,  $K=256$, and $\pathwaytime=16 \text{ 
ns}$. 
 
 Table~\ref{revisiontable} compares the MAE and computation time of CTMCs that are built  with  pathway elaboration versus CTMCs that are built with SSA and TPS. Figure~\ref{revisionfigure} further shows the prediction of these methods compared  for individual reactions. Datasets No. 1 and 2 are used in this table which are hairpin opening and closing, respectively.  The MAE of pathway elaboration with SSA ($0.04$ and $0.03$) compares well with the  MAE of SSA-T with SSA ($0.03$ and $0.03$).  However, the MAE and the variance of TPS-T is high because the paths are correlated and depend on the initial path~\cite{singhal2004using}. Increasing the number of simulations would reduce the variance of the predictions. 

For a comparison of these methods with pathway elaboration regarding computation time,
we have adapted our code for pathway elaboration to implement these methods.   In our experiments, the computation time of pathway elaboration is smaller than both SSA-T and TPS-T and the computation time of TPS-T is smaller than the computation time of SSA-T. 

\begin{figure*}

\begin{minipage}{1\textwidth}
	\captionof{table}{\label{revisiontable}   Building truncated CTMCs with pathway elaboration  versus building truncated CTMCs with SSA (which we call SSA-T)  and TPS (which we call TPS-T).  For  pathway elaboration, $N=128$, $\beta = 0.6$, $K=256$, and $\pathwaytime=16 \text{ ns}$. For SSA-T and TPS-T, 128 successful simulations are used. The \textit{mean} statistics are averaged over the `\# of reactions'.  Also, the experiments for each truncation-based approach is  repeated three times and their mean is calculated.  MAE refers to the mean absolute error of a method with SSA. $|\hat{\statespace}|$ is the size of  the truncated state space.  The  mean matrix computation time for all methods is less than 0.1 (s).  See  Figure~\ref{revisionfigure} for an illustration of individual reaction predictions. }
			\centering 		\scalebox{0.95}{	
\begin{tabular}[b]{C{1cm}C{1cm}C{3cm}C{0.8cm}C{1.3cm}C{1.7cm}C{1.6cm}}\hline   Dataset No. & \# of reactions & Method& $\text{MAE} $&  Mean $  |\hat{\statespace}|$    & Mean computation  time (s)  \\ \hline 
\multirow{3}{*}{1} & \multirow{3}{*}{63} & Pathway elaboration & $0.04$ & $5.7\times10^{2}$&$1.0\times10^{3}$ \\  
                   &                     & SSA-T               & $0.03$ & $4.0\times10^{2}$&  $1.4\times10^{5}$    \\  
                   &                     & TPS-T               & $0.18$ & $1.6\times10^{2}$ &$2.0\times10^{4}$   \\ \hline 
\multirow{3}{*}{2} & \multirow{3}{*}{62} & Pathway elaboration & $0.03$ &  $1.8\times10^{3}$&$1.0\times10^{3}$  \\  
                   &                     & SSA-T               & $0.03$ &  $1.7\times10^{3}$&$1.3\times10^{4}$   \\ 
                   &                     & TPS-T               & $0.34$ &  $3.2\times10^{2}$&$1.7\times10^{3}$ \\\hline 
		\end{tabular}} 	
\end{minipage} 

\begin{minipage}{1\textwidth}
\centering
 \includegraphics[width=0.5\textwidth]{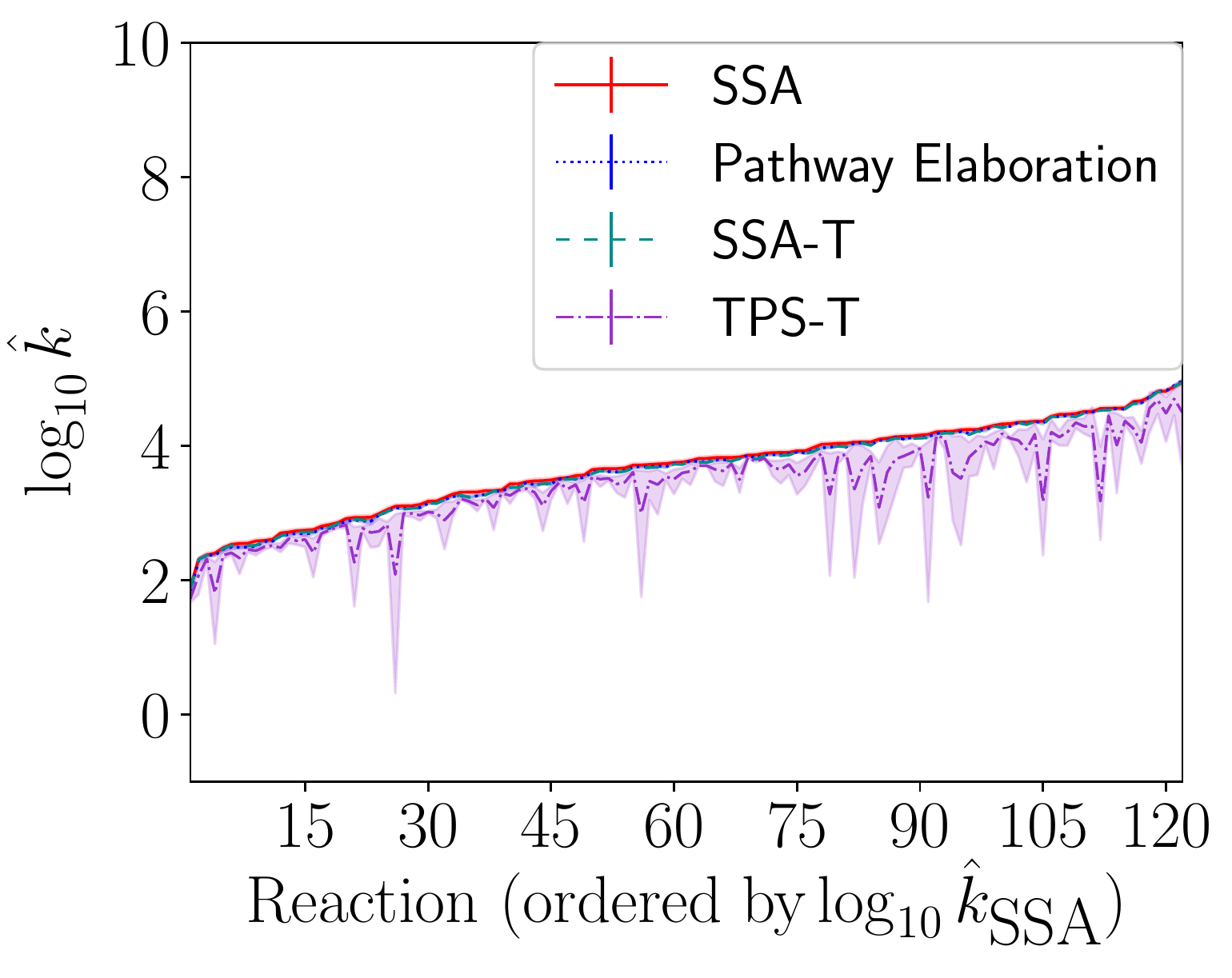}
\caption[]{  \label{revisionfigure}   The  $\logten \hat{k}$ of  SSA, pathway elaboration, SSA-T, and TPS-T  for datasets No. 1 and 2.  The reactions are ordered along the x-axis by  their predicted $\logten \kssa$.   The experiments for each truncation-based approach is repeated three times, where for each reaction,   $\logten \hat{k}$  is  calculated by the average of the three experiments. The shaded area for each truncation-based approach indicates the range (minimum to maximum) of its three experiments.  The shaded area  for SSA indicates the 95\% percentile bootstrap of the $\logten \kssa$. See Table~\ref{revisiontable} for parameter settings and mean statistics. }
\end{minipage}\hfill
\end{figure*}

\subsubsection{$\delta$-Pruning} Figure~\ref{deltapruning} shows how $\delta$-pruning affects the quality of the $\logten$ reaction rate constant estimates, the size of the state spaces, and the  computation time of solving the matrix equations, for  dataset No. 6.  The MFPT estimates satisfy the bound given by \eqn~\ref{deltapruningbound} whilst  $\delta$-pruning   reduces the computation time for solving the matrix equations  by an order of magnitude for $\delta=0.6$. Using larger values of $\delta$ we can further decrease the computation time. If we reuse the CTMCs many times, such as in parameter estimation, $\delta$-pruning could help reduce computation time significantly.

\begin{figure}
\center
 \subfloat[]{ \includegraphics[width=0.32\textwidth]{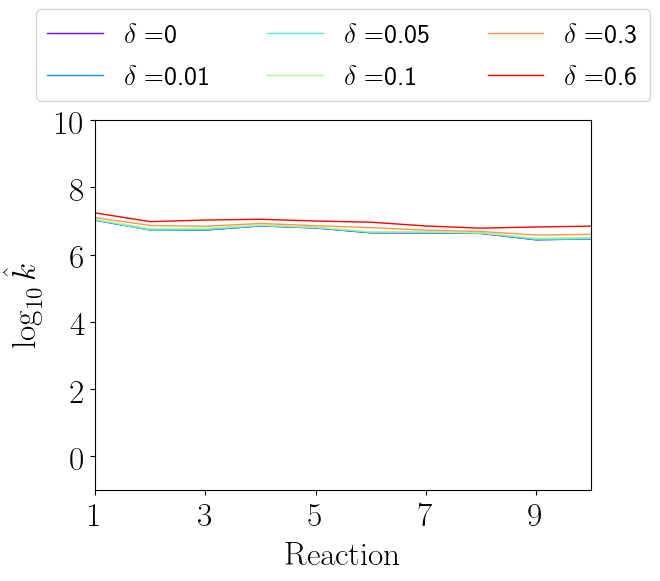}}
  \subfloat[]{ \includegraphics[width=0.32\textwidth]{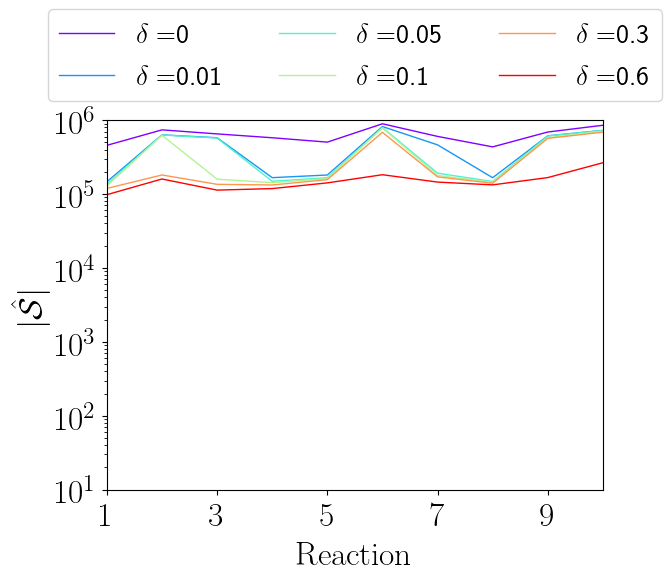}}
 \subfloat[]{ \includegraphics[width=0.32\textwidth]{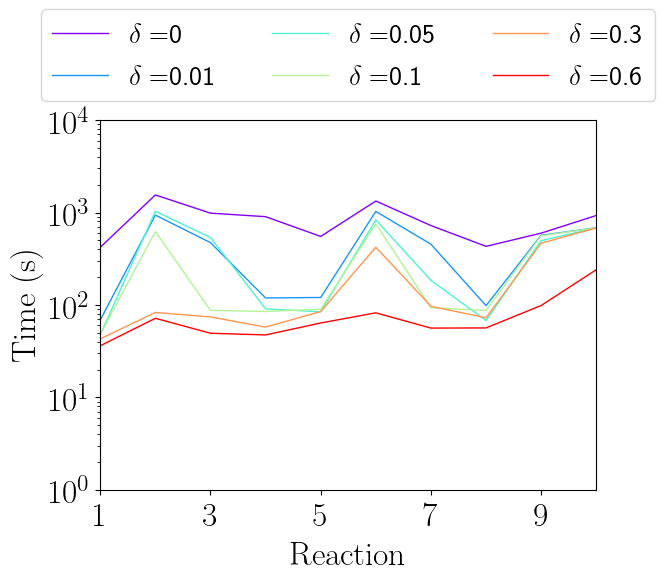}} 
\caption[]{\label{deltapruning}The effect of $\delta$-pruning with different values of $\delta$ on truncated  CTMCs that are built with  pathway elaboration  ($N=128$, $\beta=0.6 $, $K=1024$, $\pathwaytime=16 \text{ ns}$) for  dataset No. 6. $\delta=0$ indicates $\delta$-pruning is not used.   (\textbf{a}) The  $\logten \hat{k}$. (\textbf{b)} The size of the truncated state space $|\hat{\statespace}|$.  (\textbf{c}) The computation time for solving \eqn~\ref{mfptequation}.} 

\end{figure}

\subsection{Parameter Estimation}\label{section-parameterestimation}

\begin{figure}\center
\subfloat[]{ \includegraphics[width=0.35\textwidth]{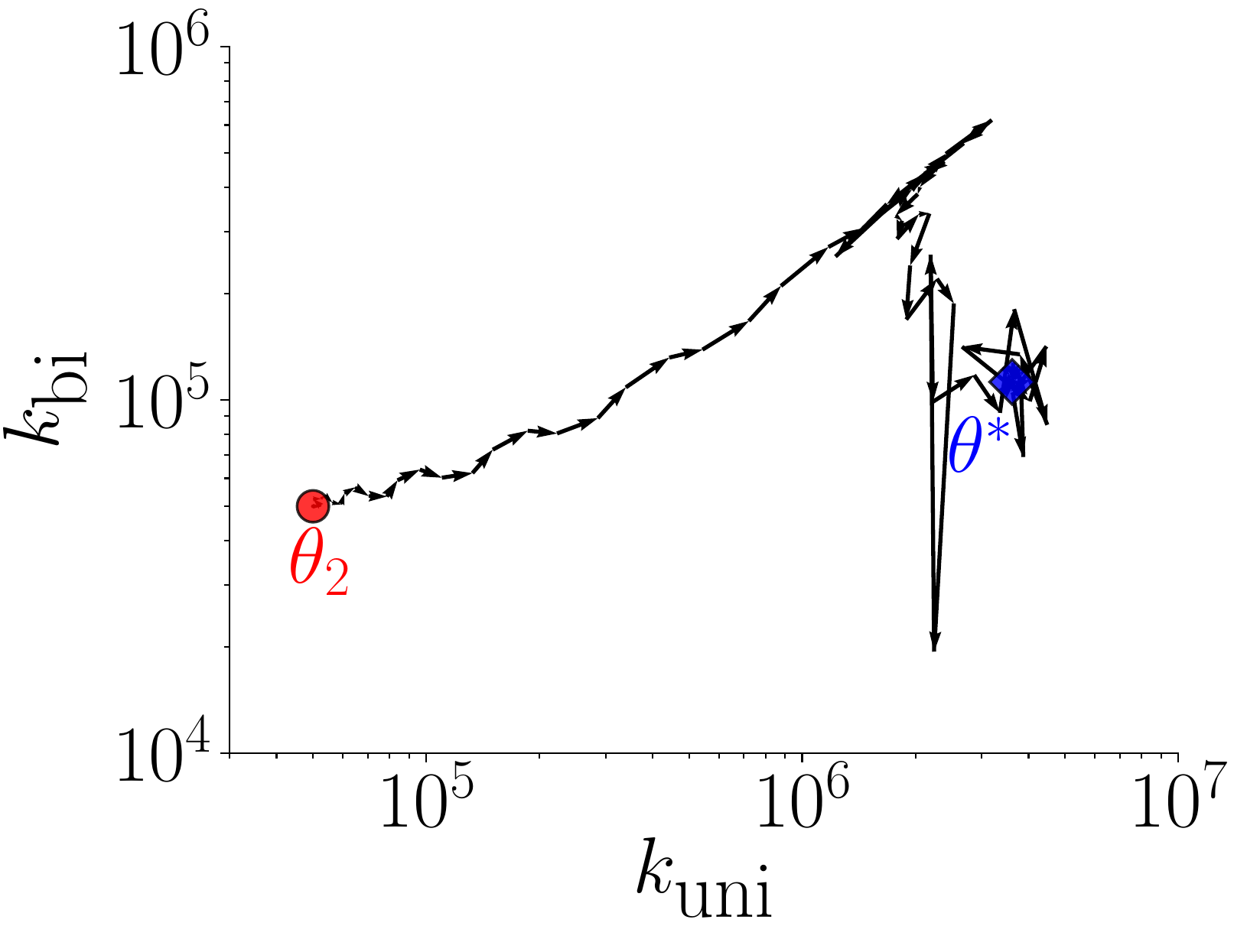}}
\subfloat[]{ \includegraphics[width=0.35\textwidth]{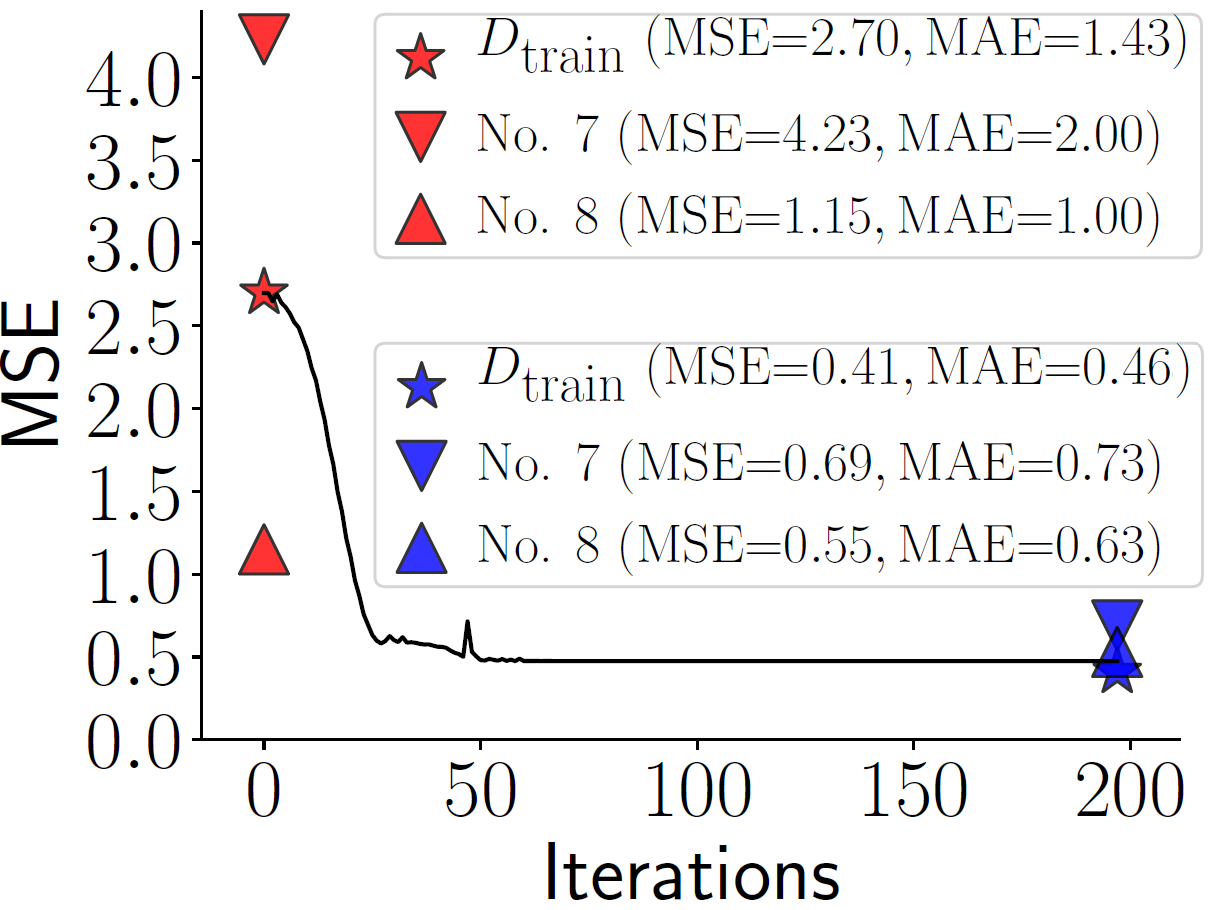}}
 \caption[]{Results of parameter estimation  using  pathway elaboration ($N=128$, $\beta = 0.4$, $K=256$,  $\pathwaytime=16 \text{ ns}$).   (\textbf{a}) The parameters are optimized from an initial simplex of $\theta_2$ and its perturbations to $\theta^*= \{ \kUni \approx 3.61 \times 10^6 \text{ s}^{-1}, \kBi \approx 1.12 \times 10^5 \text{ M}^{-1}\text{ s}^{-1} \} $.   (\textbf{b}) The parameters are optimized using $\mathcal{D}_\text{train}$, shown with a line graph, and evaluated on  dataset No. 7 and No. 8.  The   markers  are annotated with the  MSE and MAE of the datasets when the truncated CTMCs are built from scratch using $\theta_2$ and $\theta^*$.
 } 
\label{figure-parameterestimation}
\end{figure}

In the previous subsections  the underlying parameters of the CTMCs were fixed. Here we assume the parameters of the kinetic model of the CTMCs are  not calibrated and  we use pathway elaboration  to build truncated CTMCs  to  rapidly evaluate  perturbed parameter sets during parameter estimation.   We use the 237 reactions  indicated as  $\mathcal{D}_\text{train}$  in Table~\ref{sourceTable}  as our training set.  We use the 30 rare event reactions  indicated  as  $\mathcal{D}_\text{test}$  in Table~\ref{sourceTable}   to show that given a well-calibrated parameter set for the CTMC model, the pathway elaboration method can estimate  MFPTs and reaction rate constants  of reactions close to their experimental  measurement. 

We seek the parameter set that minimizes the mean squared error (MSE)  as
\begin{equation}
  \begin{split}&
   \theta^* =  \underset{\theta}{\argmin}   \frac{1}{|\mathcal{D}_\text{train}|}\sum_{r \in \mathcal{D}_\text{train}}{ ( \logten \tau^r - \logten  \hat{\tau}_\text{PE}^r(\theta) )^2 }=  \\&  \underset{\theta}{\argmin}   \frac{1}{|\mathcal{D}_\text{train}|}\sum_{r \in \mathcal{D}_\text{train}}{ ( \logten k^r - \logten  \hat{k}_\text{PE}^r(\theta))^2 }, 
  \end{split}
      \label{gmmeq}
 \end{equation} 
which is a common  cost function for regression problems.  The equality follows from \eqns~ \ref{rate-mfpt-unimolecular} and~\ref{rate-mfpt-bimolecular}.   We use the Nelder-Mead optimization algorithm~\cite{nelder1965simplex,virtanen2020scipy} to minimize the MSE.    We  initialize the simplex in the algorithm  with $\theta_2 = \{ \kUni = 5 \times 10^4 \text{ s}^{-1}, \kBi = 5 \times 10^4 \text{ M}^{-1}\text{s}^{-1} \}$ in which we choose arbitrarily and two  perturbed parameter  sets.  Each perturbed parameter set is obtained from  $\theta_2$ by multiplying one of the parameters by $1.05$, which is the default implementation of  the optimization software~\cite{virtanen2020scipy}. For every reaction, we also  initialize the Multistrand kinetic model  with
$\theta_2 $. We  build truncated CTMCs with pathway elaboration ($N=128$, $\beta = 0.4$, $K=256$,  $\pathwaytime=16 \text{ ns}$). Whenever the matrix equation solving time is large (here we consider a time of $120$ s large), we use $\delta$-pruning  (here we use  $\delta$ values of  $0.01-0.6$) to reduce the time.   During the optimization, for a new parameter set  we update the parameters in the kinetic model of the truncated CTMCs and we reuse the truncated CTMC to evaluate the parameter set. Similar to our previous work~\cite{zolaktaf2019efficient}, to reduce the bias and to ensure that
the truncated CTMCs are fair with respect to the optimized parameters, we  can occasionally rebuild  truncated CTMCs  from scratch. 

Although we use the MSE of pathway elaboration with experimental measurements as our cost function in the optimization procedure, the MAE of pathway elaboration with experimental measurements  also  decreases. Figure~\ref{figure-parameterestimation} shows  how  the parameters, the MSE, and the MAE change during optimization. The markers are annotated with the MSE  and the MAE of  $\mathcal{D}_\text{train}$ and datasets No. 7-8   when truncated CTMCs are built from scratch.   The MAE of  $\mathcal{D}_\text{train}$ with the initial parameter set $\theta_2$ is $1.43$. The  optimization finds  $\theta^*= \{ \kUni \approx 3.61 \times 10^6 \text{ s}^{-1}, \kBi \approx 1.12 \times 10^5 \text{ M}^{-1}\text{ s}^{-1} \} $  and reduces the MAE of $\mathcal{D}_\text{train}$ to $0.46$.   The  MAE of dataset No. 7 and dataset No. 8,   which are not used in the optimization, reduce  from  $2.00$ to  $0.73$ and from $1.00$ to $0.63$, respectively.


 Overall, the experiment in this subsection shows  that pathway elaboration enables  MFPT estimation of rare events. It predicts their MFPTs close to their experimental measurements given  an accurately calibrated model for their CTMCs.  Moreover,  it shows that pathway elaboration   enables the rapid evaluation of perturbed parameters and makes feasible tasks such as parameter estimation which benefit from such methods.  On average for  the 30 reactions in  the testing set,  pathway elaboration takes less than two days, whereas SSA  is  not feasible within two weeks.   The entire experiment in Figure~\ref{figure-parameterestimation} takes less than five days parallelized on 40 processors. Note that clearly our optimization procedure could  be improved, for example by using a larger dataset or  a more flexible kinetic model.  However, this experiment is a preliminary study; we leave a  rigorous  study on calibrating  nucleic acid kinetic models with pathway elaboration and possible improvements to future studies.

\section{Discussion}
 Motivated by the problem of predicting nucleic acid kinetics, we  address the problem of estimating MFPTs of rare events in large CTMCs and also  the rapid evaluation of perturbed parameters.      We propose the  pathway elaboration method, which is a time-efficient  probabilistic truncation-based approach for   MFPT estimation in CTMCs.    We conduct computational  experiments on  a wide range of  experimental  measurements to show pathway elaboration  is suitable for predicting nucleic acid kinetics. In summary, our results are promising, but there is still room for improvement.

Using pathway elaboration, in the best possible case, the sampled  region of states and transitions is obtained faster than SSA, but without significant bias in the collected states and transitions. The sampled region may however qualitatively differ from what would be obtained from SSA, which may compromise the  MFPT estimates. Moreover, reusing truncated CTMCs for significantly perturbed parameters could lead to inaccurate estimation of the MFPT in the original CTMC. 
In Section~\ref{section-method},   for exponential decay processes, we introduced a method that could help us quantify the error of the MFPT estimate.  However, it might  be slow in practice.  So how can we efficiently tune these parameters?  Similar to SSA, for a fixed $\beta$ and when $K=0$ and $\kappa=0$,  we could increase $N$ until  the estimated MFPT  stops changing significantly (based on the law of large numbers it will converge).  Note that for $K=0$ and $\kappa=0$   we could compute the MFPT by  computing the average of the  biased paths without solving matrix equations.  As shown in Proposition~\ref{proposition1}, if we set $\beta$  to less than $1/2$, then   biased paths will reach  target states in   expected time that is linear in the distance from initial to target states. For setting $K$, one possibility is to consider the number of neighbors of each state. A reaction where states have a lot of neighbors requires a larger $K$ compared to  a reaction where states have a smaller number of neighbors.   $\kappa$ should be set with respect to $K$. As stated  in Section~\ref{results},   a large value of $\kappa$   along with a small value of $K$ could result in excursions that do not reach  any target state and lead to overestimates of the MFPT.  One could set $\kappa$ to a small value and then increase  $K$ until the MFPT estimate stops changing, and could repeat this process while feasible.

In  the pathway elaboration method, we estimate MFPTs by solving matrix equations.  Thus, its performance  depends on the accuracy and speed of matrix equation  solvers.  For example, applying matrix equation solvers may not be  suitable if the initial states lie very far from the target states,  since the size of the truncated  CTMCs depends on the shortest-path distance between these states. Although solving matrix equations  through  direct  and iterative methods  has progressed,  both theoretically and  practically~\cite{fletcher1976conjugate,virtanen2020scipy,cohen2018solving}, solving  stiff (multiple time scales) or very large equations  could still be problematic in practice.    More stable and faster solvers would allow us to estimate MFPTs for  stiffer and larger truncated CTMCs. 
 Moreover, it might be possible to use fast updates for solving the matrix equations~\cite{brand2006fast,parks2006recycling}. Therefore, if we require to compute MFPT estimates with matrix equations as  we monotonically grow the size of the state  space or for a perturbed parameter set,  the total cost for solving all  the linear systems would be the same cost as solving the final linear system from scratch.

We  might be able to improve the pathway elaboration method to relieve the limitations discussed above. For example,   it might be possible  to use an ensemble of truncated CTMCs   to obtain an unbiased estimate of the MFPT~\cite{georgoulas2017unbiased}.  To avoid  excursions that lead to overestimation of the MFPT in the state elaboration step, we could run the pathway construction step from the last states visited in the state elaboration step. This would also relax the constraint of  having  reversible or detailed balance transitions.  Presumably,  an alternating approach of the two steps would make the approach more flexible.  Moreover, currently we run the state elaboration step from every state of the pathway with the same setting. Efficiently running the state  elaboration step as necessary, could reduce the time to construct the truncated CTMC in addition to the matrix computation time.

Finally, we evaluated the pathway elaboration method for  predicting the MFPT of  nucleic acid kinetics. However, the method is generally applicable to \reversible CTMC models, such as chemical reaction networks~\cite{anderson2011continuous} and  protein folding~\cite{mcgibbon2015efficient}.

\clearpage
\bibliographystyle{imsart-nameyear} 
\bibliography{ref}       

\begin{appendix}
\section{Appendix: The Mean Absolute Error  of the Pathway Elaboration Method for Nucleic Acid Kinetics}
\label{appendix}
\renewcommand{\thefigure}{A\arabic{figure}}

\setcounter{figure}{0}  

 Figures~\ref{figure-changingNandBeta} and~\ref{figure-changingKandKappa} represent  Figure~\ref{figure-maevssize}   by varying only two parameters at a time.  
\begin{figure}
    \includegraphics[width=0.2\textwidth]{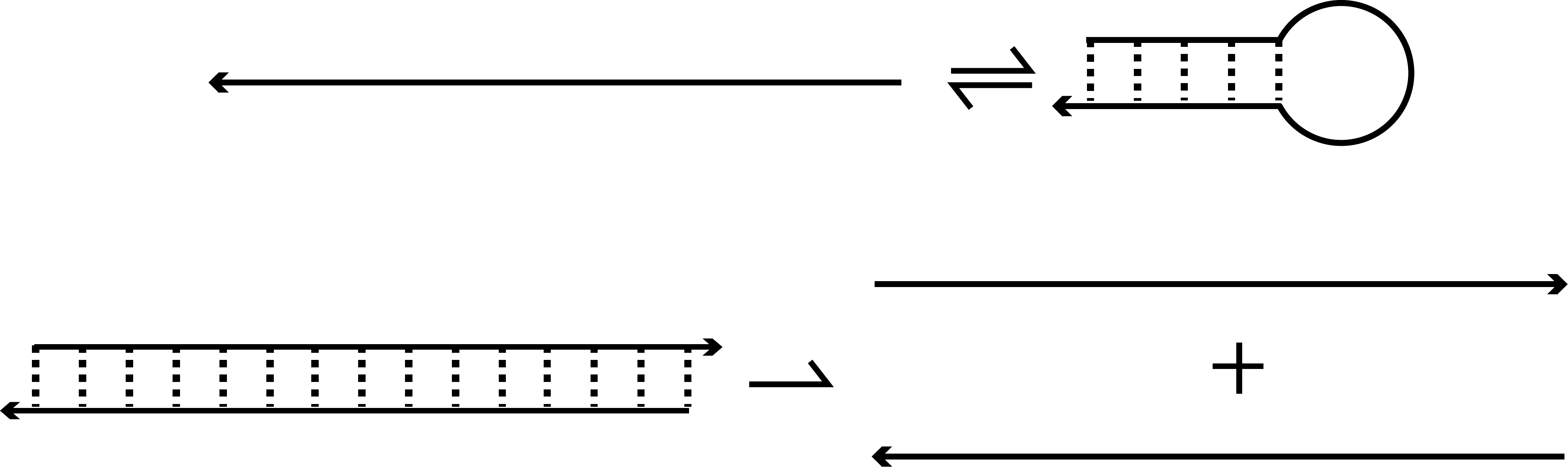}\quad  \quad
 \includegraphics[width=0.2\textwidth]{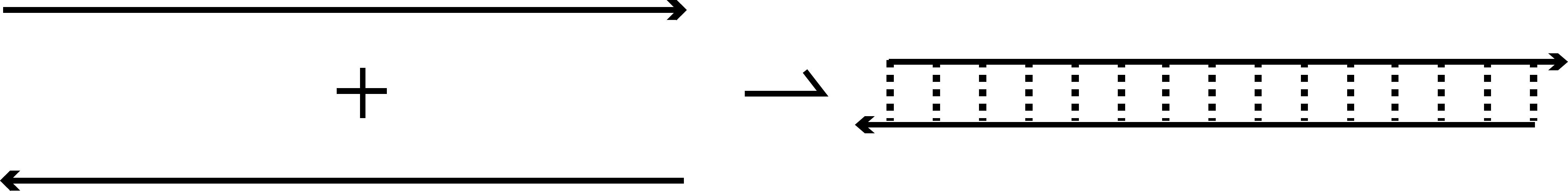} \quad \quad
 \includegraphics[width=0.18\textwidth]{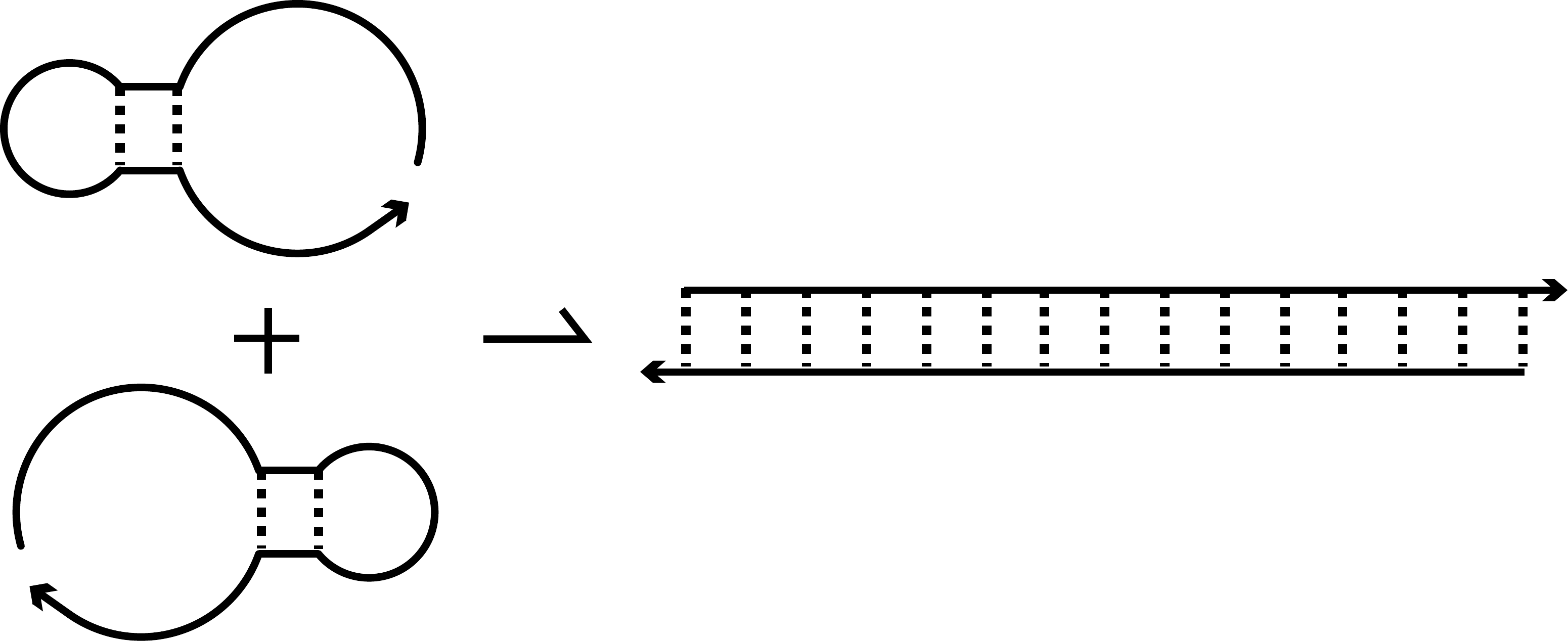} \quad\quad\quad
 \includegraphics[width=0.2\textwidth]{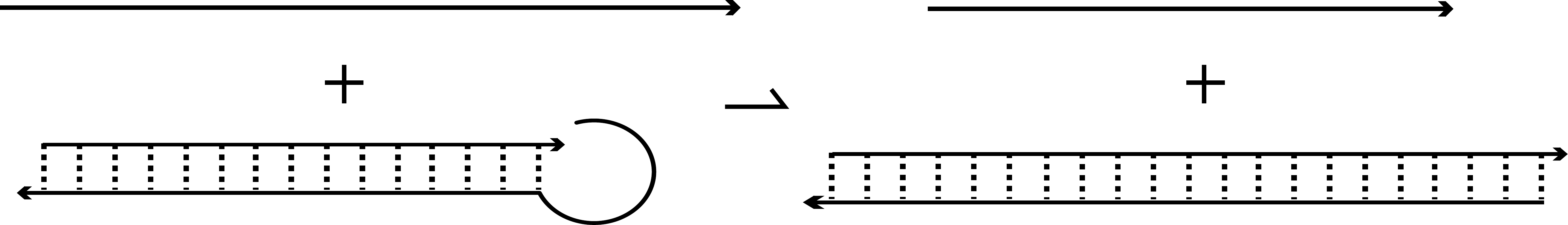}
\\
\vspace{0.5cm}
\hrule
      \text{ In  (\textbf{a}-\textbf{h}), $K=0$ and $\pathwaytime=0$ ns are fixed (the state elaboration  step is not used).}  \\ 
   
\subfloat[]{ \label{figure-changingNandBeta1}\includegraphics[width=0.245\textwidth]{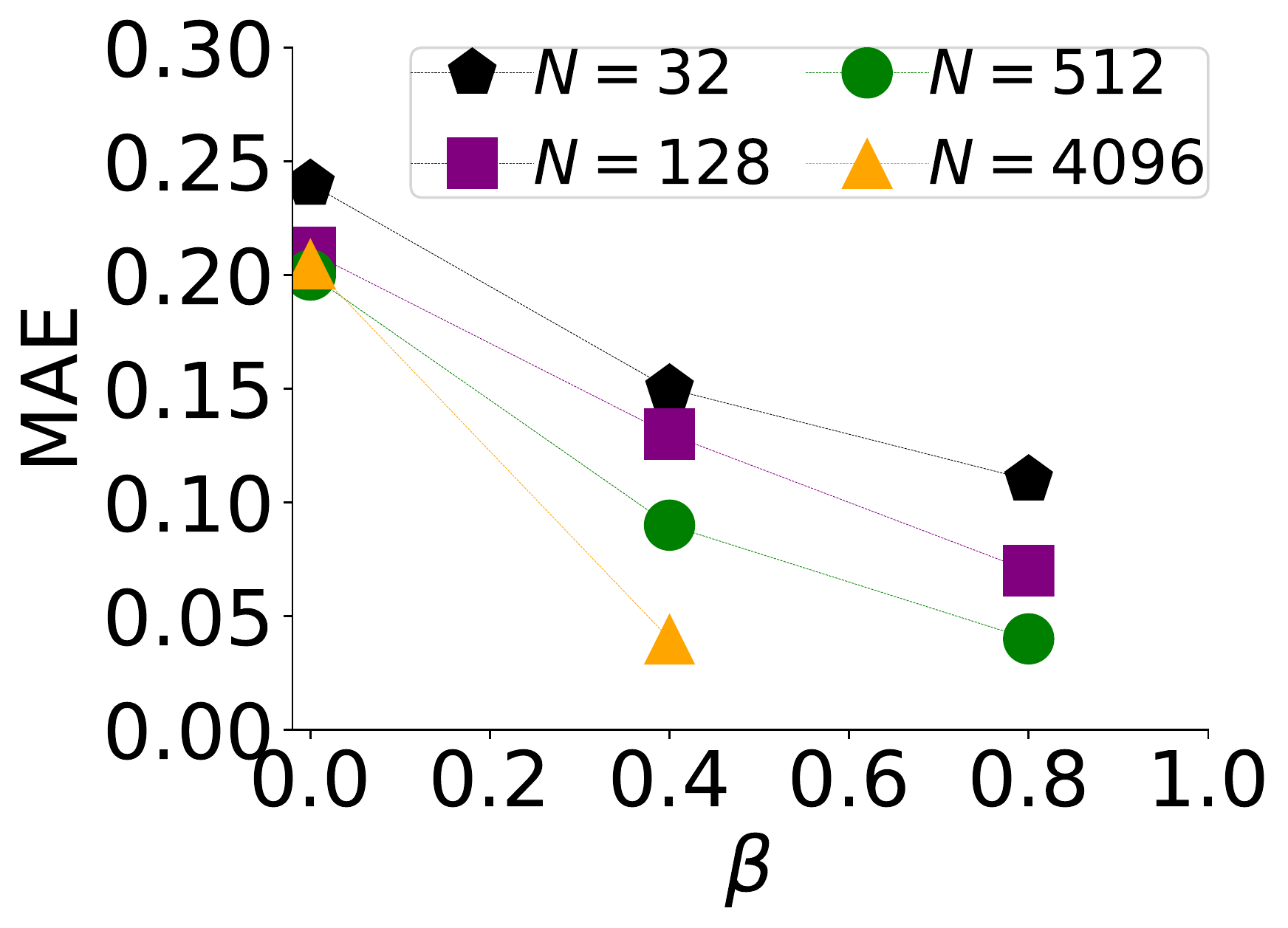}}
 \subfloat[]{\label{figure-changingNandBeta2} \includegraphics[width=0.245\textwidth]{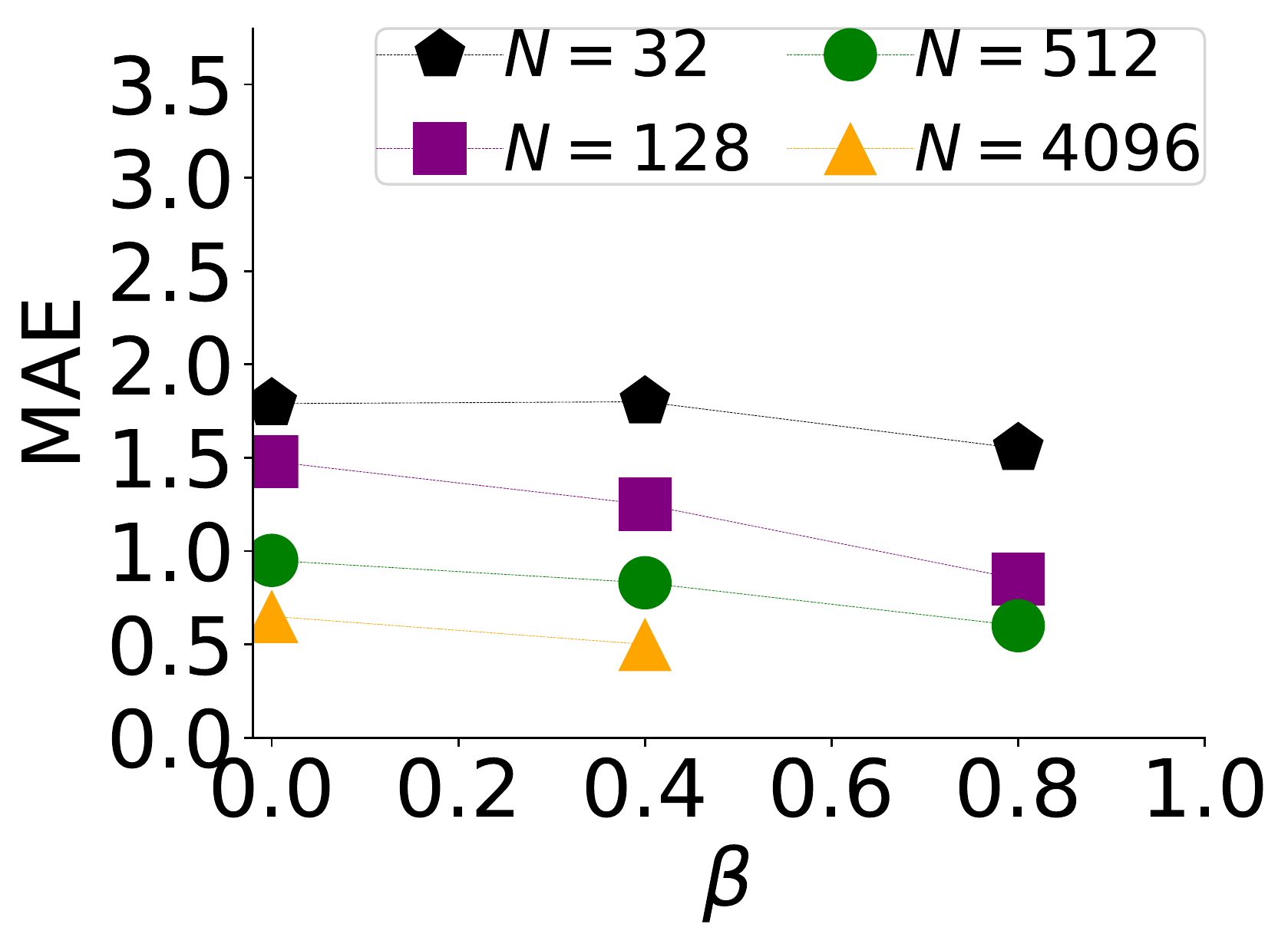}}
\subfloat[]{ \label{figure-changingNandBeta3}\includegraphics[width=0.245\textwidth]{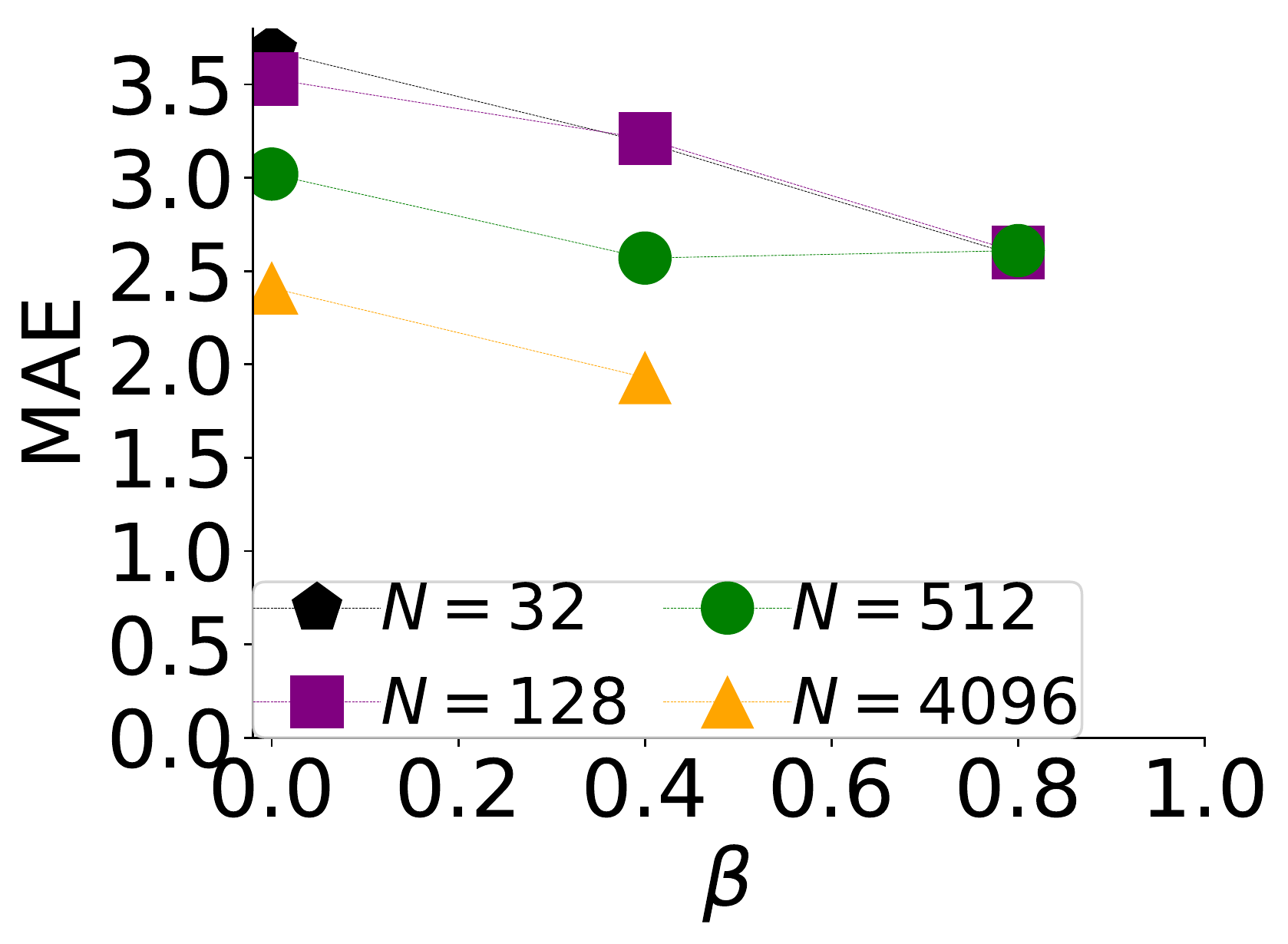}}
  \subfloat[]{\label{figure-changingNandBeta4} \includegraphics[width=0.245\textwidth]{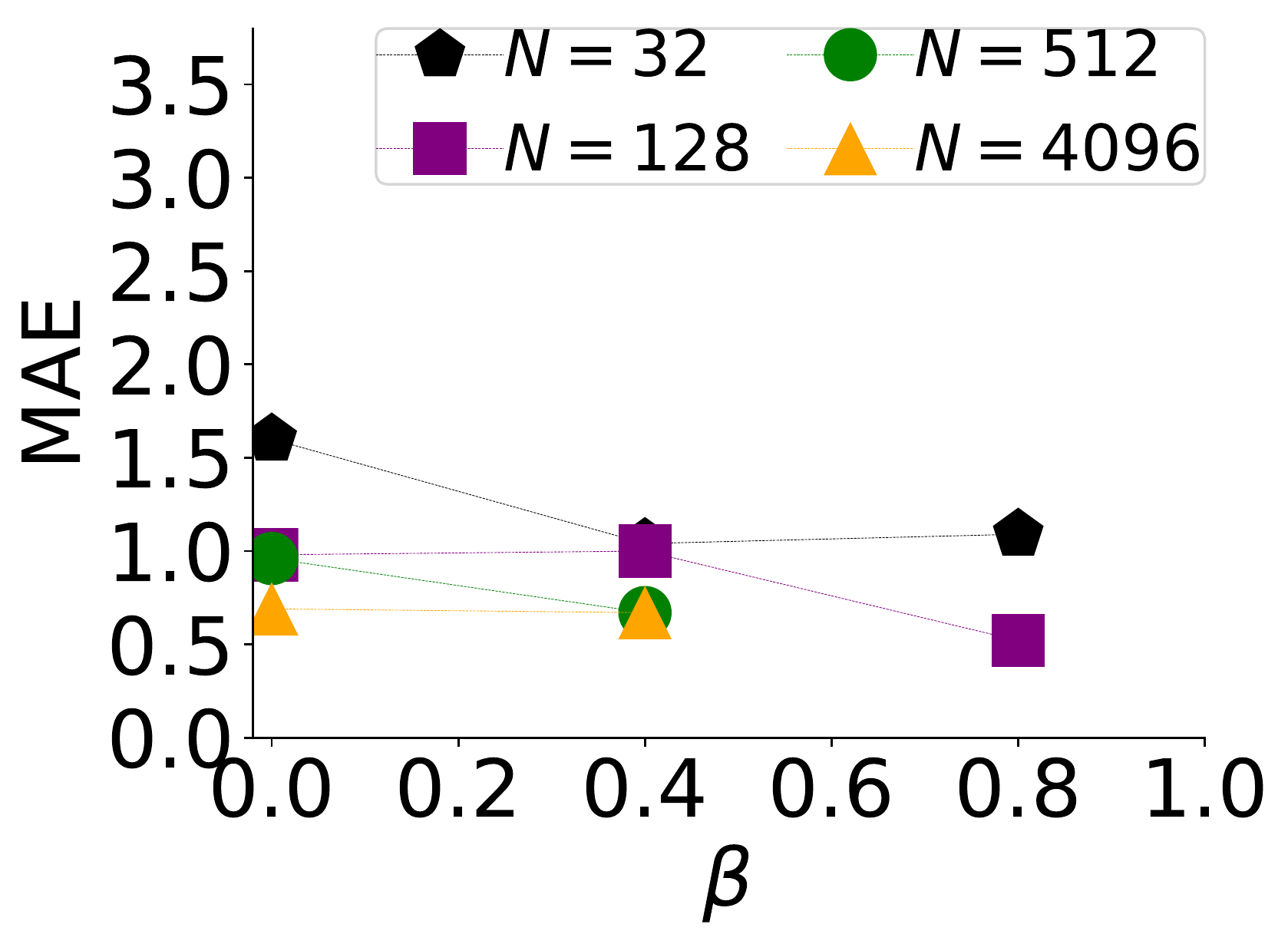}}
     \\ 
    
 \subfloat[]{\label{figure-changingNandBeta5}
 \includegraphics[width=0.245\textwidth]{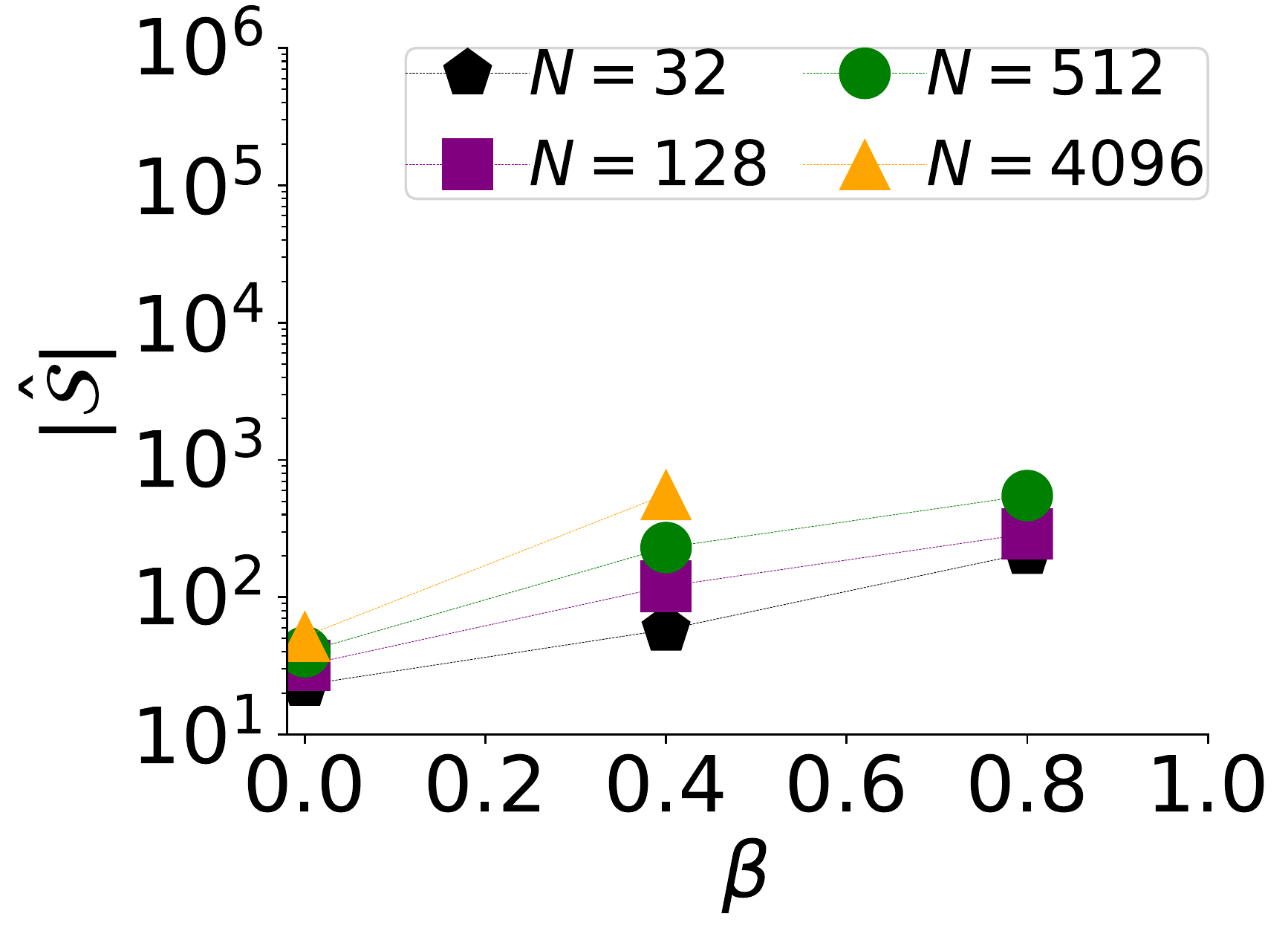}}
 \subfloat[]{ \label{figure-changingNandBeta6}\includegraphics[width=0.245\textwidth]{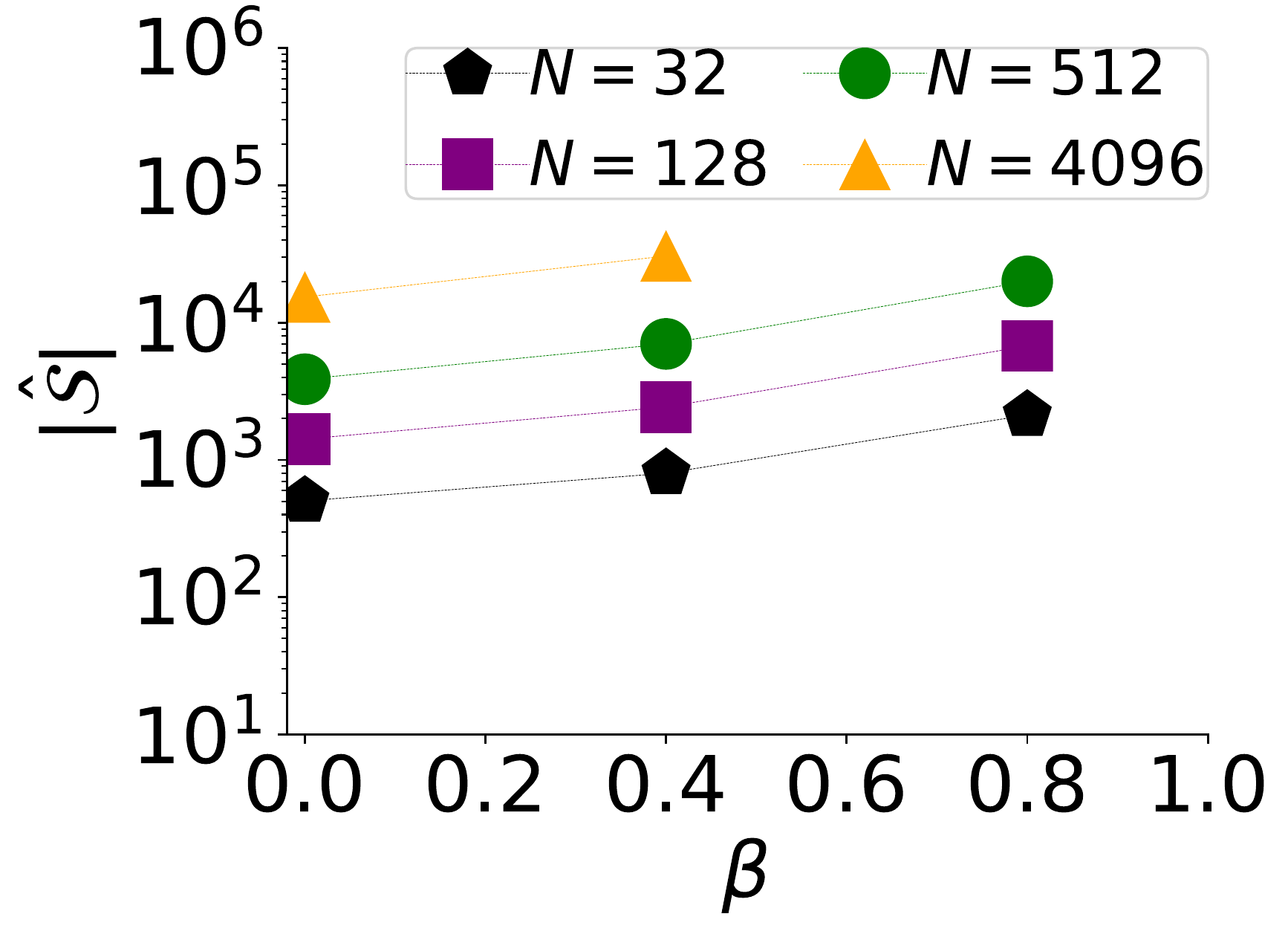}}
\subfloat[]{\label{figure-changingNandBeta7} \includegraphics[width=0.245\textwidth]{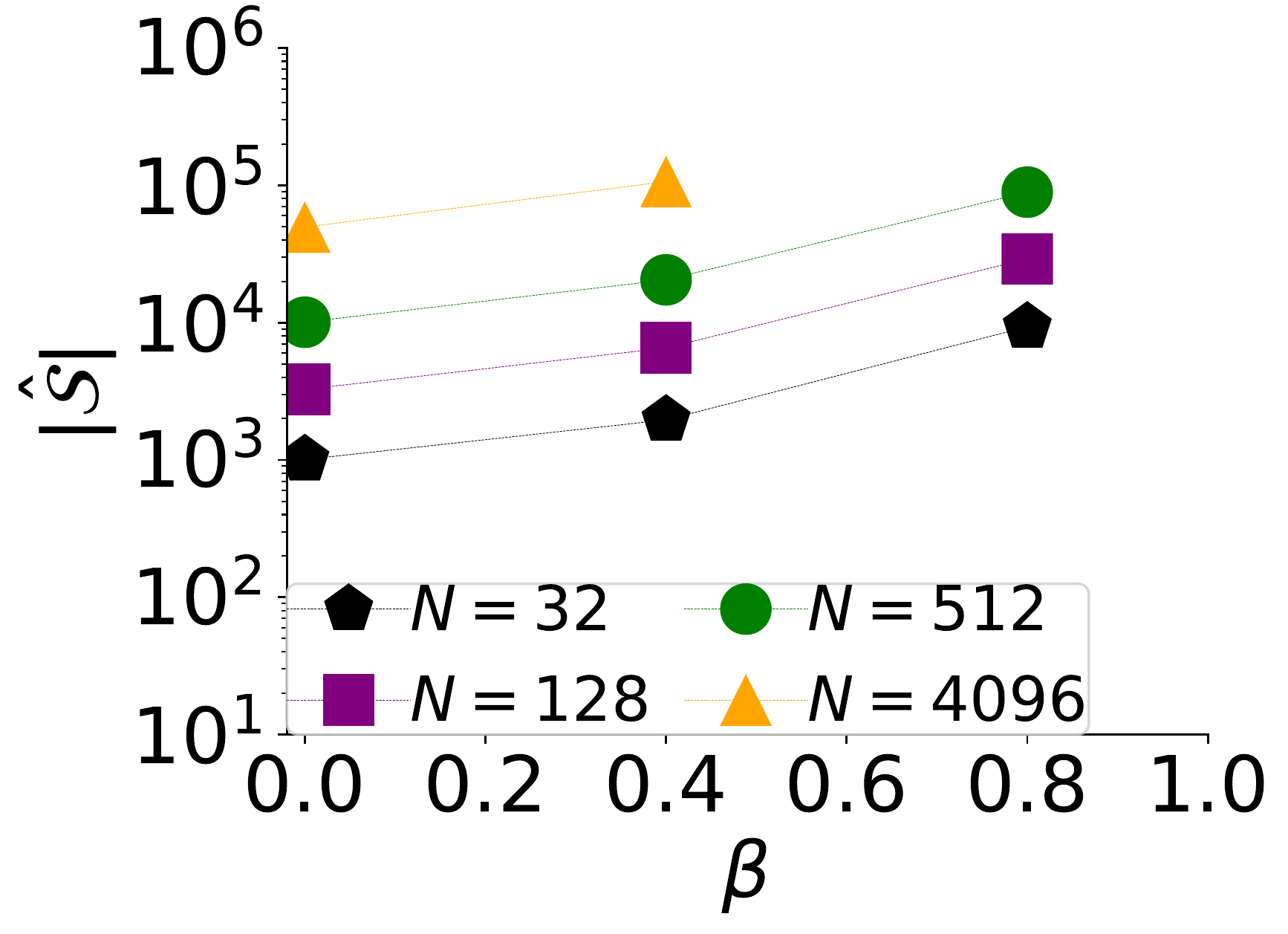}}
  \subfloat[]{ \label{figure-changingNandBeta8}\includegraphics[width=0.245\textwidth]{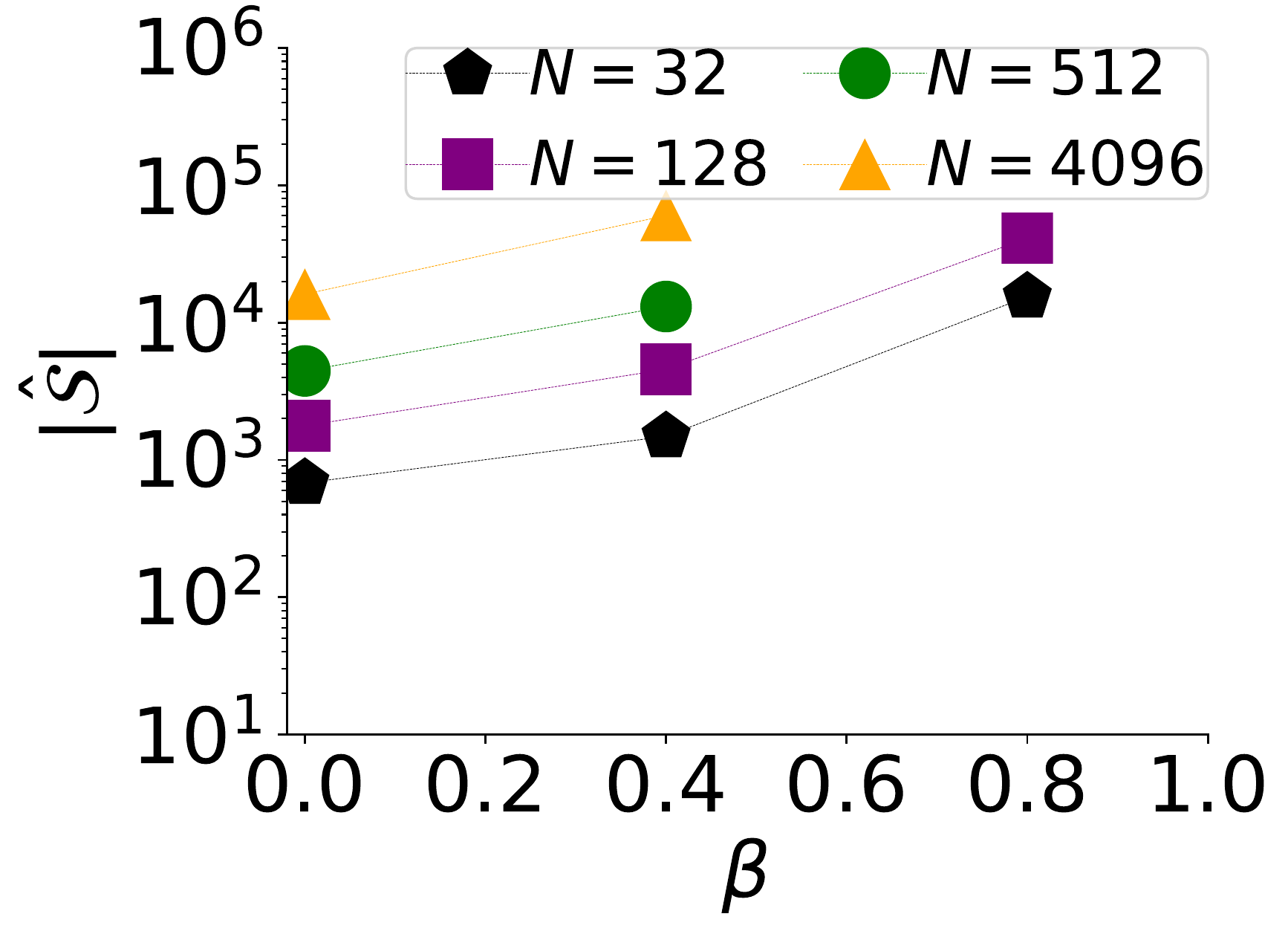}}  
  \hrule   \vspace{0.2cm}
  \hrule
     \text{In  (\textbf{a}-\textbf{h}), $K=256$ and $\pathwaytime=16 \text{ ns}$ are fixed.}   \\ 
   
\subfloat[]{\label{figure-changingNandBeta9} \includegraphics[width=0.245\textwidth]{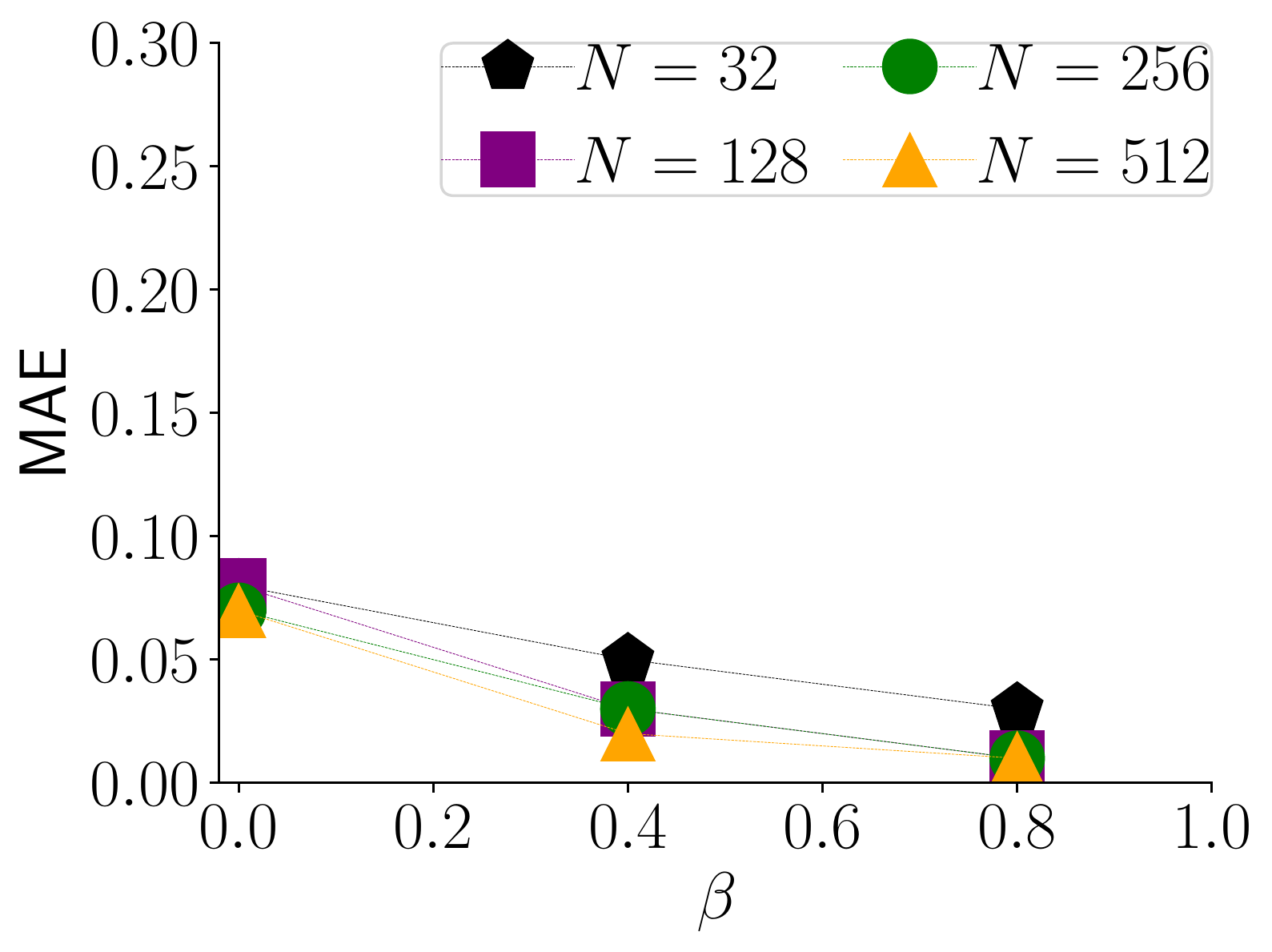}}
 \subfloat[]{\label{figure-changingNandBeta10} \includegraphics[width=0.245\textwidth]{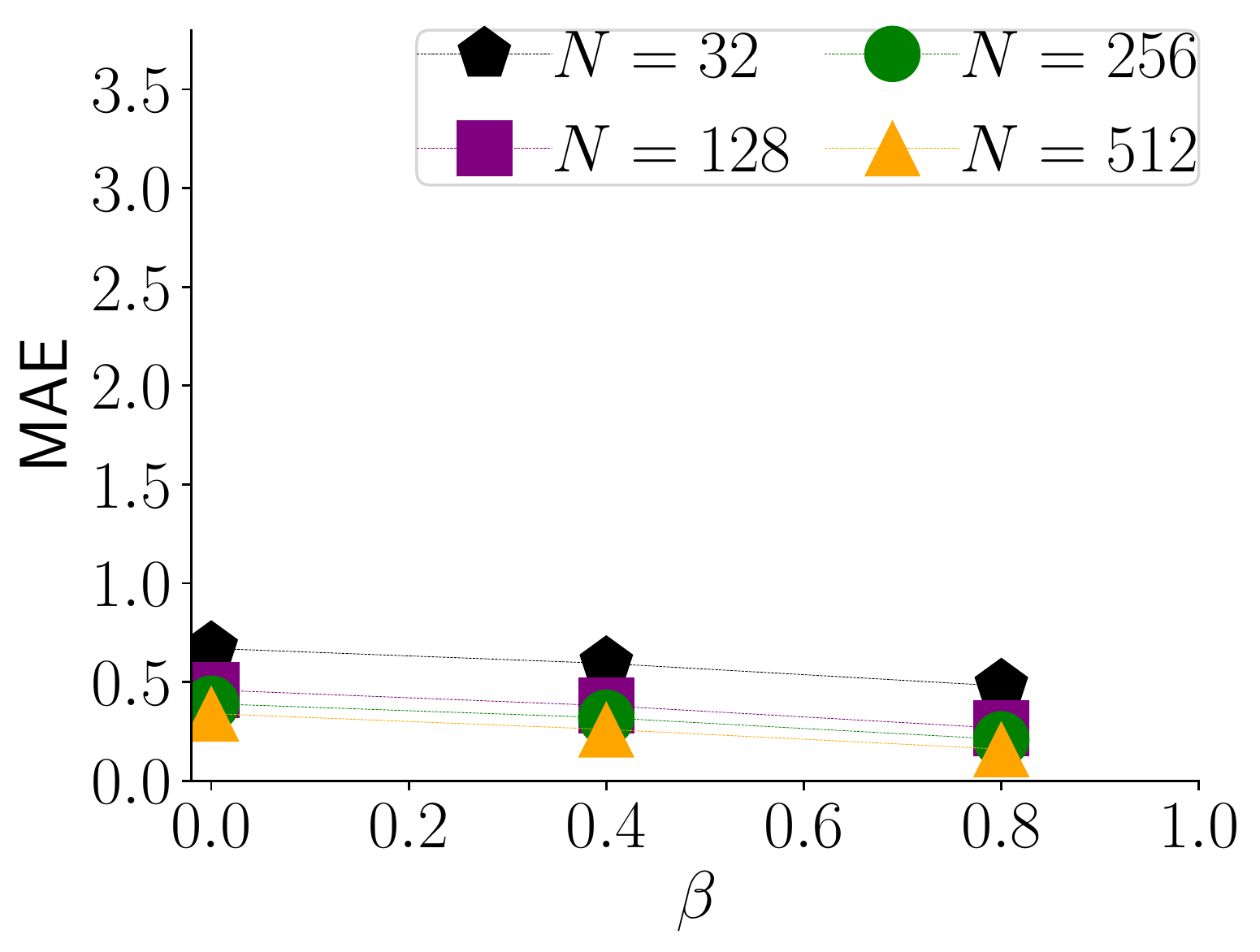}}
\subfloat[]{\label{figure-changingNandBeta11} \includegraphics[width=0.245\textwidth]{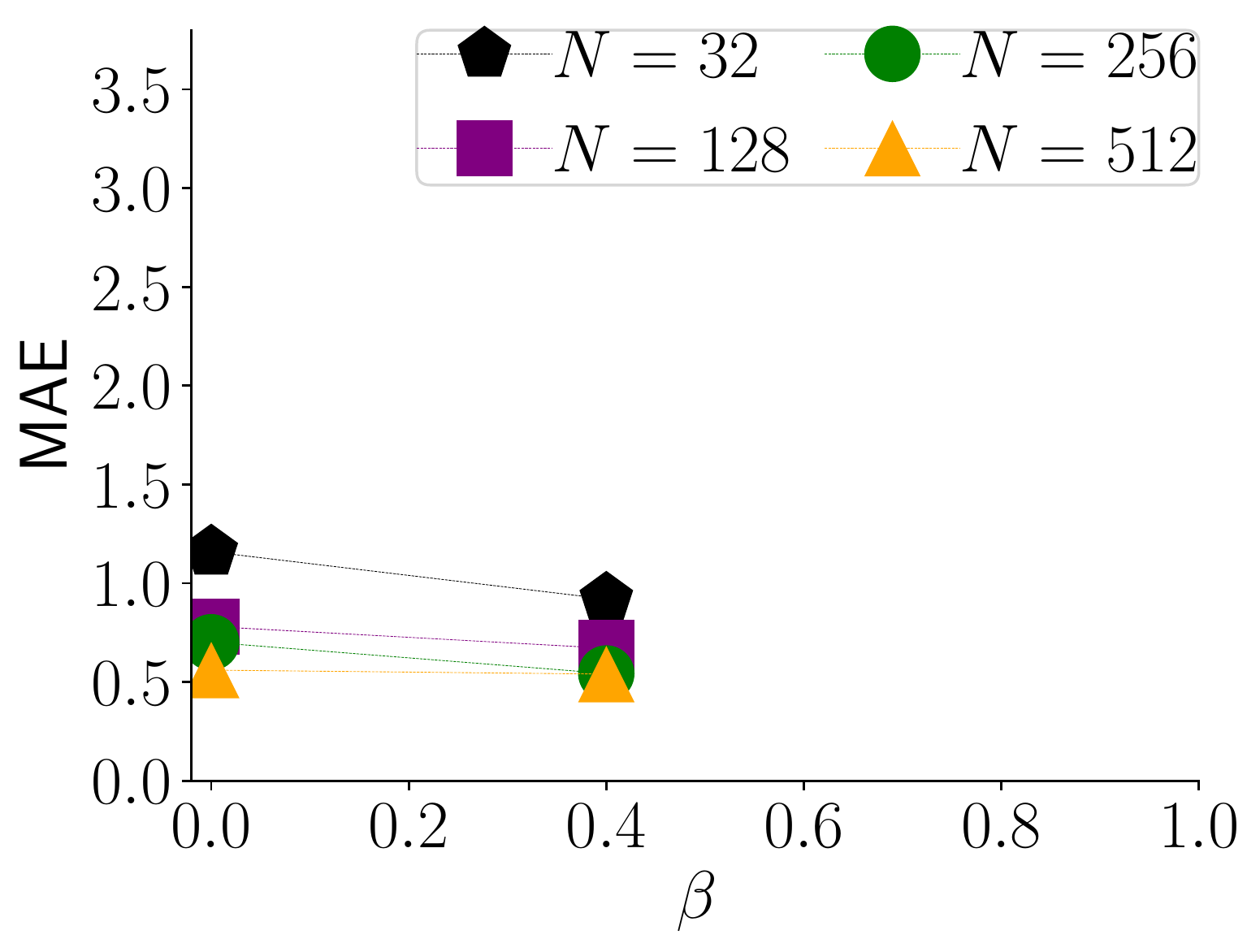}}
  \subfloat[]{\label{figure-changingNandBeta12} \includegraphics[width=0.245\textwidth]{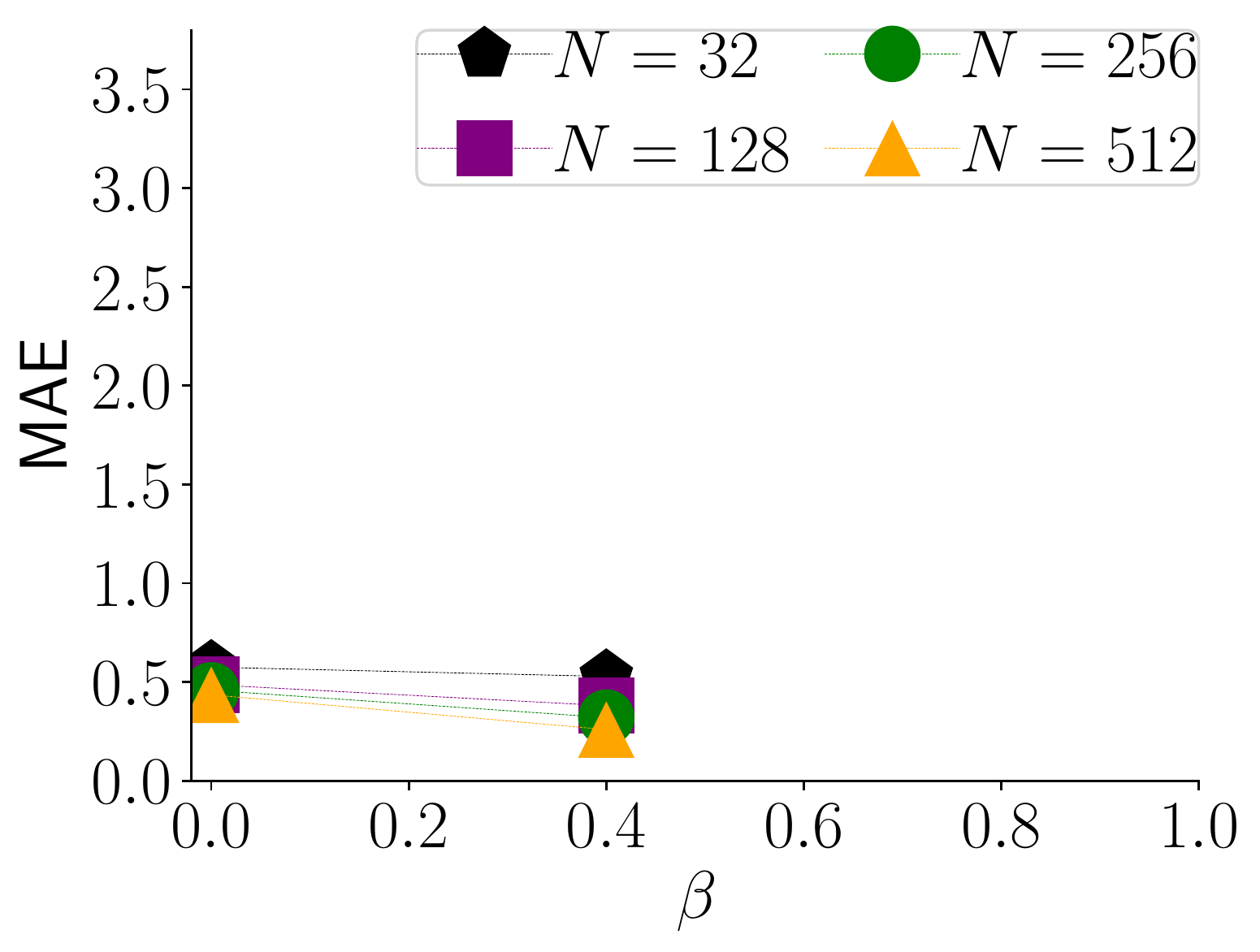}}
     \\ 
    
 \subfloat[]{\label{figure-changingNandBeta13}
 \includegraphics[width=0.245\textwidth]{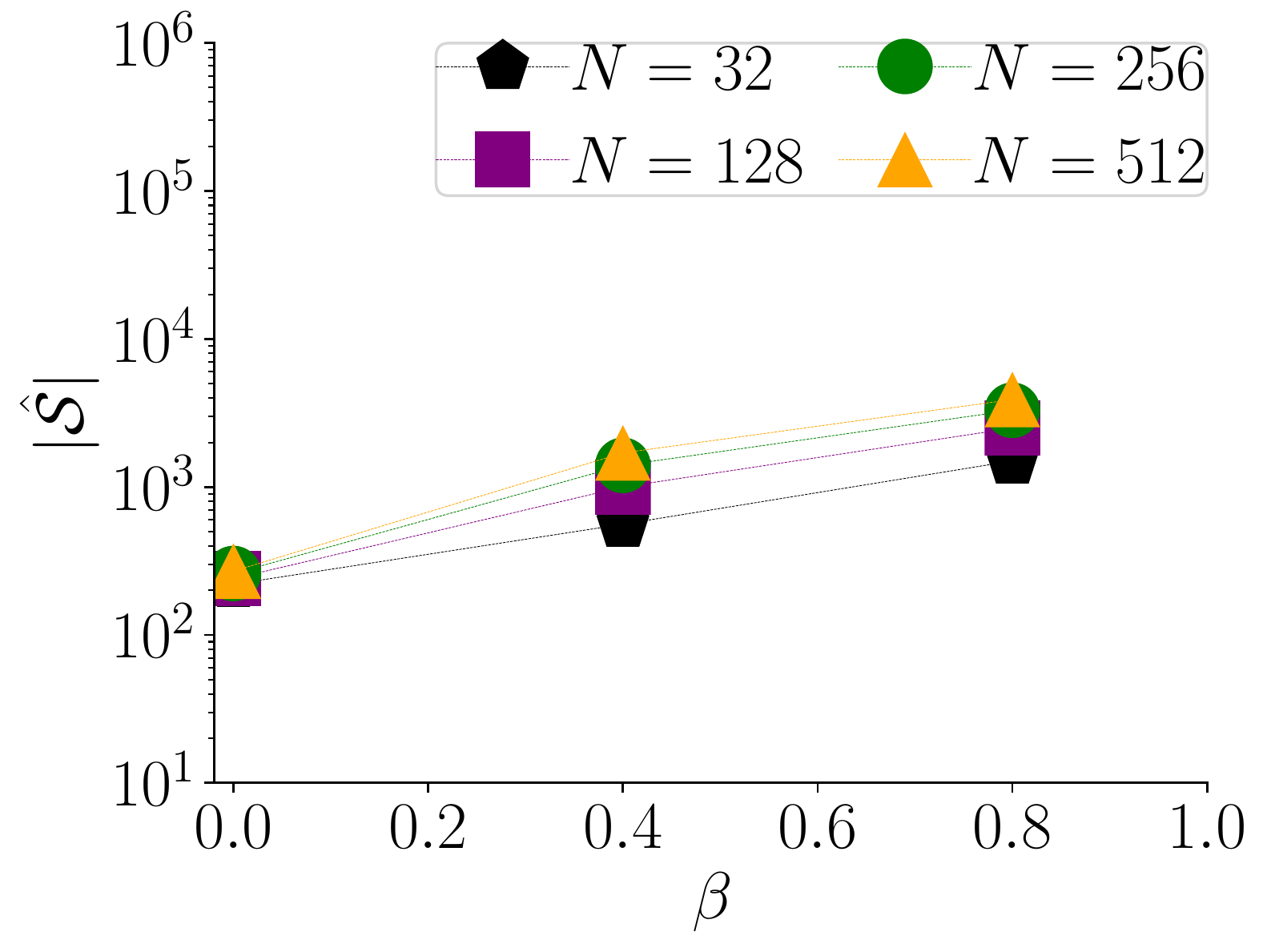}}
 \subfloat[]{\label{figure-changingNandBeta14} \includegraphics[width=0.245\textwidth]{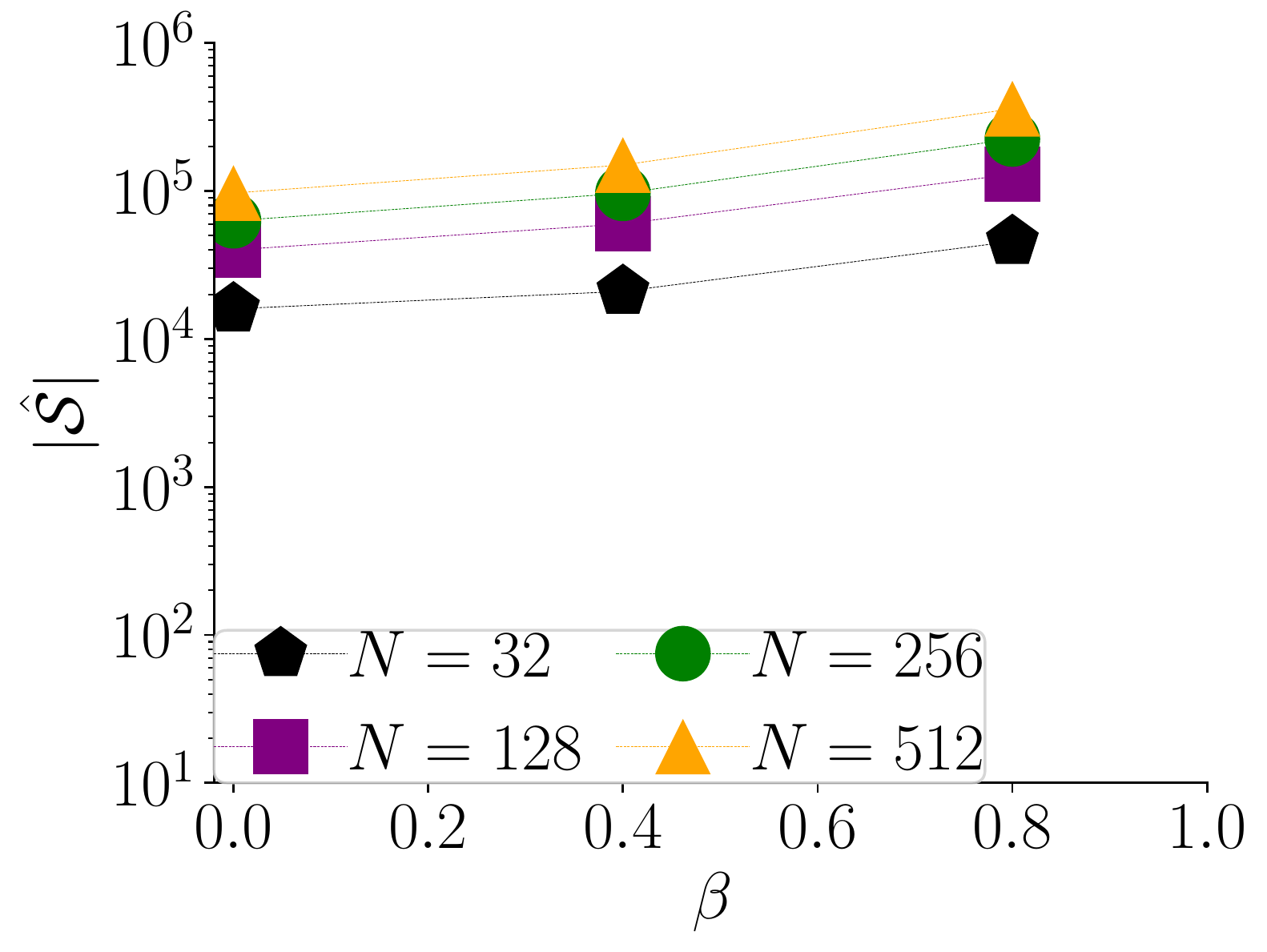}}
\subfloat[]{\label{figure-changingNandBeta15} \includegraphics[width=0.245\textwidth]{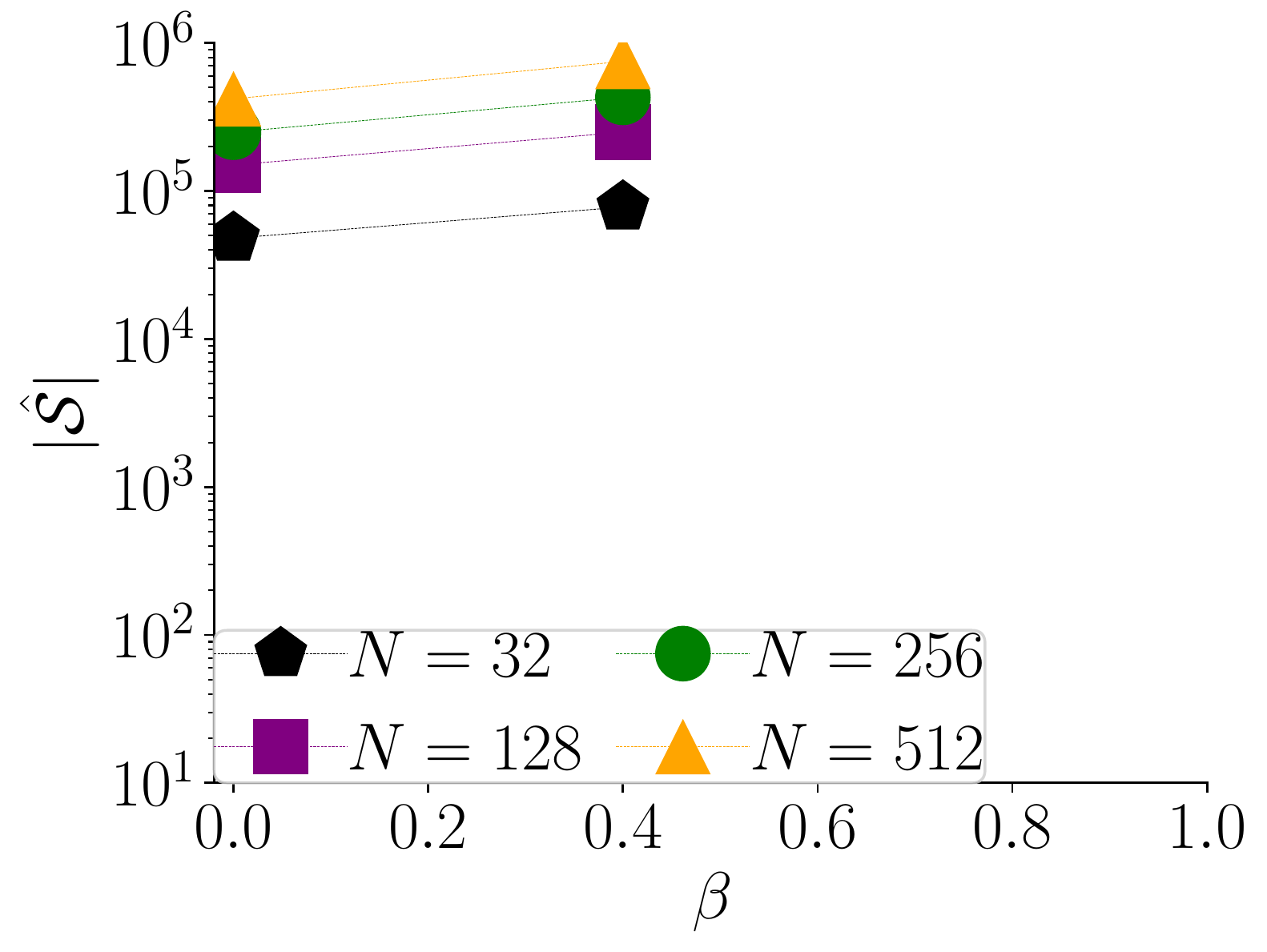}}
  \subfloat[]{ \label{figure-changingNandBeta16}\includegraphics[width=0.245\textwidth]{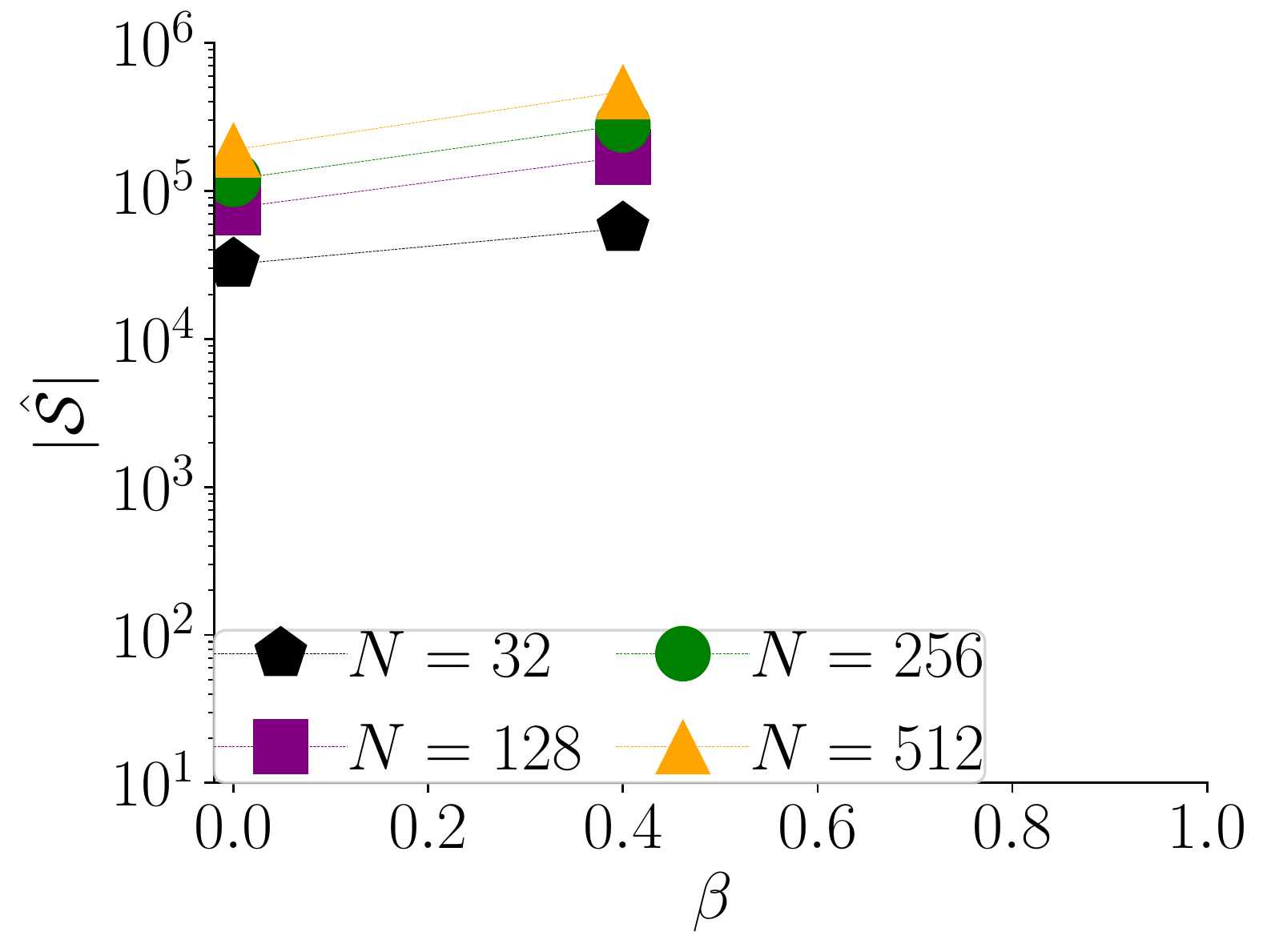}}  
  \hrule \vspace{0.5cm}
\caption[]{ The effect of pathway construction  with different values of $N$ and $\beta$ and fixed values of $K$ and $\pathwaytime$ on the MAE of pathway elaboration with SSA and the  $|\hat{\statespace|}$ of pathway elaboration. In (\textbf{a}-\textbf{h}),  $K=0$ and $\pathwaytime=0 \text{ ns}$ are fixed.  $K=0$ indicates that the states of the pathway are  not elaborated. In  (\textbf{i}-\textbf{p}),  $K=256$ and $\pathwaytime=16 \text{ ns}$ are fixed.   (\textbf{a}), (\textbf{e}), (\textbf{i}), and (\textbf{m}) correspond to  datasets No. 1,2, and 3. (\textbf{b}), (\textbf{f}), (\textbf{j}), and  (\textbf{n}) correspond to   dataset No. 4. (\textbf{c}) (\textbf{g}), (\textbf{k}), and (\textbf{o}) correspond to   dataset No. 5.  (\textbf{d}), (\textbf{h}), (\textbf{l}), and  (\textbf{p}) correspond to dataset No. 6.  For the  missing settings, pathway elaboration  did not finish within two weeks computation time. } 
\label{figure-changingNandBeta}
\end{figure}

\begin{figure}
    \includegraphics[width=0.2\textwidth]{experiments/fig/abstractpics4/unimolecular.pdf}\quad  \quad
 \includegraphics[width=0.2\textwidth]{experiments/fig/abstractpics4/helixassociation.pdf} \quad \quad
 \includegraphics[width=0.18\textwidth]{experiments/fig/abstractpics4/helixassociation-zhang.pdf} \quad\quad\quad
 \includegraphics[width=0.2\textwidth]{experiments/fig/abstractpics4/threewaystranddisplacement.pdf}
\\
\vspace{0.5cm}
  \hrule   
  \ \hspace{3.6cm}\text{In  (\textbf{a}-\textbf{h}), $N=128$ and $\beta = 0.0$ are fixed.} \hspace{4.97cm}  
  \ \\ \ \subfloat[]{\label{figure-changingKandKappa1} \includegraphics[width=0.245\textwidth]{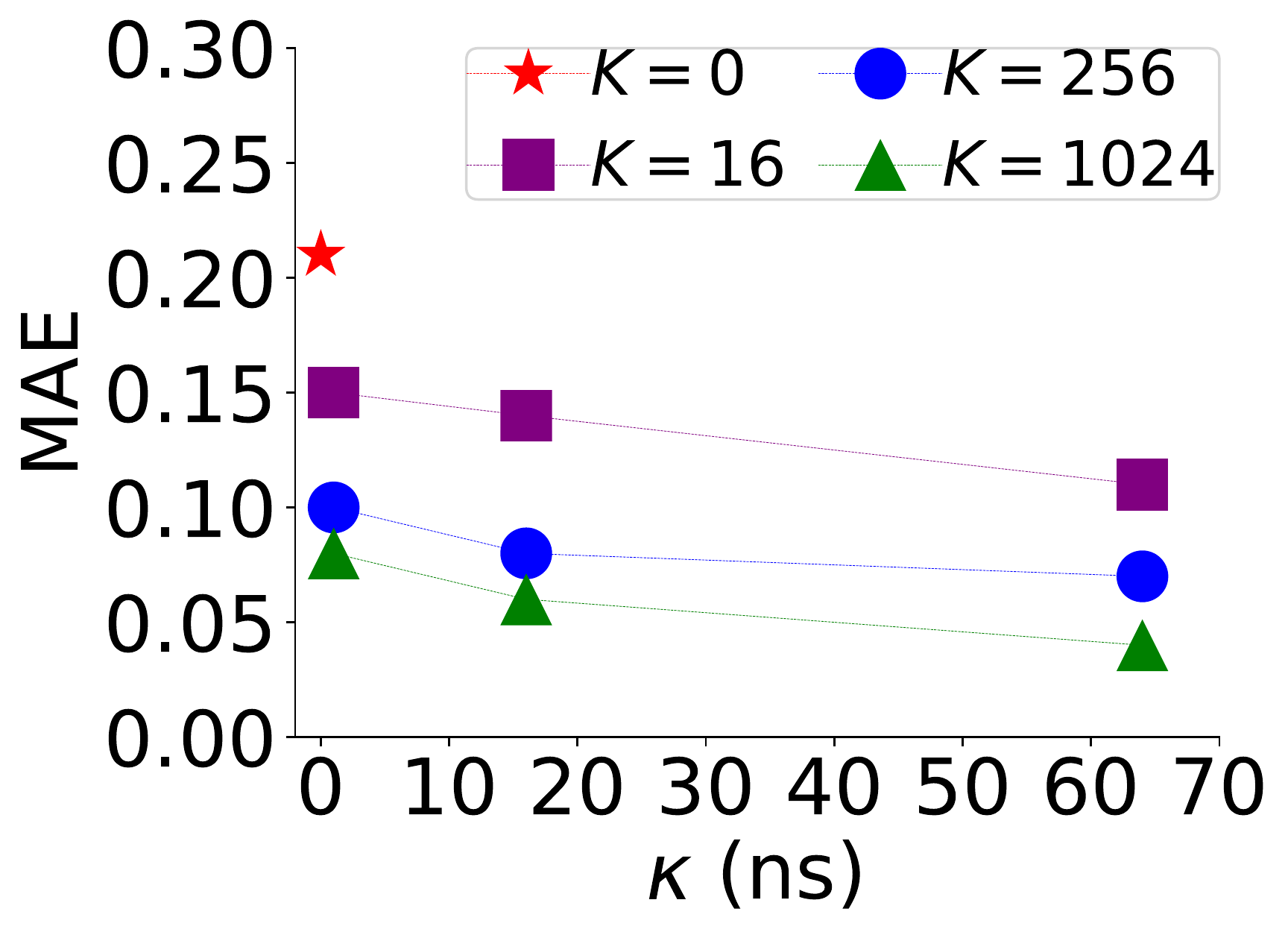}}
   \subfloat[]{ \includegraphics[width=0.245\textwidth]{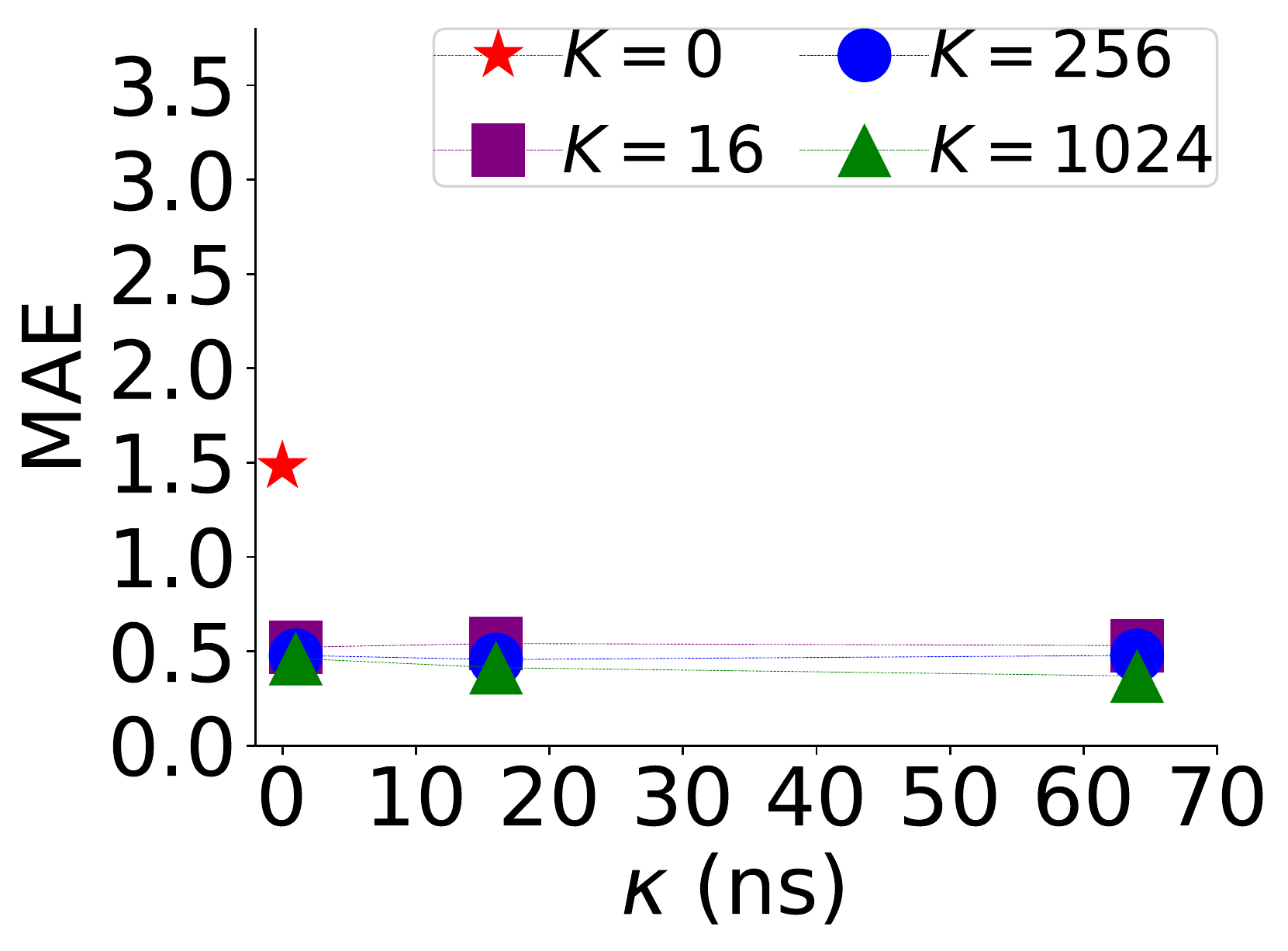}}
   \subfloat[]{ \label{figure-changingKandKappa3}\includegraphics[width=0.245\textwidth]{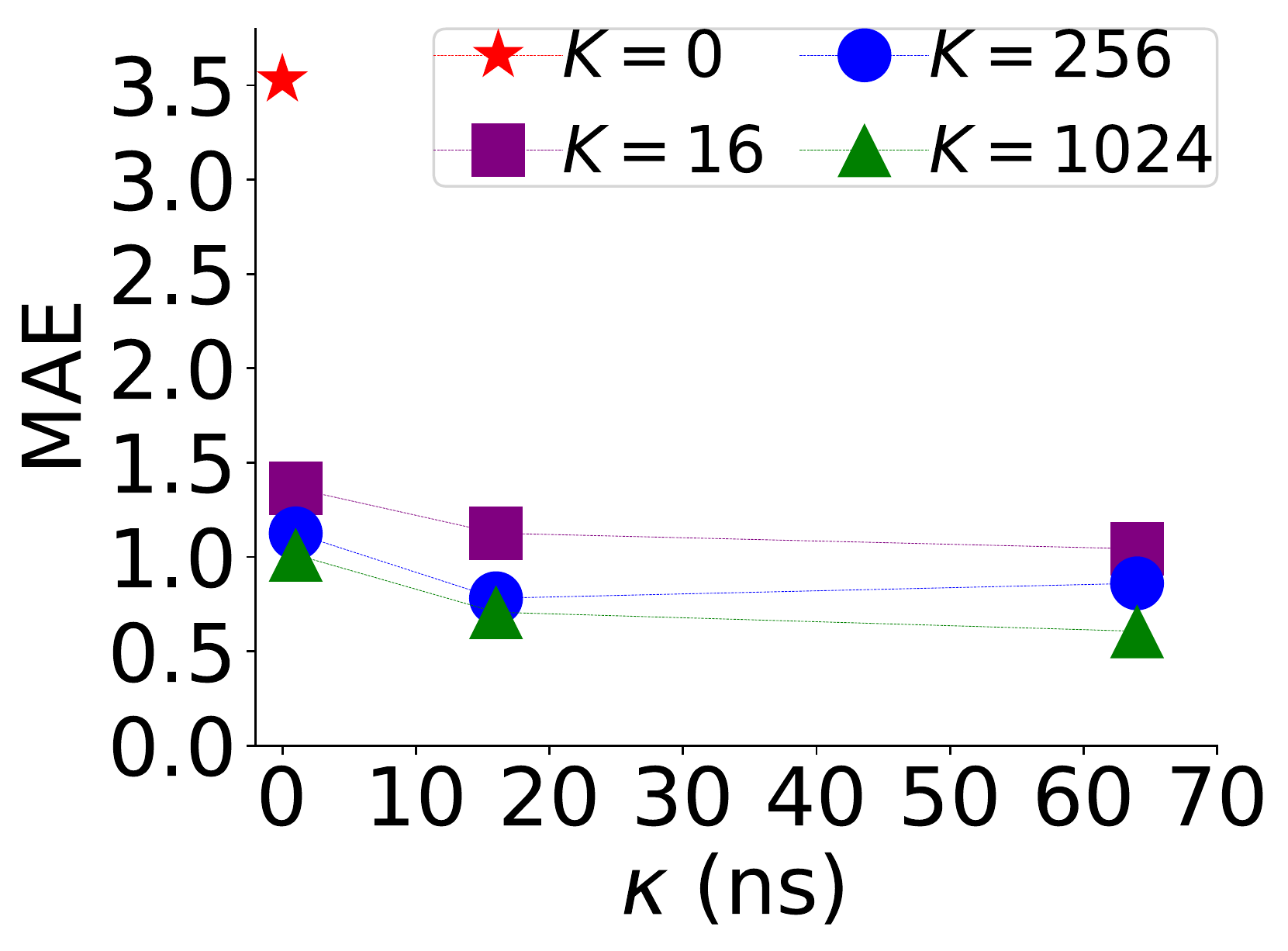}}
   \subfloat[]{ \label{figure-changingKandKappa4} \includegraphics[width=0.245\textwidth]{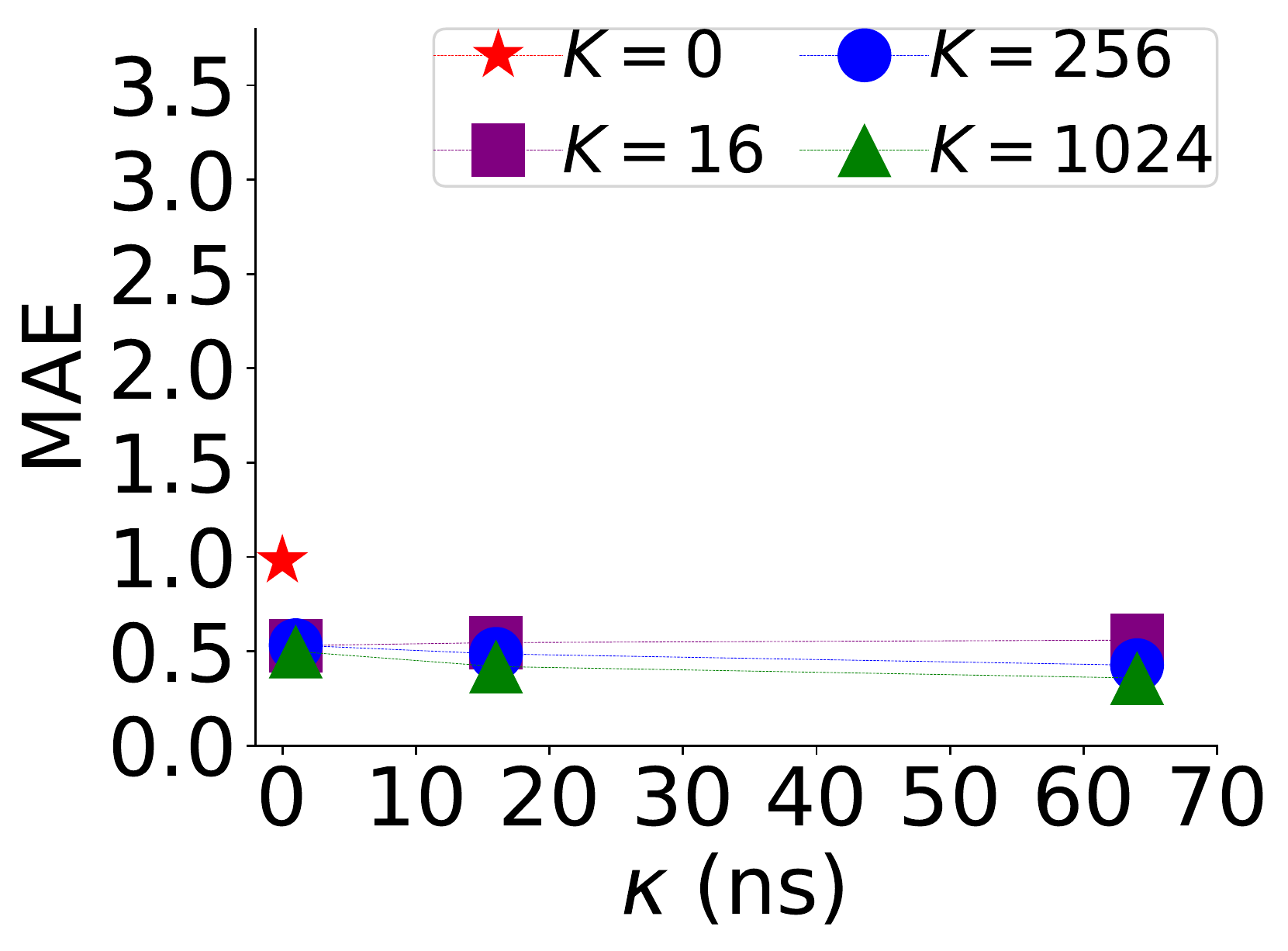}}\  \\ 
   \ 
   \subfloat[]{ \includegraphics[width=0.245\textwidth]{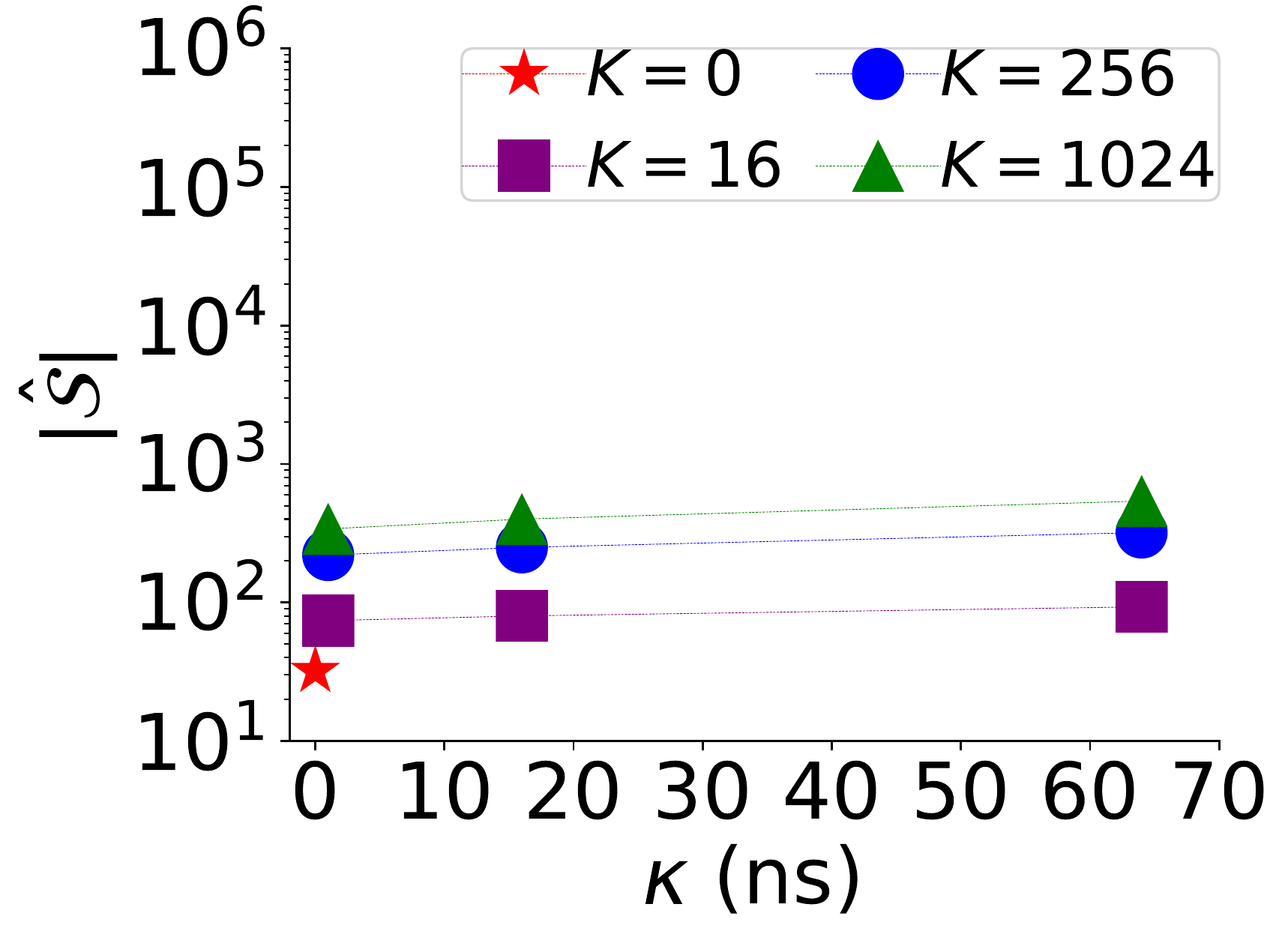}}
   \subfloat[]{ \includegraphics[width=0.245\textwidth]{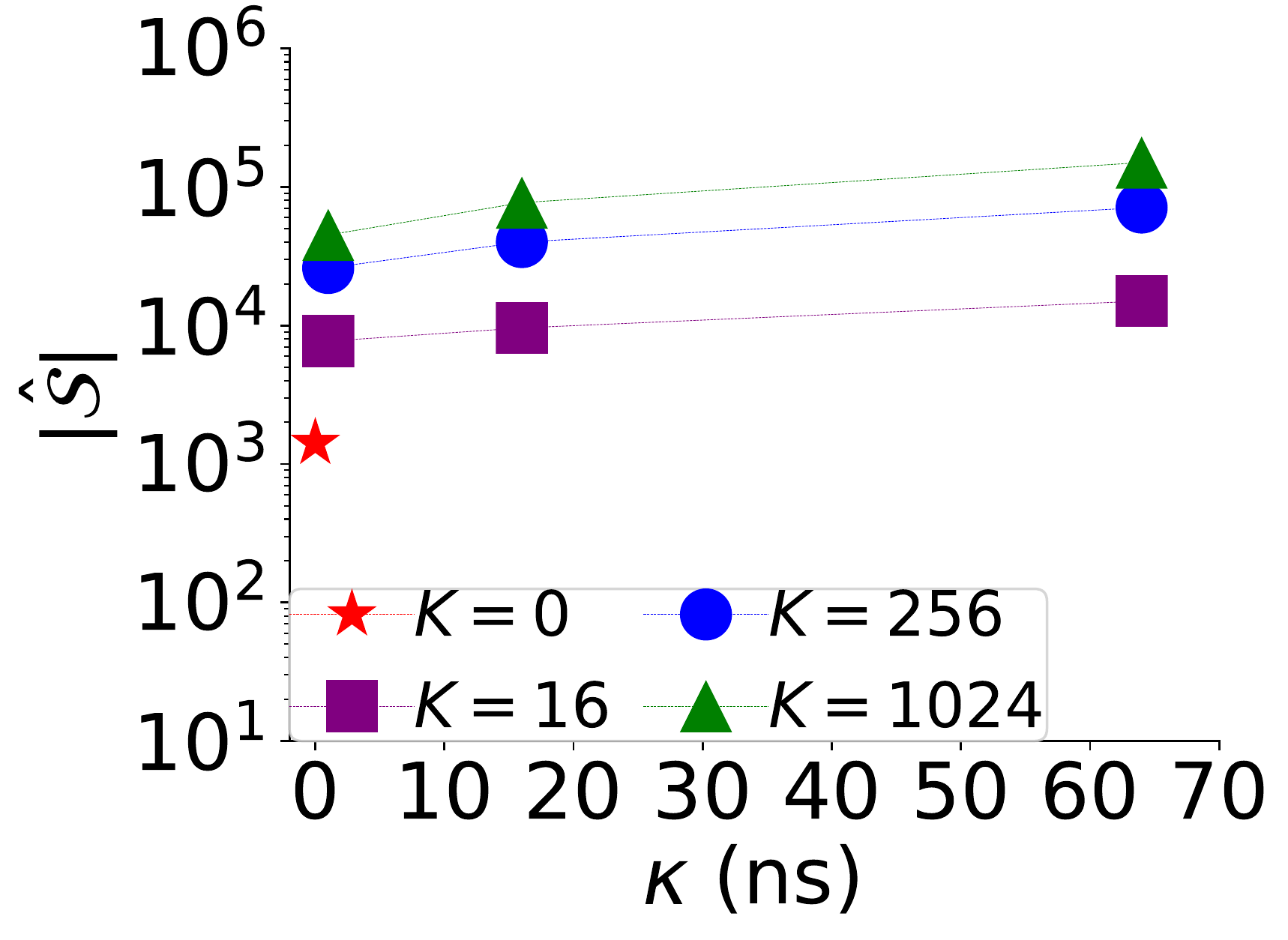}}
   \subfloat[]{ \includegraphics[width=0.245\textwidth]{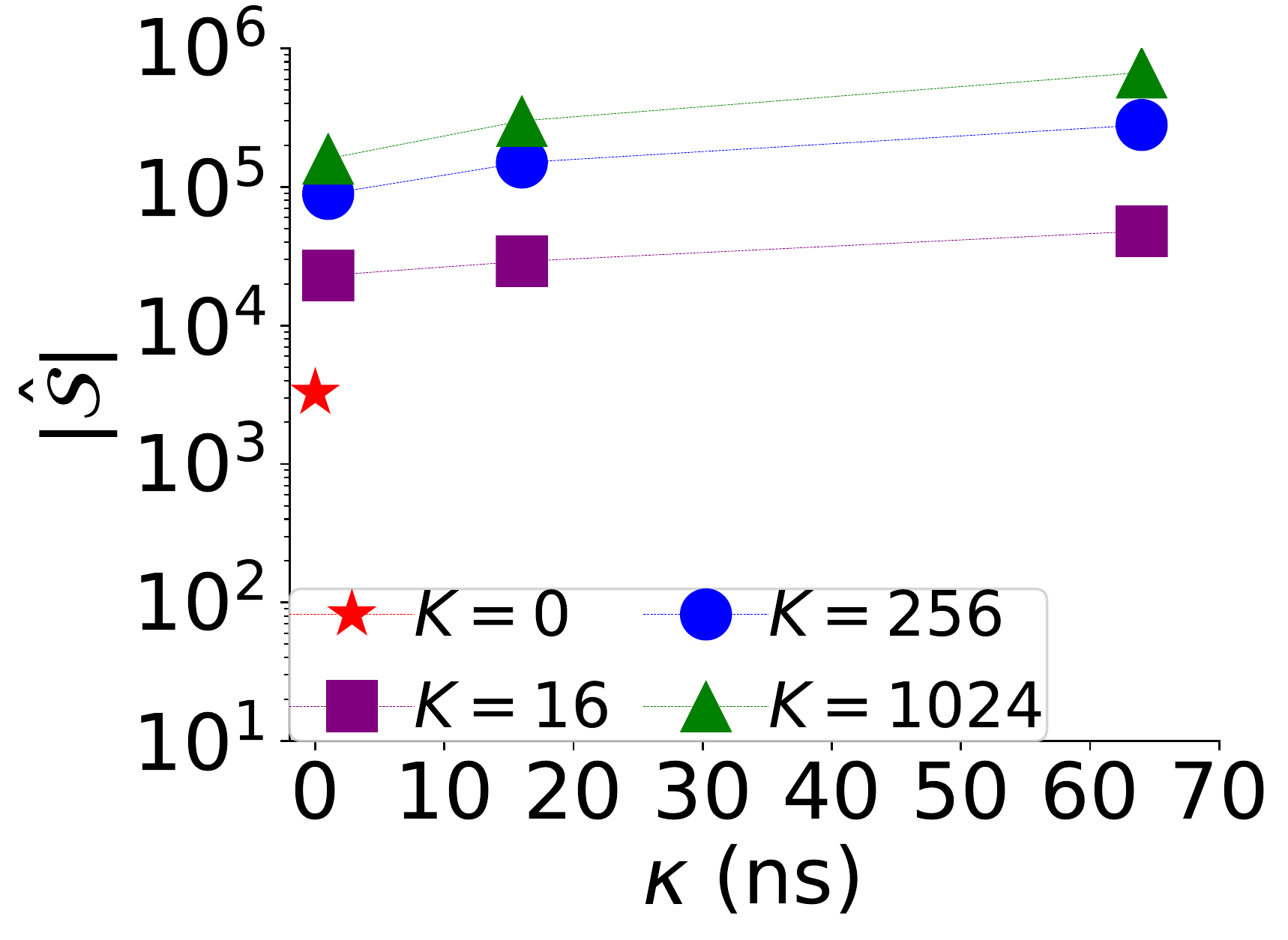}}
   \subfloat[]{\label{figure-changingKandKappa8} \includegraphics[width=0.245\textwidth]{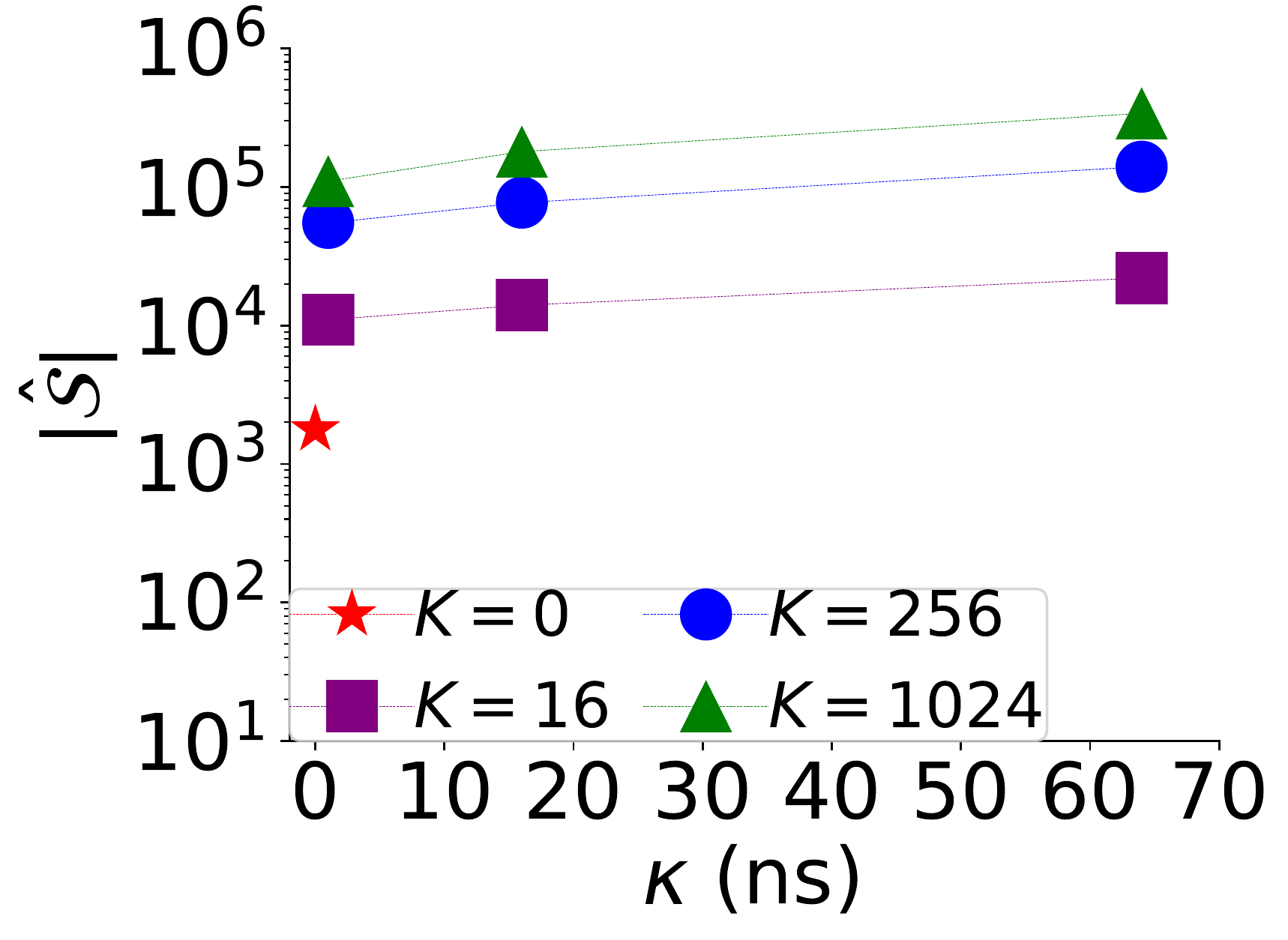}}\  
  \hrule
   \vspace{0.2cm}
  \hrule
  \  \text{In  (\textbf{i}-\textbf{p}), $N=128$ and $\beta = 0.6$ are fixed.} 
   \ \\ \ \subfloat[]{\label{figure-changingKandKappa9} \includegraphics[width=0.245\textwidth]{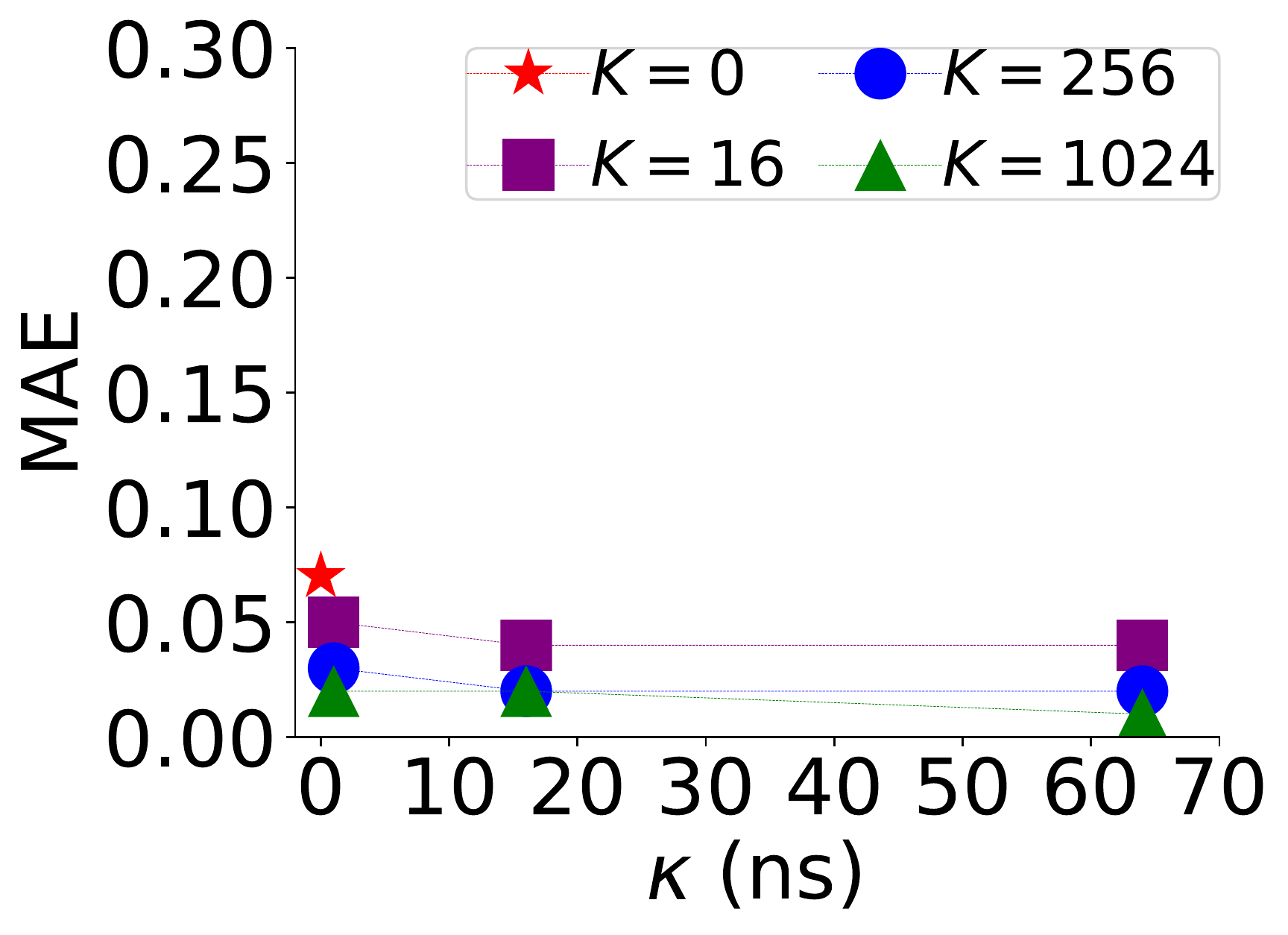}}
   \subfloat[]{ \label{figure-changingKandKappa10} \includegraphics[width=0.245\textwidth]{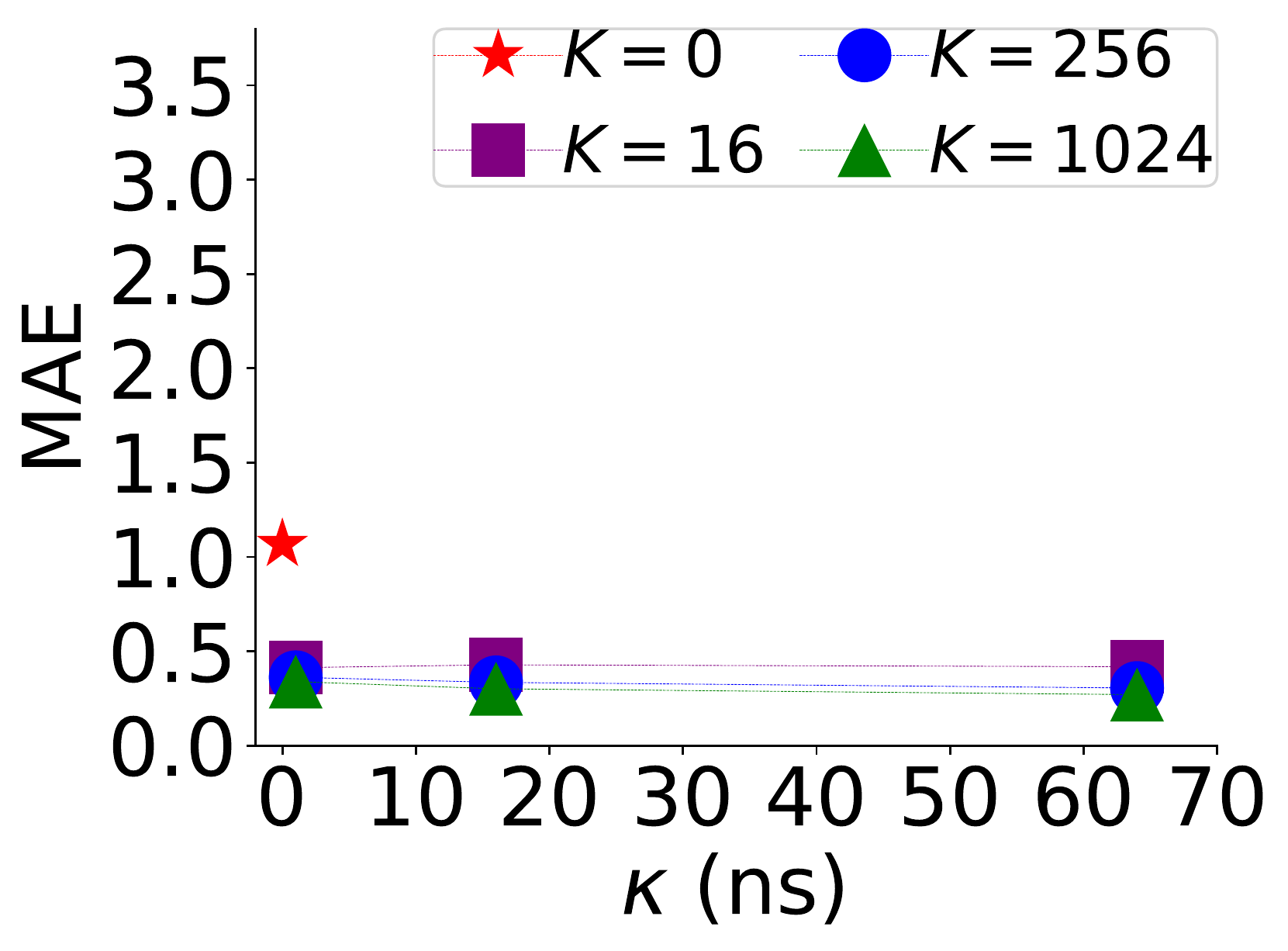}}
   \subfloat[]{ \label{figure-changingKandKappa11}\includegraphics[width=0.245\textwidth]{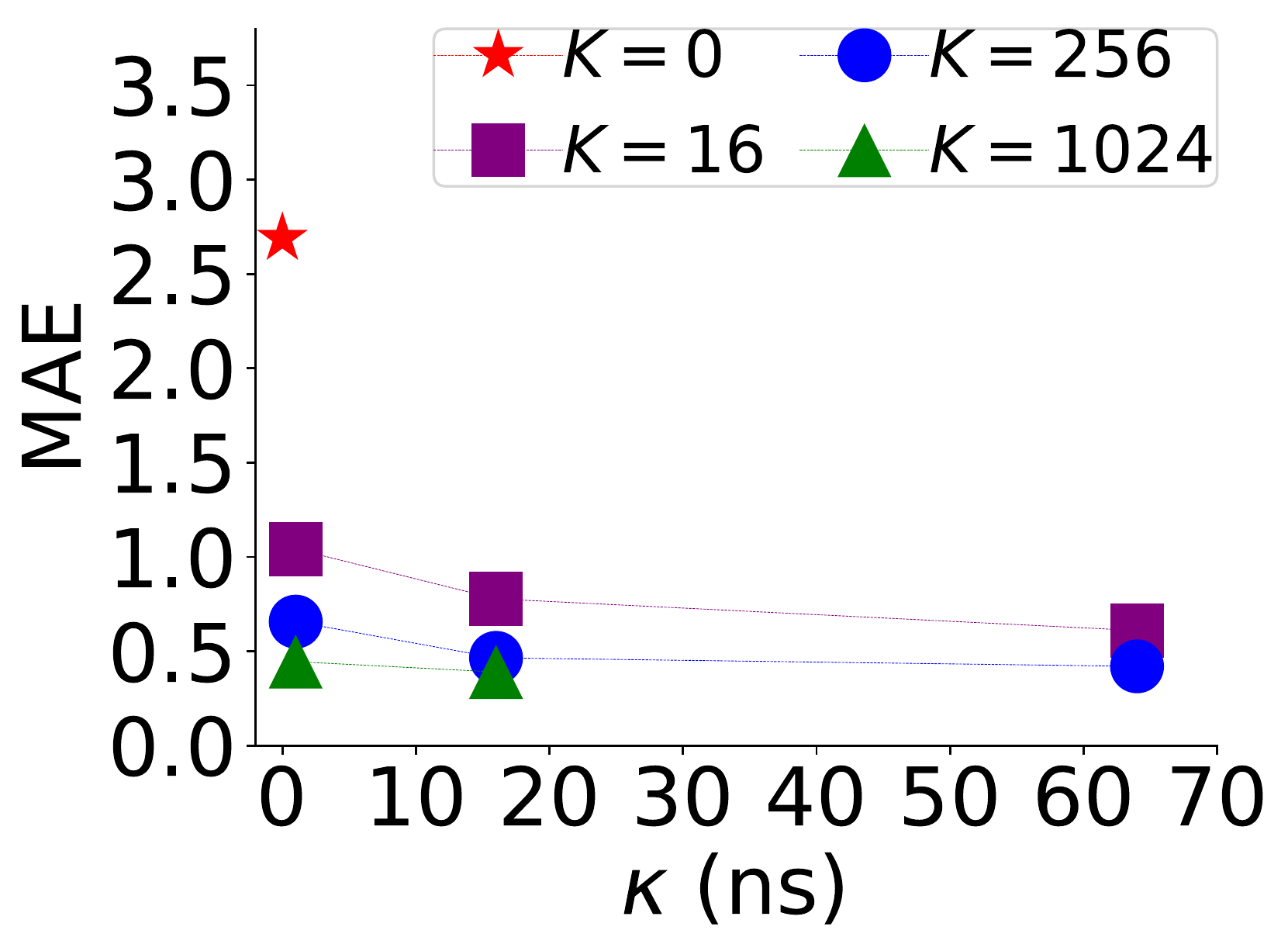}}
   \subfloat[]{\label{figure-changingKandKappa12} \includegraphics[width=0.245\textwidth]{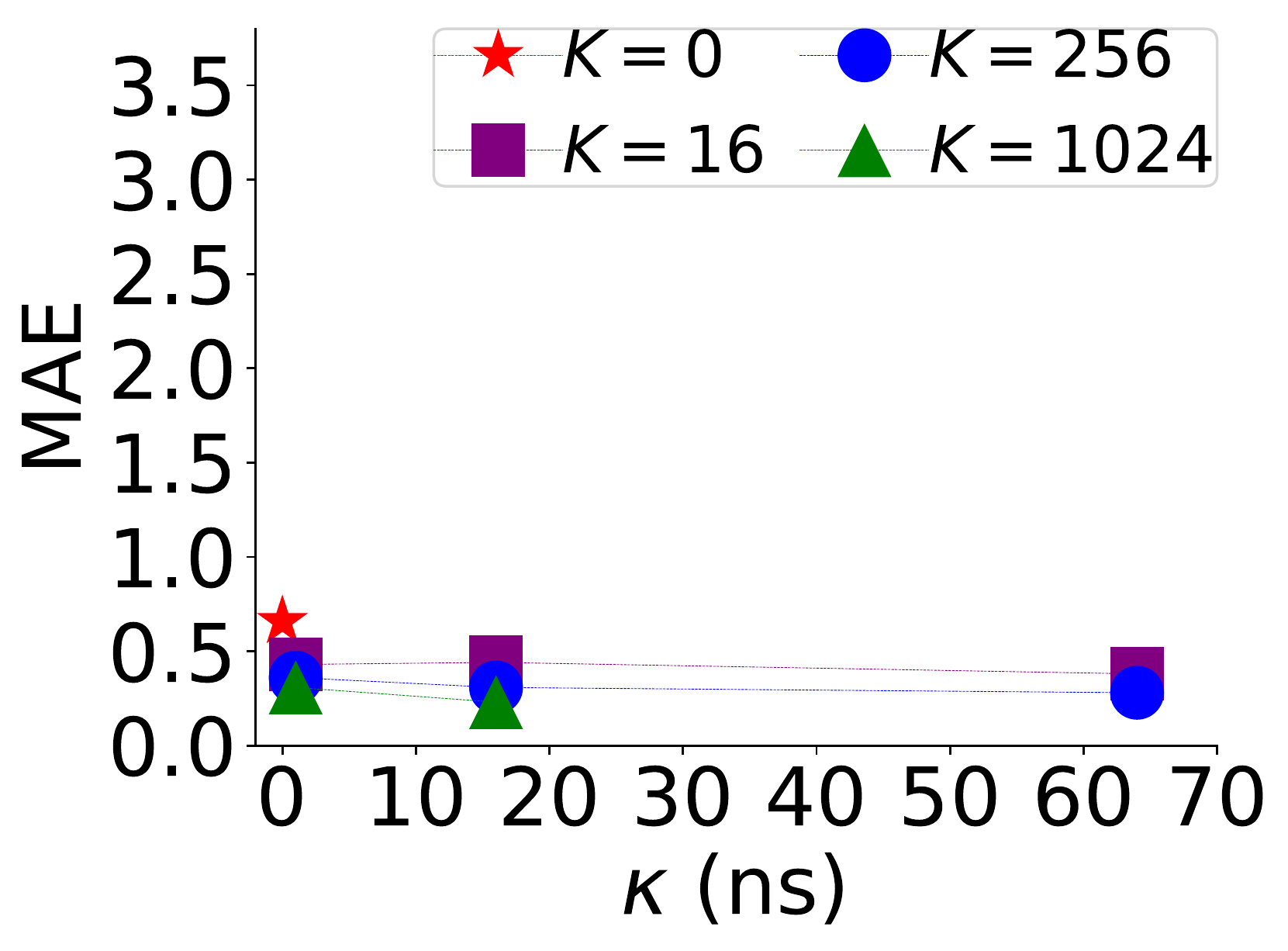}}\  \\ 
   \ 
   \subfloat[]{ \label{figure-changingKandKappa13}\includegraphics[width=0.245\textwidth]{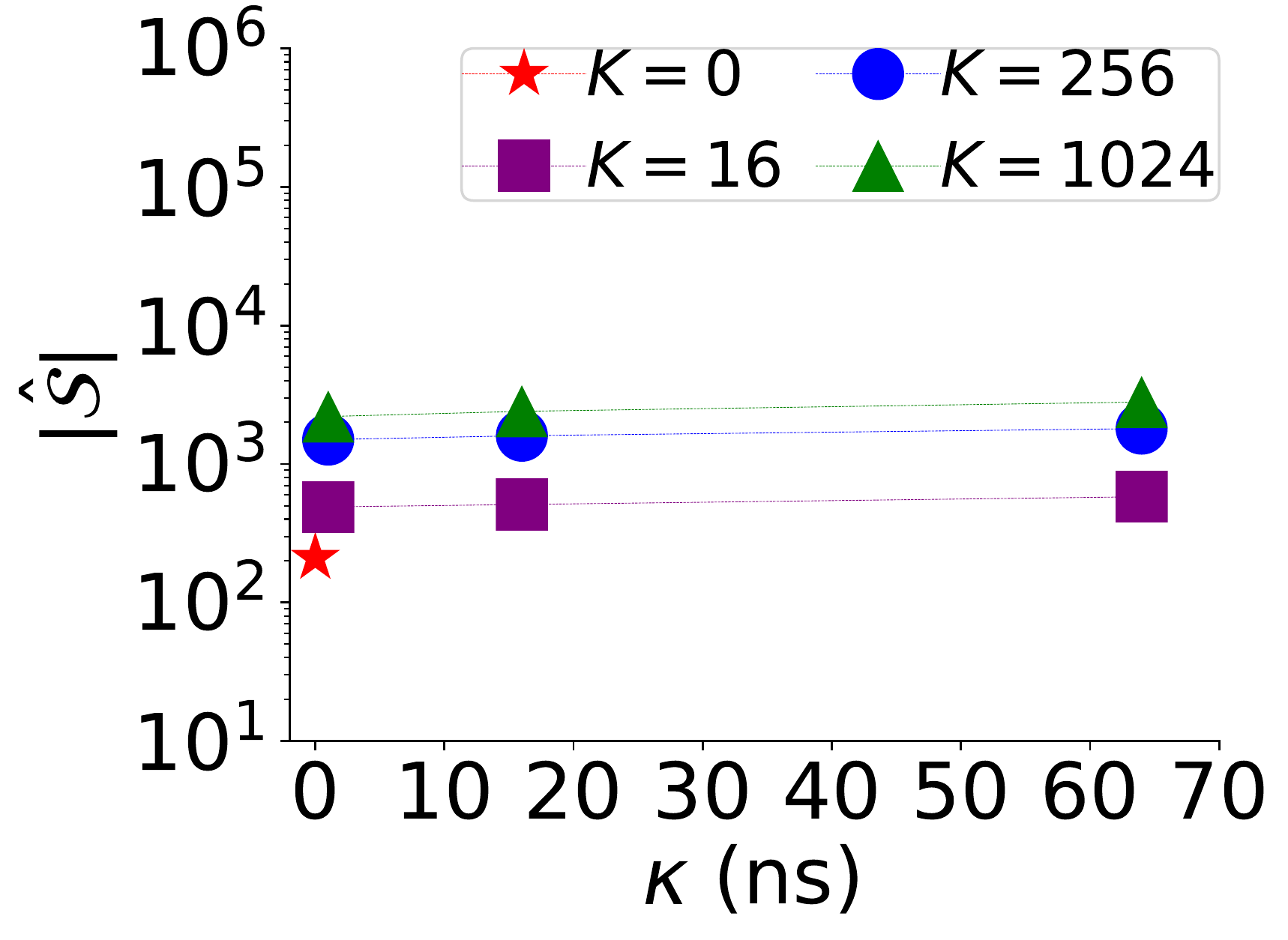}}
   \subfloat[]{ \label{figure-changingKandKappa14}\includegraphics[width=0.245\textwidth]{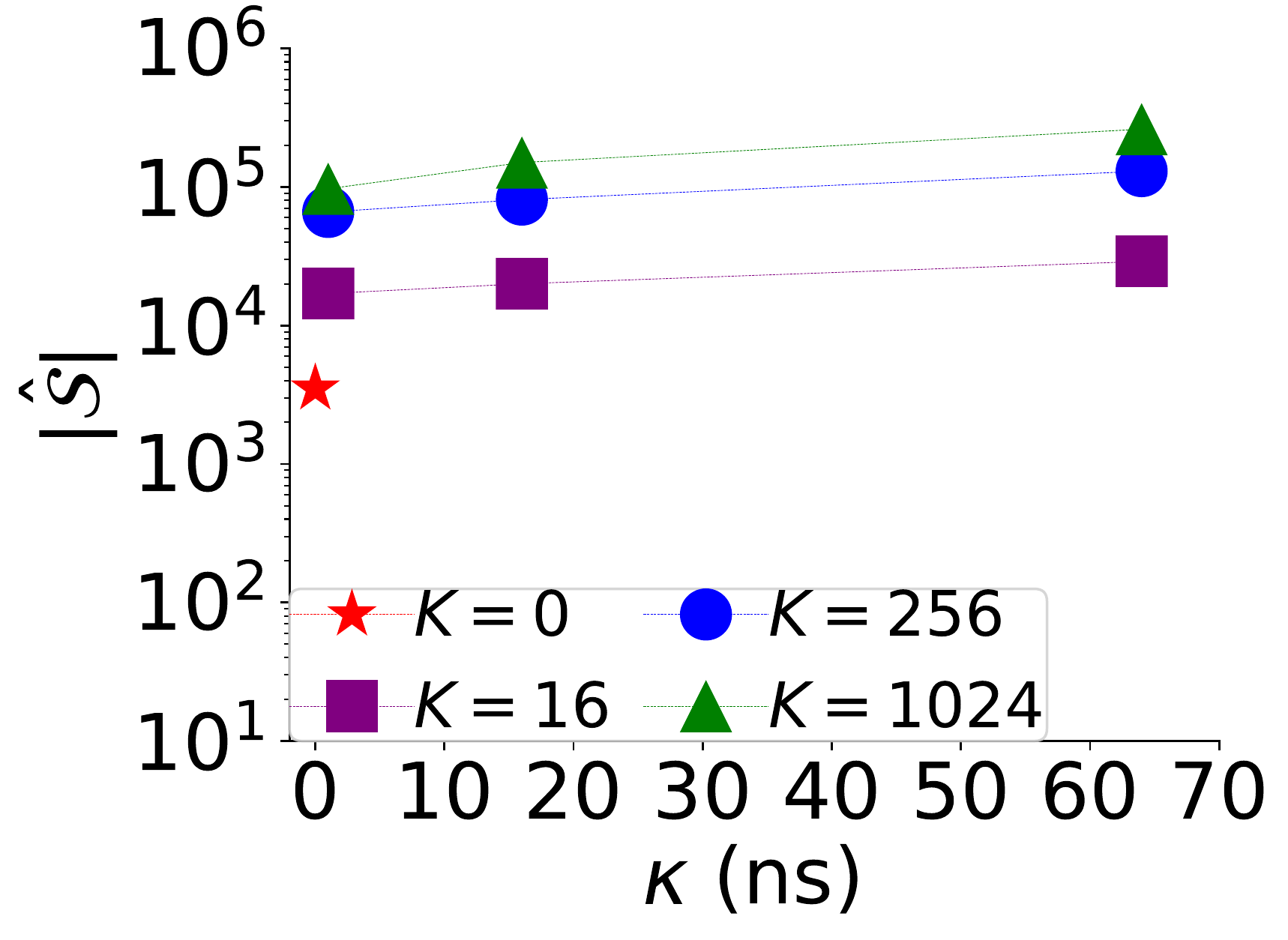}}
   \subfloat[]{ \label{figure-changingKandKappa15}\includegraphics[width=0.245\textwidth]{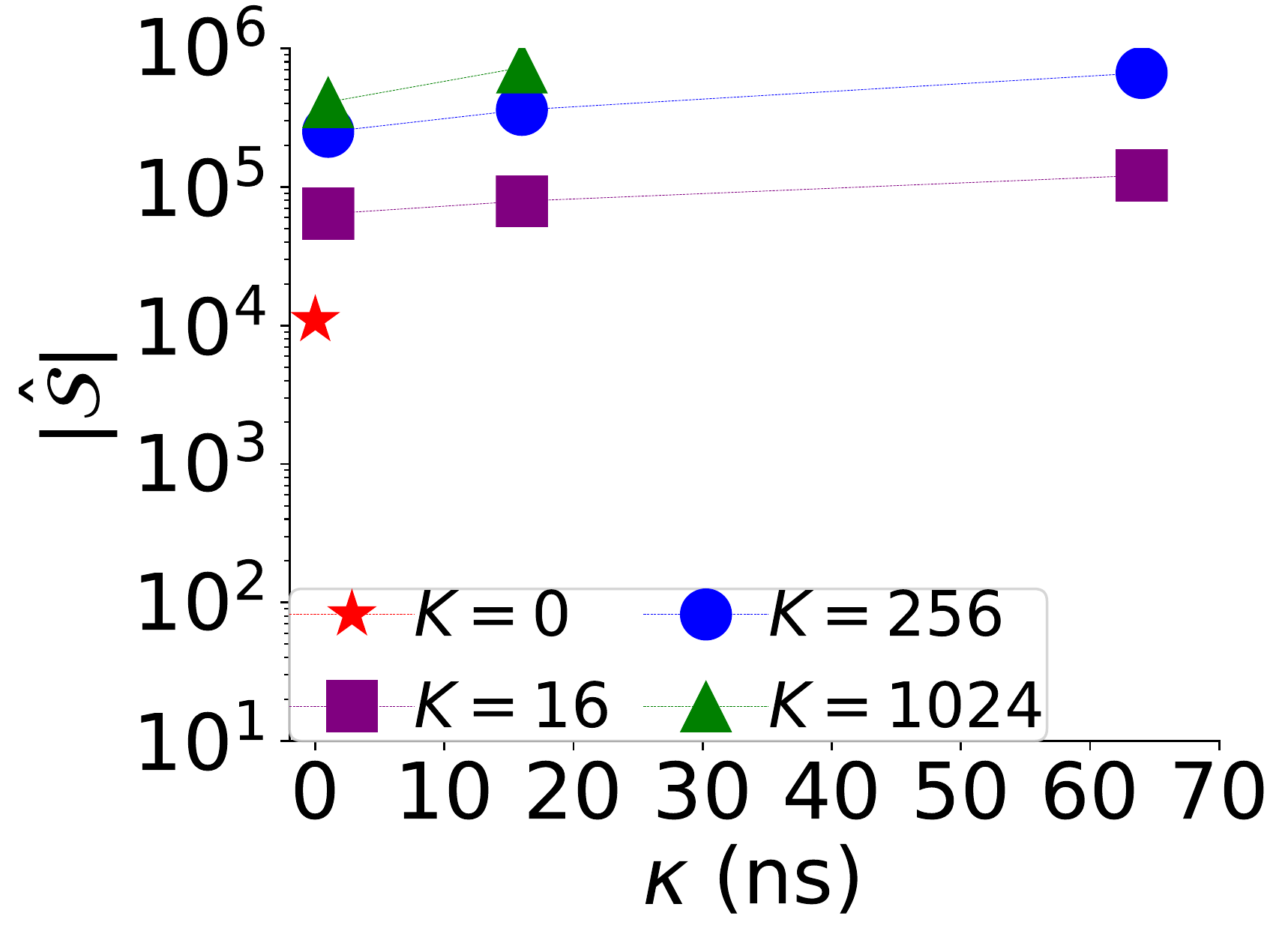}}
   \subfloat[]{\label{figure-changingKandKappa16} \includegraphics[width=0.245\textwidth]{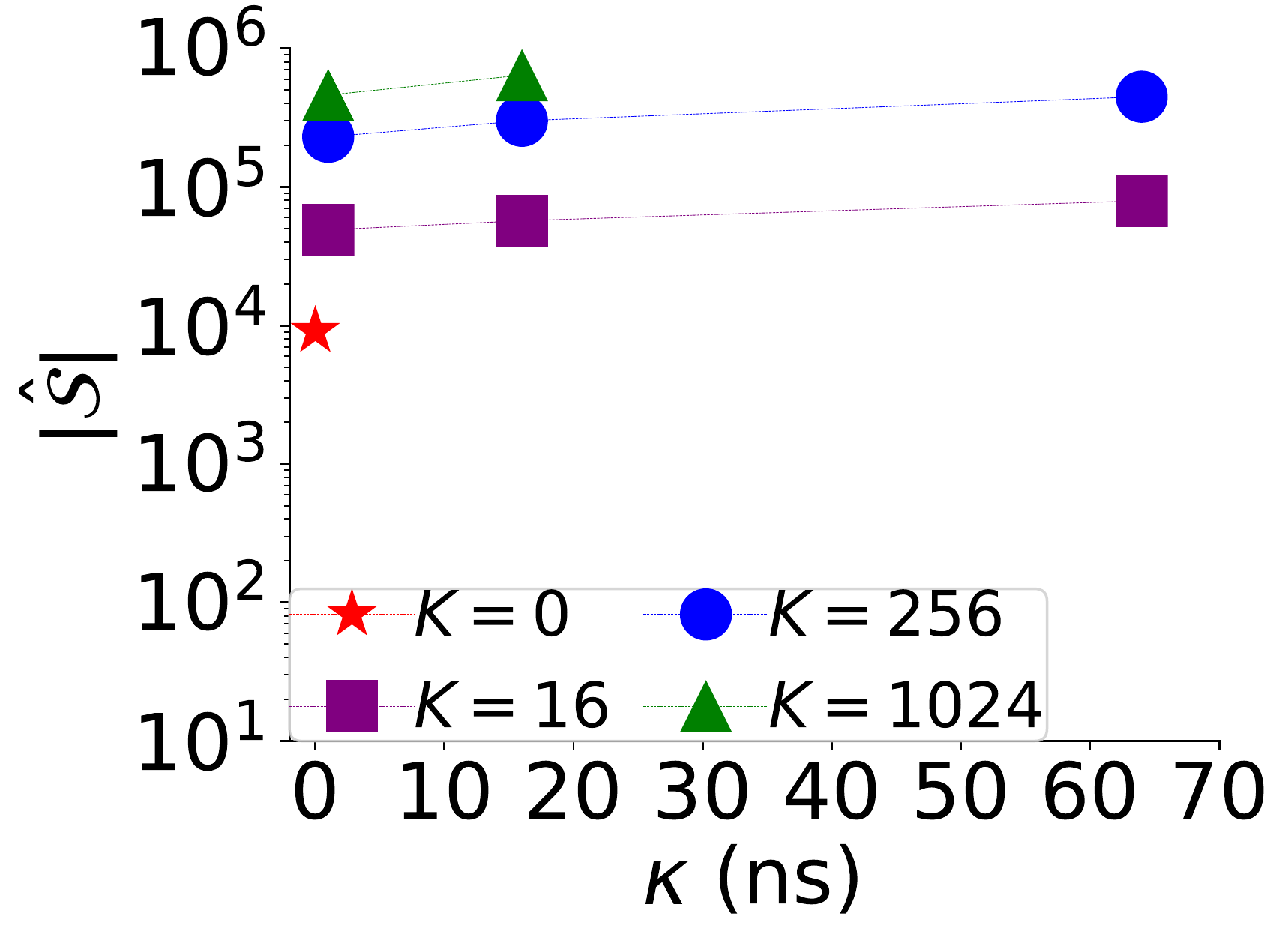}}\  
  \hrule
   
   \vspace{0.5cm}
\caption[]{The effect of state elaboration, with different values of $K$ and $\pathwaytime$  and fixed values of $N$ and $\beta$  on the MAE of pathway elaboration with SSA and the  $|\hat{\statespace|}$ of pathway elaboration.  $K=0$ indicates that the states of the pathway are  not elaborated. In (\textbf{a}-\textbf{h}),  $N=128$ and $\beta=0.0$  are fixed. In  (\textbf{i}-\textbf{p}),  $N=128$ and $\beta=0.6$ are fixed.   (\textbf{a}), (\textbf{e}), (\textbf{i}), and (\textbf{m}) correspond to  datasets No. 1,2, and 3. (\textbf{b}), (\textbf{f}), (\textbf{j}), and  (\textbf{n}) correspond to   dataset No. 4. (\textbf{c}) (\textbf{g}), (\textbf{k}), and (\textbf{o}) correspond to   dataset No. 5.  (\textbf{d}), (\textbf{h}), (\textbf{l}), and  (\textbf{p}) correspond to dataset No. 6. For the  missing settings, pathway elaboration  did not finish within two weeks computation time. }
\label{figure-changingKandKappa}
\end{figure}

\end{appendix}


\end{document}